\documentclass[final]{econsocart}

\usepackage[yyyymmdd,24hr]{datetime}
\usepackage{booktabs}
\usepackage{mathtools}
\usepackage{mathrsfs}
\usepackage{xurl}
\startlocaldefs
\usepackage{threeparttable}
\theoremstyle{plain}  

\newtheorem{theorem}{Theorem}

\newtheorem{proposition}{Proposition}
\newtheorem{hypo}{Hypothesis}

\newtheorem{lemma}{Lemma}
\newtheorem{result}{Result}

\theoremstyle{remark}
\newtheorem{definition}{Definition}
\newtheorem{example}{Example}

\usepackage{tikz}
\usetikzlibrary{shapes}

\usepackage{etoolbox}

\renewcommand{\appendixname}{Supplementary Appendix}

\newcommand{\da}{\mathtt{DA}} 
\newcommand{\ia}{\mathtt{IA}} 
\newcommand{\mia}{\mathtt{MIA}}

\usepackage{enumitem}
\usepackage{amstext} % for \text macro
\usepackage{array}   % for \newcolumntype macro
\newcolumntype{C}{>{$}c<{$}} % math-mode version of "l" column type
\usepackage{multirow}
\usepackage{float}

\endlocaldefs

\usepackage{titlesec}
\titlespacing*{\section}      {0pt}{1.2ex}{0.8ex}
\titlespacing*{\subsection}   {0pt}{1.2ex plus .3ex minus .2ex}{0.6ex}
\titlespacing*{\subsubsection}{0pt}{1.0ex plus .3ex minus .2ex}{0.5ex}

\makeatletter
\def\quote@topsep{2\p@}%      % space above AND below block quotes (was \smallskipamount)
\def\quotation@topsep{2\p@}%
\makeatother

\RequirePackage[colorlinks,citecolor=blue,linkcolor=blue,urlcolor=blue,pagebackref]{hyperref}

\makeatletter
\g@addto@macro\@afterheading{\vspace{-\parskip}}
\makeatother

\usepackage{setspace}
\begin{document}

\begin{frontmatter}

\title{School Choice with Appeals}
\runtitle{School Choice with Appeals}

\begin{aug}
\author[id=au1,addressref={add1}]{\fnms{Claudia}~\snm{Cerrone}}
\author[id=au2,addressref={add2}]{\fnms{Yoan}~\snm{Hermstr\"{u}wer}}
\author[id=au3,addressref={add3}]{\fnms{Josu\'e}~\snm{Ortega}}

\address[id=add1]{%
	%\orgdiv{Department of Economics},
	\orgname{City St George's, University of London}}

\address[id=add2]{%
	%\orgdiv{Faculty of Law},
	\orgname{University of Zurich}}

\address[id=add3]{%
	%\orgdiv{Department of Economics},
	\orgname{Queen's University Belfast}}

\end{aug}

\support{Emails: \href{mailto:claudia.cerrone@city.ac.uk}{claudia.cerrone@city.ac.uk}, \href{mailto:yoan.hermstruewer@ius.uzh.ch}{yoan.hermstruewer@ius.uzh.ch}, \href{mailto:j.ortega@qub.ac.uk}{j.ortega@qub.ac.uk}.\\We thank Sophie Bade, Estelle Cantillon, Yan Chen, John Coldron, Arcangelo Dimico, Arda Gitmez, Rustam Hakimov, Flip Klijn, Sander Onderstal, Chris Taylor, Camille Terrier, Bertan Turhan, M. Utku Ünver, and M. Bumin Yenmez for valuable comments.
%	\it \hfill Last edited on {\ddmmyyyydate \today}.
	}

\enlargethispage{3\baselineskip}

\begin{abstract}
	\vspace*{1em}
	
\noindent Roughly fifty thousand families in England appeal their school assignments each year through a centralized and understudied procedure that improves the placement of about one in five claimants. 
This paper shows that appeals are not an institutional afterthought: they change how parents report preferences to the education authority and the quality of placements it generates. 
Once the appeal stage is incorporated, Immediate Acceptance (IA) becomes less manipulable and can support equilibrium outcomes that Pareto dominate those of Deferred Acceptance (DA) under truth-telling, reversing the DA–IA welfare comparison.
In a preregistered experiment, we find results consistent with the theoretical predictions: appeals increase truth-telling in IA, leaving DA essentially unchanged, and widen the IA-DA efficiency gap.
Around 60\% of the IA efficiency gain appears before any appeal is upheld, because subjects anticipating the possibility of a potential appeal rank schools differently.

\vspace*{1em}
\end{abstract}

\begin{keyword}
\kwd{appeals}
\kwd{school choice}
\kwd{aftermarkets}
\end{keyword}

\begin{keyword}[class=JEL]
\kwd{D47}
\kwd{C78}
\kwd{C92}
\kwd{I28}
\end{keyword}

\end{frontmatter}

\setlength{\parskip}{0.25\baselineskip}
\newpage

\section{Introduction}
\label{sec:introduction}

School choice is the process through which students are assigned to schools via a centralized clearinghouse that asks and uses parents' preferences to determine placements. In systems as different as England, Northern Ireland, and New York City, the assignment procedure proceeds in three steps. First, parents rank the available schools. Second, the clearinghouse uses an algorithm to produce an allocation. Finally, parents may appeal each rejection they received from a school they ranked, and an independent panel decides which appeals to uphold. 

Economists have profoundly shaped the first part of this process. A large literature studies how families should rank schools, which matching mechanisms a clearinghouse should run, and how priorities should enter centralized admissions, and this work has driven major reforms in school assignment systems around the world. Yet, the final appeal stage has received almost no attention. The omission matters. In England alone, tens of thousands of families appeal secondary-school placements each year, and appeals succeed often enough to reshape outcomes: roughly one in five appeals is upheld, and about one in seven families rejected from their first-choice school ultimately gains admission to a more preferred school through this route \citep{hunt2019}. 

In this paper, we explore how appeal rights change the economics of school choice. We first ask: how do appeals change how parents rank schools? And second, how do differences in the way parents rank schools affect the final student placement? These questions are not merely theoretical. Families and school officials already understand that the possibility of a later appeal affects how families should rank schools.
In response to a parent asking how to order selective-school applications and a safer alternative, one school-admissions forum advised:

\begin{quote}
``{\it[Put] whichever Berks school you want ... higher than Burnham, so that you will then be able to appeal when you are refused a place.}''\footnote{Source: \url{https://www.elevenplusexams.co.uk/forum/11plus/viewtopic.php?t=44895}. Accessed on 3 September 2026.}
\end{quote}

This advice captures the mechanism we study. The parent is not told to abandon a safe option. She is told to rank the ambitious school above it, so that rejection preserves an appeal right while the lower-ranked school remains a fallback. A rank-order list therefore serves two roles. It determines the initial assignment, and it determines which rejected schools can later be challenged. Appeals can change behavior before the assignment is made, as well as outcomes after it is made.

We model school choice with appeals as a two-stage process. In the first stage, students submit rank-order lists and either Immediate Acceptance (IA) or Deferred Acceptance (DA) produces an initial assignment. In the second stage, students may appeal to schools that they ranked above their assignment. This structure provides a stylized representation of the statutory process in England, as mandated by the nationwide regulations in the School Admission Appeals Code.\footnote{School Admission Appeals Code 2022, paragraphs 3.5 to 3.10. The Code first requires the panel to determine whether the admission arrangements were unlawful or incorrectly applied and states that ``the panel must uphold the appeal'' when the child would otherwise have been offered a place. If the appeal is not upheld on those grounds, the panel must balance the family's case against the prejudice that admitting an additional child would cause the school. Department for Education statistics show substantial geographic variation in appeal outcomes, with success rates ranging from zero in some local authorities to above 90 percent in others in recent years.} Our strict rule captures the mandatory correction of a violation of admissions criteria. A second, probabilistic, rule provides a simple representation of the uncertainty surrounding the subsequent balancing stage. It is motivated by the wide variation in observed appeal success rates, which suggests idiosyncratic interpretations of the Appeals Code across local authorities. Successful appeals add seats rather than remove students admitted in the first stage, consistent with the Code's treatment of the remedy as the admission of an additional child. The final allocation is therefore determined jointly by the assignment mechanism and the appeal rule.

This two-stage structure reflects a feature that mechanism design typically abstracts away: the designer is implicitly a unitary authority that sets the rules, applies them, and has the final word. Constitutional law separates these functions (see e.g., \textit{Marbury v. Madison}, 5 U.S. (1 Cranch) 137 (1803)): allocations by public authorities are reviewable by an independent branch — typically the judiciary — that did not design the mechanism but can revise its outcomes. The allocation is thus produced jointly by mechanism and review, and the designer controls only the first.

Our first theoretical result establishes a sharp asymmetry in how parents' incentives are affected by appeals in the two assignment mechanisms. Strict appeals preserve the strategy-proofness of Deferred Acceptance. Since DA produces an assignment with no strict blocking pair relative to submitted preferences, the strict appeal stage has no priority violation to correct. Under Immediate Acceptance, the same rule changes incentives. A standard manipulation in IA is to skip applications to unfeasible schools because they harm chances at feasible second-choice schools that may be filled with lower-priority applicants by the time the second-choice school receives an application; this reasoning disappears once appeals are allowed, as a student can always overturn a priority violation through the appeal aftermarket. Appeals therefore make ambitious applications less risky. Although they do not make IA strategy-proof, they make it less manipulable in the profile-by-profile sense of \citet{pathak2013school}: every preference profile at which IA with strict appeals can be profitably manipulated is also manipulable under IA without appeals, while the converse fails at some profiles. 
The probabilistic rule adds a second channel that promotes truth-telling in IA: skipping a seemingly unattainable school throws away a lottery ticket through which the student might have gained admission.

Our main theoretical result concerns welfare under strategic behavior. Without appeals, \citet{ergin2006} show that every Nash equilibrium of the preference revelation game in Immediate Acceptance is weakly Pareto-dominated by truthful Deferred Acceptance. We show that this comparison can reverse once strict appeals are incorporated. For every school choice problem with strict priorities, Immediate Acceptance followed by strict appeals admits a Nash equilibrium whose outcome weakly Pareto-dominates truthful Deferred Acceptance, and makes at least one student strictly better off whenever the two outcomes differ. This is an equilibrium-existence result, not a claim that every equilibrium of IA with appeals dominates DA. Its implication is that the classical comparison between the two mechanisms depends on the review process that follows the initial assignment.

We test these predictions in a preregistered laboratory experiment with 1,022 participants across 59 independent matching groups. We vary both stages of the process. Participants face either IA or DA, followed by no appeals, strict appeals, or probabilistic appeals. The design lets us observe initial reports, first-stage assignments, appeal choices, and final assignments. We can therefore separate changes in behavior before any appeal is decided from the direct effect of successful appeals afterwards.

The reporting results closely follow the theoretical asymmetry. Under IA, strict and probabilistic appeals each raise truth-telling by approximately 11 percentage points, from 28 to 39 percent. Under DA, the corresponding effects are only 3 to 4 percentage points and are at most marginally significant. The meaning of a non-truthful report also differs sharply across mechanisms. Under DA, approximately 92 percent of non-truthful reports leave the participant's assignment unchanged, and every consequential deviation is harmful. Under IA, almost two thirds of non-truthful reports change the assignment, although harmful deviations are more common than beneficial ones. Appeals therefore reduce a behaviorally important strategic problem under IA, while most deviations observed under DA are assignment-irrelevant.

Appeals also substantially improve student welfare under IA. Strict appeals raise welfare under IA by approximately 0.3 ranks while leaving DA essentially unchanged. They increase the share of realized IA market outcomes that strictly Pareto-dominate truthful DA from 11 to 48 percent. The experiment does not reproduce a reversal in the sign of the average welfare ranking, because IA already performs better than DA without appeals in this environment. Instead, it confirms the comparative-static prediction of the theory: strict appeals improve IA much more than DA and move realized IA allocations sharply toward the Pareto region above truthful DA.

A decomposition of the welfare effect shows that approximately \(60\) percent of the observed
rank improvement under IA is already present in the first-stage assignment, before any appeal is decided, and is driven by changes in how subjects report their preferences. Only the remaining \(40\) percent is due to appeals overturning the initial assignment. This non-causal decomposition shows
that appeal rights affect outcomes before as well as after claims are resolved, but mainly through the first channel of changes in subjects' behavior.

Participants use the appeal stage systematically, although imperfectly. They are more likely to pursue claims against more valuable schools and claims supported by strict priority. At the same time, they file many appeals that cannot succeed and leave approximately three in ten sure, payoff-improving claims unfiled. This incomplete use of appeals helps explain why strict appeals reduce, but do not fully eliminate, the number of blocking pairs in IA.

The standard of review matters. Under IA, strict appeals more than halve the number of strict blocking pairs and substantially increase the share of stable market outcomes while creating almost no waste. Under DA, strict appeals leave the stable first-stage assignment unchanged. Probabilistic appeals are less targeted. They improve IA, but under DA they can disturb an initially stable assignment, create vacancies, and generate new blocking pairs.

Our contribution is to show that a centralized assignment procedure cannot always be evaluated independently of the institutions that review its decisions. 
When families can challenge an assignment, the review rule changes both the final allocation and the incentives driving the initial report. In our experiment, in fact, more through the second channel than the first.
The relevant object of analysis is therefore the assignment procedure together with the appeal rule that follows it. More broadly, whenever one institution can revise an allocation produced by another, the design of the review process can alter the usual ranking of assignment mechanisms. Separation of powers is thus not a detail to be added after the fact to a mechanism designed without it: because families anticipate review, its availability shapes the reports on which the initial assignment is built, and a mechanism that performs well when its output is final may perform differently under (judicial) review.

The remainder of the paper proceeds as follows. Section~\ref{sec:lit} discusses the related literature. Section~\ref{sec:setting} describes the institutional setting. Section~\ref{sec:model} presents the model and theoretical results. Section~\ref{sec:exp} describes the experimental design, and Section~\ref{sec:results} presents the results. Section~\ref{sec:conclusion} concludes.

\section{Related Literature}
\label{sec:lit}

Our work contributes to several strands of literature on school choice and aftermarkets. We organize our discussion around four interconnected themes.

\paragraph{Empirical literature.} \citet{taylor} show that appeals are more common in urban, more affluent, and higher-demand areas. \citet{hunt2019} finds that disadvantaged and minority pupils are less likely to secure a top-choice school through appeals, while parents who submit incomplete rank-order lists are 15\% more likely to appeal. Successful appeals also tend to place students in less deprived schools with fewer disadvantaged peers, suggesting that the appeals process may disproportionately benefit advantaged families.

\paragraph{Models with appeals.} Only a handful of papers explicitly examine the role of appeals. \citet{kojima2011robust} shows that allowing students to appeal their non-admission can compromise the stability and strategy-proofness of the DA mechanism. He introduces the concept of {robust stability}, arguing that if appeals are unrestricted, it becomes a weakly dominant strategy for students to forgo ranking schools initially and later appeal to their most preferred school with available spots (a result further generalized by \cite{afacan2012group}). Our model builds on this framework but introduces a crucial distinction: students can only appeal to schools they initially ranked and were rejected from. This restriction, aligned with real-world systems in England and Northern Ireland, prevents the strategic behavior identified by Kojima and ensures that appeals serve as a secondary step rather than a complete circumvention of the initial matching process.

\citet{afacan2017sticky} propose a different appeals model with two key differences: (1) appealing is costly, and (2) instead of a post-assignment appeal process, parents preemptively report how much instability they are willing to tolerate. They introduce the concept of {sticky stability}, where priority violations up to a fixed threshold are permitted, ensuring that every stable mechanism is also sticky stable. Unlike \cite{afacan2017sticky}, we do not model appeals as an ex-ante tolerance for instability, but as an ex-post institutional process that follows the initial matching.

Appeals also appear in the experiment of \citet{hakimov2026improving}, who study transparency and verifiability in centralized admissions. There, appeal decisions serve as an experimental measure of whether participants can detect that the announced mechanism was not followed, rather than as a second-stage allocation rule. Our paper instead treats appeals as part of the mechanism itself and studies how alternative rules for upholding them affect reporting incentives and final outcomes.

More recently, \citet{decerf2024incontestable} study when a matching is incontestable, that is, appeal-proof for students who observe priorities and capacities but not other students' preferences or assignments. Their focus is on the informational foundations of a weakened notion of stability. Our paper asks a different question. Rather than characterizing when a final matching is appeal-proof, we model appeals as an actual institutional second stage and study how alternative uphold rules reshape incentives and final outcomes.

While in our model every priority violation can be reviewed and corrected if the affected party appeals, a different yet related literature instead asks which priority violations can be justified, therefore allowing for matchings that Pareto-dominate stable ones. The justifiability of said priority violations can depend on ex-ante consent, the validity of reassignment claims that would follow, or whether a student's priority is violated, she does not benefit and yet could do better in an alternative DA improvement \citep{kesten2010school, troyan2020essentially, ehlers2020legal, reny2022efficient, ortega2026justifiable}. 

\paragraph{Immediate versus Deferred Acceptance.}
Our welfare result also contributes to the long-standing comparison between IA and DA. Under complete information and without appeals, every Nash equilibrium outcome of IA is weakly Pareto-dominated by truthful DA \citep{ergin2006}. This ranking is not robust to uncertainty and cardinal welfare: IA can outperform DA under incomplete information \citep{miralles2009school,abdulkadirouglu2011resolving,akyol2024bayesian}, and even under truth-telling IA does not produce placements that are meaningfully better than those generated by DA \citep{ortega2026asymptotic}. Studies using administrative data from Barcelona, Beijing, and Cambridge similarly find that IA yields higher estimated average welfare than DA \citep{calsamiglia2020structural,he2015gaming,agarwal2018demand}. Our paper identifies a different source of reversal: even under complete information, the appeal rule following the first-stage mechanism can overturn the classical Pareto comparison.

\paragraph{Beyond preference revelation.}
Our paper also connects to a broader literature that views school choice as a process extending beyond preference submission and the initial matching outcome. A first strand studies multi-stage admission procedures involving interviews, information acquisition, and other decisions before the matching takes place \citep{opendays, artemov2021assignment, chen2021information, chen2022information, maxey2024school}. A second strand studies aftermarkets and outside options, emphasizing that the assignment produced by a centralized clearinghouse need not coincide with the school or program a student ultimately attends \citep{akbarpour2022centralized}. Most closely related to our institutional perspective, \citet{kapor2024aftermarket} show that off-platform offers, declined placements, and frictional waitlists can alter final enrollment and generate externalities for other participants. Appeals similarly create a post-matching stage in which initial assignments may be revised. Unlike the aftermarket they study, however, appeals are governed by formal standards of review and are anticipated when families submit their rank-order lists; they therefore affect both final assignments and first-stage reporting incentives. Finally, our study relates to work allowing school capacities to adjust or be exceeded under specified conditions \citep{aygun2023placement, afacan2024}.

\section{Institutional Setting}\label{sec:setting}

%\josue{Finally, we connect the model to school admissions in England. The 2007 move away from first-preference-first admissions toward equal-preference systems provides a natural institutional setting in which appeal incentives and reported preferences are jointly affected by mechanism design. A central prediction is that the reform need not increase the reported top-choice offer rate. By making truthful and ambitious ranking safer, equal-preference admissions may lead families to list more ambitious first choices, mechanically lowering the share assigned to the reported first choice even if the mechanism improves incentives. At the same time, stable admissions should reduce the number of priority-based claims with legal force, while the effect on appeal volume is theoretically ambiguous because disappointed families may still appeal.}

Since the 1998 Education Reform Act, both England and Northern Ireland have maintained clearly defined appeals procedures for school admissions \citep{taylor}. When a child is refused a school place, parents receive a letter explaining how to appeal. The process is free and conducted by independent panels, with each school refusal requiring a separate appeal.
However, 
the criteria for upholding appeals differ between England and Northern Ireland. In England, an appeal must be upheld if either:

\begin{enumerate}
\item The school's admissions criteria were not properly followed or do not comply with the school admissions code.
\item The parents' reasons for admission outweigh the school's justification for not accepting additional students.\footnote{For England, see \url{https://www.gov.uk/schools-admissions/appealing-a-schools-decision}. For Northern Ireland, see \url{https://www.nidirect.gov.uk/articles/appealing-school-place-decision}. Accessed on 3 September 2026.}
\end{enumerate}

On the other hand, in Northern Ireland appeals are upheld only if:

\begin{enumerate}[]
\item The school did not apply, or misapplied, its admissions criteria.
\item The child would have been offered a place had the criteria been correctly applied.
\end{enumerate}

This regional difference captures the distinction between our theoretical uphold rules: England's second criterion allows discretionary decisions based on weighing competing claims, while Northern Ireland's approach focuses purely on procedural compliance.

Appeals concerning infant classes, the first three years of primary school (Reception and Years~1 and~2, typically ages four to seven), are subject to a stricter legal standard when admitting an additional child would breach the
statutory limit of 30 pupils per teacher. In 2024/25, \(9.7\%\) of appeals heard for infant classes were successful, compared with \(28.3\%\) for other primary classes.

Appeals operate at substantial scale and exhibit considerable geographic variation. For admissions at the start of the 2024/25 academic year,
\(14{,}584\) primary and \(37{,}289\) secondary appeals were lodged in England. Of the appeals heard, \(17.7\%\) of primary and \(19.9\%\) of
secondary appeals were successful. Variation across the 152 local authorities was much sharper: among authorities that heard at least one appeal, primary success
rates ranged from zero to \(66.7\%\), with 8 of 12 appeals successful in Westminster, while secondary success rates ranged from zero to \(95.1\%\), with 39 of 41 appeals successful in Sunderland.

Outside the UK, the New York City High School match also combines centralized admissions with a formal appeals process. Successful appeals are processed using
a variant of the top trading cycles mechanism \citep{morrill2024top}; in the 2003/04 academic year, approximately \(5.7\%\) of applicants, or \(5{,}100\)
out of \(90{,}000\), lodged an appeal, about half of which were successful \citep{abdulkadiroglu2005}. Our model more closely reflects the systems in
England and Northern Ireland, where independent panels review admission decisions and successful appeals admit an additional student rather than displacing an incumbent, with substantial variation within local authorities. These institutional features map naturally onto our model and appeal rule design in the next Section.

\section{Model}
\label{sec:model}

We model school choice with appeals as a two-stage process. In the first stage, students submit preference rankings to a centralized mechanism, which produces an initial matching. In the second stage, students may appeal to schools they ranked above their initial assignment, and an uphold rule determines which appeals are successful. Successful appeals create additional seats rather than displacing initially assigned students.

The first stage is a standard many-to-one matching problem, defined by the following primitives:

\begin{itemize}
\item a set of students $I$,
\item a set of schools $S$, each able to accept $q_s$ students, with the school $s_\emptyset$ of capacity $q_{s_\emptyset}=\infty$ representing being unassigned,
\item a strict preference $\succ_i$ of each student $i$ over schools,
\item a weak school priority $\unrhd_s$ of each school $s$, so that some students may have the same priority at some school, and
\item a single, strict tie-break order $\tau$ over students, which resolves priority ties, inducing a strict priority order.
\end{itemize}

We use $\succsim_i$ to denote the weak preference relation induced by $\succ_i$. A school choice problem is a tuple $(I,S,q,\succ,\unrhd,\tau)$. Since all primitives except student preferences are fixed throughout, we represent a school choice problem simply by $\succ$.

A matching $\mu$ is a function that assigns each student to a school or to the outside option. The school assigned to student $i$ in $\mu$ is denoted by $\mu_i$, whereas the set of students assigned to $s$ in $\mu$ is denoted by $\mu_s^{-1}$. A matching is feasible if no school is assigned more students than its capacity.

A weak blocking pair at matching \(\mu\) is a student-school pair \((i,s)\)
such that \(s\succ_i\mu_i\) and \(i\unrhd_s j\) for some
\(j\in\mu_s^{-1}\).
A strict blocking pair is a weak blocking pair in which the priority comparison is strict, i.e., $i \rhd_s j$ for some $j \in \mu_s^{-1}$. 
A matching \(\mu\) is non-wasteful if there is no student \(i\) and school
\(s\) such that $s\succ_i\mu_i$ and $|\mu_s^{-1}|<q_s$.
A matching is stable if it admits no strict blocking pair and is non-wasteful. 

A (not necessarily feasible) matching $\nu$ weakly Pareto-dominates matching $\mu$ if, for every student $i$, $\nu_i \succsim_i \mu_i$, and Pareto-dominates it if additionally $\nu_j \succ_j \mu_j$ for some student $j$.

A mechanism $M$ associates a feasible matching to every profile of preferences $\succ$, and we use $M_i(\succ)$ to denote the school to which student $i$ is assigned under mechanism $M$. 
We allow submitted preference lists to be truncated in the model,
with omitted schools treated as unacceptable.	
Due to their popularity, we focus on the immediate acceptance (IA, also known as the Boston mechanism) and the student-proposing deferred acceptance mechanism (DA). The IA mechanism works in sequential rounds, as follows:

\begin{itemize}
\item In the first round, each student applies to her most preferred school. Each school considers the students who have listed it as their first choice and assigns seats of the school to these
students one at a time, following their priority order, until there are no seats left or there is no student left who has listed it as her first choice.

\item In each subsequent round $k$, each remaining student applies to her $k$-th preferred school. Each school that was already full automatically rejects all new applicants. For each school with seats still available, consider the students who have listed it as their $k$-th choice and assign the remaining seats to these students one at a time following their priority order, until there are no seats left or there is no student left who has listed it as her $k$-th choice.
\end{itemize}

DA differs from IA in that schools' acceptances are temporary.
In each round, rejected students apply to their next preferred school. Each
school then considers its current tentative holders together with its new
applicants, keeps the highest-priority students up to capacity, and rejects
the rest.
$\ia(\succ)$ and $\da(\succ)$ denote the allocation produced by IA and DA for the matching problem $\succ$, respectively.

\paragraph{The Appeals Process.} We now describe the second-stage process. Given the publicly observable reported preferences $\hat \succ$ and an initial matching $M(\hat \succ)$, the set of appealable schools for student $i$ includes all schools ranked as more desirable by the student to which she was not admitted, i.e.,
\[
A_i(\hat \succ,M)\coloneqq
\{s\in S: s\, \hat\succ_i\, M_i(\hat \succ)\}.
\]
Thus, a student may appeal only to schools that she ranked as more preferred, i.e., schools that rejected her in the mechanisms we consider. This restriction is aligned with the rules in England and Northern Ireland.
After the initial matching is announced, student $i$ submits a set of appeals $a_i\subseteq A_i(\hat \succ,M)$ to an independent panel. The panel then applies an uphold rule, which determines which of these appeals succeed.

\begin{definition}
An uphold rule $r$ maps each submitted set of appeals $a_i$ into a subset $r(a_i)\subseteq a_i$, representing the appeals that are upheld.\footnote{The uphold rule depends on the reported preferences, priorities and initial matching too, but we make this dependence implicit to simplify notation.}
\end{definition}

Because appeals are free for families and hearings are conducted online, we assume that appealing is costless. We focus on rules in which the success of an appeal at a school $s$ is independent of whether a student decided to appeal her rejection at a different school $s'$. 
Under these assumptions, it is a weakly dominant strategy for each student
to appeal every school in her appealable set that she truly prefers to her
first-stage assignment. We therefore assume throughout that
\[
a_i
=
\{s\in A_i(\hat\succ,M):
s\succ_i M_i(\hat\succ)\}.
\]
If more than one appeal is upheld, the student enrols in her most preferred
school, according to \(\succ_i\), among the upheld appeals.

We consider two uphold rules in the main text, described in Table~\ref{tab:upholdrules}.

\begin{table}[H]
\centering
\caption{Uphold rules}
\label{tab:upholdrules}
\begin{tabular}{p{4cm}p{9cm}}
\toprule
Rule & An appeal by student $i$ to school $s\in a_i$ is upheld if \\
\midrule
Strict rule ($r_1$)
& $(i,s)$ is a strict blocking pair in $M(\hat \succ)$. \\[0.4em]
Probabilistic rule ($r_2$)
& $(i,s)$ is a strict blocking pair in $M(\hat \succ)$, or otherwise\\ & with independent probability $p\in(0,1)$. \\
\bottomrule
\end{tabular}
\end{table}
Rule \(r_1\), our main focus, captures the priority-based component of school
admission appeals in England.\footnote{In Appendix \ref{app:weakpriorities}, we consider a similar rule $r_3$, in which a student can successfully appeal if they were not admitted to a school but a student with the same priority was admitted due to a tie-break. This rule, unlike $r_1$ and $r_2$, breaks DA's strategy-proofness.
} Specifically, England's School Admission Appeals Code (Section 3.5) states:

\begin{quote}
``\emph{The panel must uphold the appeal...if it finds that the admission arrangements did not comply with admissions law or had not been correctly and impartially applied, and the child would have been offered a place if the arrangements had complied or had been correctly and impartially applied}.''
\end{quote}

Rule $r_2$ is motivated by the substantial variation in appeal success rates across local authorities, which range from 0\% to over 90\%, as well as by evidence from our survey of appeal panel members suggesting considerable heterogeneity in how similar cases are assessed (described in Appendix \ref{app:survey}). In the main text, we will focus on $r_1$, postponing the analysis of $r_2$ and $r_3$ to Appendices \ref{app:randomized} and \ref{app:weakpriorities}, respectively, highlighting only the main differences in the text. 

We denote the resulting composite mechanism by \(M\circ r\). Given \(\hat\succ\), it first applies \(M\), then applies the uphold rule \(r\) to the appeals generated from \(M(\hat\succ)\), and finally assigns each student either her first-stage school or, if at least one appeal is upheld, her most preferred upheld appeal according to \(\succ_i\).
To place mechanisms with and without appeals on the same footing, let $r_0$ denote the null uphold rule that rejects every appeal. Hence $M\circ r_0$ coincides with the baseline mechanism $M$ without appeals.

Note that the final matching need not be feasible: schools may admit more students than their published quotas when multiple appeals are upheld. This is consistent with real-world evidence. In England in 2023/2024, 24\% of secondary schools were over capacity, while during the same period, 25\% of post-primary schools in Northern Ireland exceeded their quotas due to successful appeals.\footnote{Source: England's Explore Education Statistics: \url{https://explore-education-statistics.service.gov.uk/find-statistics/school-capacity/2024-25}, and NI Education Authority FAQs: \url{https://www.eani.org.uk/parents/pupil-applications-and-grants/admissions/admissions-support}. Accessed on 3 September 2026.}

\subsection{Preliminary Observations}

We now present two straightforward but useful observations regarding the consequences of appeals. First, in our framework, every student becomes weakly better off by adding appeals, and the outcome after appeals are upheld admits no appealable strict blocking pair involving a school that the student truly prefers to her final assignment.

\begin{lemma}
\label{prop:pareto}
For every mechanism \(M\) and every report profile \(\hat\succ\),
\((M\circ r_1)(\hat\succ)\) weakly Pareto-dominates
\((M\circ r_0)(\hat\succ)\) with respect to true preferences, and
admits no strict blocking pair \((i,s)\) such that
\(s\in A_i(\hat\succ,M)\) and
\[
s\succ_i (M\circ r_1)_i(\hat\succ).
\]
\end{lemma}

We postpone a straightforward proof to Appendix~\ref{app:missingproofs}.
It is important to note that final outcomes need not be stable because the final matching can be wasteful: appeals may leave first-stage seats vacant that a student would prefer to her own, yet was unable to claim through an appeal because she had worse priority than all initially assigned students. 

Lemma~\ref{prop:pareto} implies an asymmetric effect of appeals in DA and IA. If DA is Pareto-dominated by a feasible matching, its outcome will remain inefficient after incorporating the appeals aftermarket. On the other hand, if IA generates appealable justified-envy violations involving schools that the affected students truly prefer, the appeals procedure removes those violations through capacity extensions, although the final outcome need not be stable because successful appeals may create vacancies elsewhere.

\subsection{Incentives}
We now focus on how appeals affect parents' incentives. It is well known that $\da$ is strategy-proof whereas $\ia$ is not, which has long been a central argument in favor of the former. We show that the possibility of later appeals can modify these incentive properties, making IA less manipulable, simply because manipulations in IA are no longer needed in some problems where a more preferred school that could be obtained through a manipulation can now be obtained with a truthful report followed by a subsequent appeal.

Formally, fix a student $i$ and a preference profile $(\succ_i,\succ_{-i})$. We say that student $i$ can profitably manipulate the composite mechanism $M \circ r$ at $(\succ_i,\succ_{-i})$ if there exists a different preference $\hat\succ_i$ such that
\[
(M \circ r)_i(\hat\succ_i,\succ_{-i}) \succ_i (M\circ r)_i(\succ_i,\succ_{-i}).
\]
The composite mechanism is strategy-proof if no student has a profitable manipulation, for every possible preference $\succ_i$, manipulation $\succ_i'$ and preference profile $\succ_{-i}$.
Mechanism $N$ is \emph{as manipulable as} mechanism $M$ if, whenever some student can profitably manipulate $M$ at some profile, that same student can profitably manipulate $N$ at that profile. Mechanism $N$ is \emph{more manipulable} than mechanism $M$ if, in addition, there exists a profile at which some student can profitably manipulate $N$ but no student can profitably manipulate $M$ \citep{pathak2013school}.

We first compare $\ia\circ r_0$ and $\ia\circ r_1$. 
\begin{proposition}
\label{thm:ia_manipulable}
$\ia\circ r_0$ is more manipulable than $\ia\circ r_1$.
\end{proposition}

\begin{proof}[Proof Sketch]
To show that every profitable manipulation of $\ia \circ r_1$ is also a
profitable manipulation of $\ia \circ r_0$, fixing $\succ_{-i}$, note that
under the manipulation student $i$ obtains her improved school $y$ either
directly in the first stage or through an appeal. In the former case the same
report is trivially profitable under $\ia \circ r_0$. In the latter, the
appeal succeeded because $y$ admitted a student with lower priority than $i$;
but then $i$, by ranking $y$ first, would have been admitted to $y$ in the
first round under $\ia \circ r_0$. 
Either way, $i$ secures $y$ under
$\ia\circ r_0$. By Lemma~\ref{prop:pareto}, her truthful
$\ia\circ r_1$ outcome is weakly preferred to her truthful
$\ia\circ r_0$ outcome. Since she strictly prefers $y$ to the former, she also
strictly prefers $y$ to the latter.
Example~\ref{ex:moremanipulable} below, with three students and
three unit-capacity schools, illustrates how a manipulation opportunity (in
parentheses for $i_1$) disappears in IA once appeals are incorporated.

\begin{example}
\label{ex:moremanipulable}
\textcolor{white}{l}

\begin{center}
	\begin{tabular}{CCCC|CCC}
		\succ_{i_1}&(\succ'_{i_1})& \succ_{i_2}&\succ_{i_3}&\unrhd_{s_1}&\unrhd_{s_2}&\unrhd_{s_3}\\
		\hline
		s_1&(s_2) &s_1&s_2&i_2&i_1&\cdot\\
		s_2&(s_1) &s_3&s_3&i_1&i_3&\cdot\\
		s_3&(s_3) &s_2&s_1&i_3&i_2&\cdot\\
	\end{tabular}
\end{center}

Under truthful reporting, the first-stage $\ia$ outcome is
\[
(i_1-s_3,\ i_2-s_1,\ i_3-s_2).
\]
Indeed, in round 1, $i_1$ and $i_2$ apply to $s_1$, and $i_2$ is accepted;
$i_1$ is rejected. Student $i_3$ applies to $s_2$ and is accepted. In round 2,
$i_1$ applies to $s_2$ but is rejected because $s_2$ is already full. In
round 3, $i_1$ applies to $s_3$ and is accepted.

Under $\ia\circ r_0$, student $i_1$ can profitably manipulate by reporting
$s_2 \succ'_{i_1} s_1 \succ'_{i_1} s_3$.
Then in round 1, student $i_1$ applies to $s_2$ and is accepted, $i_2$ applies
to $s_1$ and is accepted, and $i_3$ applies to $s_2$ but is rejected. Student
$i_3$ then applies to $s_3$ and is accepted. Hence the outcome is
\[
(i_1-s_2,\ i_2-s_1,\ i_3-s_3),
\]
which is strictly better for $i_1$ than her truthful $\ia\circ r_0$ assignment
$s_3$.

Now consider $\ia\circ r_1$ at the truthful profile. The first-stage outcome is
again
\[
(i_1-s_3,\ i_2-s_1,\ i_3-s_2).
\]
Student $i_1$ has a successful appeal to $s_2$, because $i_1$ ranks $s_2$
above her first-stage assignment $s_3$ and
$i_1 \rhd_{s_2} i_3$.
Thus, after appeals, $i_1$ obtains $s_2$.
No student can profitably manipulate $\ia\circ r_1$ at this profile. Students
$i_2$ and $i_3$ obtain their top choices truthfully, so they cannot improve.
Student $i_1$ obtains $s_2$ truthfully after appeals. Her only strictly better
school is $s_1$. But $i_2$ ranks $s_1$ first and has higher priority than
$i_1$ at $s_1$, so $i_1$ cannot obtain $s_1$ directly by any report and cannot
successfully appeal to $s_1$. Therefore $i_1$ cannot profitably manipulate
$\ia\circ r_1$.
\end{example}
The full proof appears in Appendix~\ref{app:missingproofs}.
\end{proof}

While we have seen that appeals reduce the manipulability of the Immediate Acceptance mechanism, they do not make it strategy-proof. This is a potential concern because IA's efficiency under truthfulness disappears once we account for strategic behavior, so that in the no-appeals case the Nash equilibria of the corresponding preference revelation game become weakly Pareto-dominated by DA's truthful outcome (as shown by \cite{ergin2006})\footnote{\cite{ortega2026note} quantifies IA's inefficiency loss in random markets when moving from truthful to strategic behavior.}. We will see shortly that this issue does not represent a concern in our setup.    

Before showing that equilibrium play need not have an adverse welfare effect on IA with appeals, we note that appeals upheld through $r_1$ leave DA's strategy-proofness untouched (also under $r_2$, but not under $r_3$, as we show in Appendix~\ref{app:weakpriorities}).
\begin{proposition}
\label{thm:da_manipulable}
$\da \circ r_1$ is strategy-proof.
\end{proposition}

\begin{proof}
Note that the DA outcome never admits a blocking pair under any reported preferences. Hence $r_1$ never upholds any appeal, and the composite mechanism coincides with DA, which is strategy‑proof.
\end{proof}

\subsection{The Classical Ranking Can Reverse}

Proposition~\ref{thm:ia_manipulable} shows that appeals reshape the incentive properties of standard mechanisms. We now show that appeals can also reverse the classical welfare comparison between DA and IA in the Nash equilibrium of their induced preference revelation games under complete information. Without appeals, a seminal result by \cite{ergin2006} shows that every Nash equilibrium of the IA game is weakly Pareto-dominated by the unique equilibrium in weakly dominant strategies of DA. But when we add appeals, not only does there exist a Nash equilibrium of the IA game that produces a weakly better outcome than when students report their preferences truthfully; in fact, there exists a Nash equilibrium that weakly Pareto-dominates (strictly so if DA is inefficient\footnote{\cite{ortega2026large} show that DA is Pareto-inefficient with high probability, and that the majority of students can be simultaneously improved without harming anyone else.}) the DA truthful outcome with or without appeals.   

\begin{theorem}
\label{thm:ia_r1_dominates_da}
For every school choice problem with strict priorities, there exists a Nash
equilibrium \(\sigma^*\) of \(\ia\circ r_1\) such that
\[
(\ia\circ r_1)(\sigma^*) \succsim \da(\succ).
\]
Furthermore, if
\[
(\ia\circ r_1)(\sigma^*)\neq \da(\succ),
\]
then at least one student is strictly better off under
\((\ia\circ r_1)(\sigma^*)\) than under \(\da(\succ)\).
\end{theorem}

\begin{proof}
Our proof strategy is as follows: we construct an auxiliary one-shot mechanism, \emph{Modified IA} ($\mia$), that
runs like IA on truthful preferences but overbooks a rejected student whenever
she outranks an admitted student, mimicking a successful strict appeal within
the mechanism itself. We show that MIA weakly Pareto-dominates DA, using DA's stability, and then implement the MIA outcome as a
Nash equilibrium of $\ia\circ r_1$: regularly admitted students rank their MIA
school first, while overbooked students report truthfully and truncate at their
MIA school, so that they are rejected in the first stage and recover it through
an upheld appeal. The first stage of IA thus reproduces MIA's regular
admissions and the appeal stage its overbooking, so the equilibrium outcome
equals $\mia(\succ)\succsim\da(\succ)$.

\paragraph{Modified IA.}
The mechanism \(\mia\) proceeds in rounds, as IA does. In round
\(k\), every currently unassigned student applies to the \(k\)-th school on
her true preference list. Admissions are permanent. At each school \(s\),
regular seats are filled only while fewer than \(q_s\) students have been
regularly admitted to \(s\). If \(s\) has remaining regular capacity, it admits
the highest-priority current applicants up to the remaining regular capacity.
Once \(q_s\) students have been regularly admitted to \(s\), no further regular
admissions occur. Any further current applicant \(i\) is admitted by
overbooking if and only if $i\rhd_s j$
for some regularly admitted student \(j\) at \(s\). In that case, \(s\)'s
effective capacity increases by one. Students admitted within the regular
capacity are called \emph{regularly admitted}; students admitted through the
capacity-expansion rule are called \emph{overbooked}. All other applicants are
rejected and proceed to the next school on their list in the following round.

A useful consequence of this definition is that no student can be overbooked
at a school she ranks first. If \(s\) is her first choice, she applies to
\(s\) in round 1. If she is among the highest-priority applicants selected for
the regular seats, she is regularly admitted. If she is not, then all regularly
admitted students at \(s\) have higher priority than her, so she cannot trigger
overbooking. Hence she is rejected.
Let \(\mia(\succ)\) denote the outcome of this modified mechanism. We now present a useful property of MIA, which we call the Rejection Lemma.

\begin{lemma}
\label{lem:mia_rejection}
If student \(i\) is rejected from school \(s\) under \(\mia(\succ)\), then
every regularly admitted student at \(s\) under \(\mia(\succ)\) has strictly
higher priority at \(s\) than \(i\).
\end{lemma}

\begin{proof}
Student \(i\) is rejected from \(s\) only if all regular seats at \(s\) have
already been filled and \(i\) has lower priority than every regularly admitted
student at \(s\). Since regular admittees are never removed, and no further
regular admittees can be added after the \(q_s\) regular seats are filled,
every regularly admitted student at \(s\) under \(\mia(\succ)\) has strictly
higher priority at \(s\) than \(i\).
\end{proof}

By the definition of overbooking, if student \(i\) is overbooked at school
\(s\) under \(\mia(\succ)\), then there exists a regularly admitted student
\(j\) at \(s\) such that $i\rhd_s j$.

Having established the Rejection Lemma, we now show that MIA weakly Pareto-dominates DA.

\begin{lemma}
\label{lem:mia_dominates_da}
For every school choice problem with strict priorities,
\[
\mia(\succ)\succsim \da(\succ).
\]
\end{lemma}

\begin{proof}
Suppose, toward a contradiction, that some student is strictly better off under
\(\da(\succ)\) than under \(\mia(\succ)\). Let
\[
B=\{i\in I:\da_i(\succ)\succ_i \mia_i(\succ)\}.
\]
By assumption, \(B\neq\varnothing\).

For each \(i\in B\), let $s_i=\da_i(\succ)$.
Since \(s_i\succ_i \mia_i(\succ)\), student \(i\) applied to \(s_i\) under
\(\mia\) before receiving her \(\mia\)-assignment, and was rejected. Let
\(r_i\) be the round in which \(i\) applied to \(s_i\) under \(\mia\). Let
\(\rho_i\) be the round in which \(i\) receives her \(\mia\)-assignment, and
set \(\rho_i=\infty\) if \(i\) is unassigned under \(\mia\). Then $r_i<\rho_i$.

By Lemma~\ref{lem:mia_rejection}, every regularly admitted student at \(s_i\)
under \(\mia\) has strictly higher priority at \(s_i\) than \(i\). Since
\(i\) was rejected from \(s_i\), the \(q_{s_i}\) regular seats at \(s_i\) had
already been filled. Let \(D(s_i)\) denote the set of these \(q_{s_i}\)
regular admittees.

Since \(i\notin D(s_i)\), \(\da_i(\succ)=s_i\), and school \(s_i\) has only
\(q_{s_i}\) seats under DA, at most \(q_{s_i}-1\) students from \(D(s_i)\) can
be assigned to \(s_i\) under \(\da(\succ)\). Hence there exists some student \(h\in D(s_i)\)
who is not assigned to \(s_i\) under \(\da(\succ)\).

Because \(h\rhd_{s_i} i\) and \(i\) is assigned to \(s_i\) under
\(\da(\succ)\), stability of \(\da(\succ)\) implies that
$\da_h(\succ)\succ_h s_i$.
Otherwise, \((h,s_i)\) would be a blocking pair of \(\da(\succ)\). But
\(h\)'s \(\mia\)-assignment is \(s_i\), since \(h\) is regularly admitted to
\(s_i\) under \(\mia\). Therefore \(\da_h(\succ)\succ_h \mia_h(\succ)\),
so \(h\in B\).
Moreover, \(h\) was regularly admitted to \(s_i\) before or at the round in
which \(i\) was rejected from \(s_i\). Hence $\rho_h\le r_i<\rho_i$.

This shows that from any \(i\in B\), we can construct another student
\(h\in B\) with \(\rho_h<\rho_i\).
If all students in \(B\) had \(\rho_i=\infty\), this argument would produce a
student in \(B\) with finite \(\rho_h\), a contradiction. Hence \(B\) contains
at least one student with finite \(\rho_i\). Choose \(i^*\in B\) with minimal
finite \(\rho_{i^*}\). Applying the argument above to \(i^*\) produces
\(h\in B\) with finite \(\rho_h<\rho_{i^*}\), contradicting the minimality of
\(\rho_{i^*}\). Therefore \(B=\varnothing\), and \(\mia(\succ)\succsim \da(\succ)\).
\end{proof}

\paragraph{Constructing the report profile.} Having established that MIA weakly Pareto dominates DA, we proceed to construct the equilibrium strategies that deliver our result.

Partition students according to their status under \(\mia(\succ)\):
\begin{itemize}
\item \(D\): regularly admitted students;
\item \(O\): overbooked students;
\item \(U\): unassigned students.
\end{itemize}

Define the strategy profile \(\sigma^*\) as follows:
\begin{itemize}
\item If \(i\in D\) and \(\mia_i(\succ)=s\), then \(i\) ranks \(s\) first,
followed by the remaining schools in her true order.
\item If \(i\in O\) and \(\mia_i(\succ)=s\), then \(i\) reports truthfully
down to \(s\), truncating her list after \(s\).
\item If \(i\in U\), then \(i\) reports truthfully.
\end{itemize}

The general construction permits truncation. The laboratory interface requires
complete rank-order lists, but the particular experimental environment supports
the relevant equilibrium with complete lists, as shown in
Section~\ref{sec:exp}.

\paragraph{Implementation.}	
We show that
\[
(\ia\circ r_1)(\sigma^*)=\mia(\succ).
\]

Fix a school \(s\), and let \(D_s\) denote the set of students regularly admitted to \(s\) under \(\mia\). Every student in \(D_s\) ranks \(s\) first under \(\sigma^*\), so every student in \(D_s\) applies to \(s\) in round 1 of IA.
We claim that the round-1 admits at \(s\) under \(\ia(\sigma^*)\) are exactly the students in \(D_s\).

First suppose that \(|D_s|<q_s\). Then no student outside \(D_s\) ranks \(s\) first under \(\sigma^*\). To see this, suppose some student \(i\notin D_s\) also ranks \(s\) first under \(\sigma^*\). If \(i\in D\), then \(i\) ranks her
own \(\mia\)-school first, so she cannot rank \(s\) first unless she belongs to \(D_s\), a contradiction. If \(i\in O\) or \(i\in U\), then \(s\) is also her first choice under \(\succ_i\). Hence \(i\) applied to \(s\) in round 1 under
\(\mia\). 

Since \(|D_s|<q_s\), school \(s\) never filled its regular capacity under \(\mia\). Hence a student applying to \(s\) in round 1 could not have been rejected from \(s\), and would have been regularly admitted, again a
contradiction.

Therefore the only round-1 applicants to \(s\) under \(\sigma^*\) are the students in \(D_s\), and they are admitted.

Now suppose that \(|D_s|=q_s\). Consider any student \(i\notin D_s\) who also applies to \(s\) in round 1 under \(\sigma^*\). There are three cases.

First, suppose \(i\in O\) and \(\mia_i(\succ)=s\). This is impossible. If \(s\) were \(i\)'s first-ranked school under \(\sigma_i^*\), then \(s\) would also be \(i\)'s first choice under \(\succ_i\). But no student can be overbooked at her first-choice school under \(\mia\), as observed above.

Second, suppose \(i\in O\), \(\mia_i(\succ)=t\neq s\), and \(s\succ_i t\). Then \(i\) was rejected from \(s\) under \(\mia\). By Lemma~\ref{lem:mia_rejection}, every student in \(D_s\) has strictly higher priority at \(s\) than \(i\).

Third, suppose \(i\in U\) and \(s\) is her first choice. Then \(i\) was
rejected from \(s\) under \(\mia\). Again, by Lemma~\ref{lem:mia_rejection},
every student in \(D_s\) has strictly higher priority at \(s\) than \(i\).

Thus every student in \(D_s\) has higher priority at \(s\) than every
additional round-1 applicant to \(s\) under \(\sigma^*\). Since
\(|D_s|=q_s\), the students in \(D_s\) are exactly the round-1 admits at \(s\).

Now consider \(i\in O\), and let \(\mia_i(\succ)=s\). Under
\(\sigma_i^*\), student \(i\) applies to all schools she truly prefers to
\(s\), and then to \(s\). At every school \(t\succ_i s\), student \(i\) was
rejected under \(\mia\). By Lemma~\ref{lem:mia_rejection}, every regular
admittee at \(t\) has strictly higher priority than \(i\). Since the round-1
IA admits at \(t\) are exactly the students in \(D_t\), student \(i\) is
rejected from every such \(t\) in the IA first stage.

At \(s\), the regular seats are held by \(D_s\), so \(i\) is rejected in the
first stage. Her list is truncated at \(s\), so she applies to no further
school. Since \(i\) is overbooked at \(s\) under \(\mia\), there exists
\(j\in D_s\) such that
$i\rhd_s j$.

This student \(j\) occupies \(s\) in the IA first stage. Therefore \((i,s)\)
is a strict blocking pair of the first-stage outcome, and \(r_1\) upholds
\(i\)'s appeal. Hence \(i\) obtains \(s\), exactly as under \(\mia\).

Finally, consider \(i\in U\). Student \(i\) was rejected from every school
under \(\mia\). By Lemma~\ref{lem:mia_rejection}, every regular admittee at
every school \(t\) has strictly higher priority than \(i\). Since the
first-stage IA seats at \(t\) are filled by the students in \(D_t\), student
\(i\) is rejected everywhere in the first stage and has no successful strict
appeal. Hence \(i\) remains unassigned, as under \(\mia\).

We conclude that \((\ia\circ r_1)(\sigma^*)=\mia(\succ)\), as desired.

\paragraph{\(\sigma^*\) is a Nash equilibrium.}

Fix a student \(i\), and consider any school \(t\) such that $t\succ_i \mia_i(\succ)$.
Student \(i\) was rejected from \(t\) under \(\mia\). By
Lemma~\ref{lem:mia_rejection}, every regular admittee of \(t\) has strictly
higher priority at \(t\) than \(i\). Under \(\sigma^*_{-i}\), these regular
admittees rank \(t\) first and fill the first-stage IA seats at \(t\),
independently of \(i\)'s report. Since IA admissions are permanent, \(i\)
cannot obtain \(t\) directly, regardless of when she applies to \(t\).
Nor can \(i\) obtain \(t\) through a strict appeal. The first-stage occupants
of \(t\) are precisely the students in \(D_t\), all of whom have strictly
higher priority at \(t\) than \(i\). Thus \((i,t)\) is not a strict blocking
pair, and \(i\)'s appeal to \(t\) fails.
Therefore no deviation allows \(i\) to obtain a school she strictly prefers to
\(\mia_i(\succ)\). Since \(i\) was arbitrary, \(\sigma^*\) is a Nash
equilibrium of \(\ia\circ r_1\).

Combining the implementation result with Lemma~\ref{lem:mia_dominates_da},
we obtain
\[
(\ia\circ r_1)(\sigma^*)=\mia(\succ)\succsim\da(\succ).
\]
If \((\ia\circ r_1)(\sigma^*)\neq \da(\succ)\), then Pareto dominance together
with strict preferences implies that at least one student is strictly better
off. This proves the theorem.
\end{proof}

Theorem~\ref{thm:ia_r1_dominates_da} is an equilibrium-existence result: it shows that once strict appeals are
available, IA admits an equilibrium whose outcome weakly Pareto-dominates the
truthful DA outcome. It does not say that IA with appeals is always better
than DA, nor that every equilibrium of IA with appeals has this property.
This distinction is important. The point is not that IA becomes uniformly
superior once appeals are introduced. Rather, the theorem shows that appeals
change the comparison between mechanisms in a fundamental way. In the standard
one-stage model, IA is strategically fragile relative to
DA's efficiency. With appeals, however, some of the risk that generates ``defensive'' ranking
under IA is shifted to the second stage. A student who is rejected from a
preferred school may still obtain it if the rejection creates a strict priority
claim. Thus appeals act as partial insurance against one of IA's central
strategic costs.

Together, Proposition~\ref{thm:ia_manipulable} and Theorem~\ref{thm:ia_r1_dominates_da} suggest that appeals are not a small
institutional detail added after the mechanism has done its work. They change
both the feasible final outcomes and the strategic environment in which
families submit rankings. Consequently, a model that compares DA and IA without
appeals may miss an important part of the institutional trade-off.

In Appendix~\ref{app:randomized}, we extend the reversal result to the
probabilistic rule \(r_2\). Let \(\mu=\da(\succ)\), and let each student
report her true ranking truncated immediately after her assignment under
\(\mu\). For every school-choice problem with strict priorities and every
\(p\in(0,1)\), this DA-cutoff profile is an sd-Nash equilibrium of
\(\ia\circ r_2\), and its induced lottery sd-Pareto-dominates the
degenerate truthful DA outcome. The equilibrium notion is ordinal: no
unilateral deviation yields a stochastically dominant lottery.
Probabilistic appeals therefore preserve the reversal in
stochastic-dominance terms while turning otherwise unsuccessful
applications into lottery tickets.

\subsection{Hypotheses}

The theoretical analysis yields predictions concerning reporting behavior,
student welfare, and the number of blocking pairs.

\begin{hypo}
\label{hyp:truth}
{\normalfont Truth-telling
(Propositions~\ref{thm:ia_manipulable} and
\ref{thm:da_manipulable}).}
\begin{enumerate}[label=(1.\Alph*)]
\item Truth-telling is more frequent under \(\ia\circ r_1\) than under
\(\ia\circ r_0\), and more frequent under \(\ia\circ r_2\) than under
\(\ia\circ r_0\).

\item The effect of strict and probabilistic appeals on truth-telling is
smaller under DA than under IA.
\end{enumerate}
\end{hypo}

\begin{hypo}
\label{hyp:reversal}
{\normalfont Efficiency
(Theorem~\ref{thm:ia_r1_dominates_da}).}
Strict appeals improve student welfare more under IA than under DA.
The same directional comparison holds in stochastic-dominance terms under
the probabilistic rule.
\end{hypo}

The welfare prediction is based on the equilibrium constructed in
Theorem~\ref{thm:ia_r1_dominates_da}. It is therefore an equilibrium-selection
prediction rather than a claim that every equilibrium of IA with appeals
dominates the truthful DA outcome.

\begin{hypo}
\label{hyp:stability}
{\normalfont Blocking pairs
(Lemma~\ref{prop:pareto}).}
Strict appeals reduce the incidence of strict blocking pairs under IA.
They do not improve this outcome under DA, which already produces an
assignment without strict blocking pairs.
\end{hypo}

We do not formulate an analogous hypothesis for probabilistic
appeals. Unlike the strict rule, the probabilistic rule does not guarantee the
elimination of blocking pairs.

\section{Experiment}\label{sec:exp}

To test our hypotheses, we conduct a laboratory experiment in the spirit of the seminal school-choice experiment of \cite{chen2006school}.\footnote{For a comprehensive review of the experimental literature on school choice, see \cite{hakimov2021experiments}.}
Participants take part in a simulated school-choice market consisting of seven students and three schools, each with two seats.
Preferences and priorities are as follows:

\begin{center}
\begin{tabular}{CCCCCCC|CCC}
\succ_{i_1} & \succ_{i_2} & \succ_{i_3} & \succ_{i_4} & \succ_{i_5} & \succ_{i_6} & \succ_{i_7}
& \unrhd_{s_1} & \unrhd_{s_2} & \unrhd_{s_3} \\
\hline
s_3 & s_3 & s_2 & s_1 & s_1 & s_2 & s_3
& i_6,i_7 & i_4 & i_4,i_5 \\
s_2 & s_2 & s_1 & s_3 & s_2 & s_3 & s_2
& i_5 & i_1,i_2,i_3,i_5 & i_3,i_6 \\
s_1 & s_1 & s_3 & s_2 & s_3 & s_1 & s_1
& i_2,i_3,i_4 & i_6,i_7 & i_2,i_7 \\
\cdot & \cdot & \cdot & \cdot & \cdot & \cdot & \cdot
& i_1 & \cdot & i_1 \\
\end{tabular}
\end{center}

Comma-separated students have equal underlying priority, with ties broken in favor of the lower-indexed student. Participants receive \(15\), \(10\), or \(5\) points when assigned to their first, second, or third choice, respectively, and \(0\) points when unmatched.

We use a \(2\times3\) factorial between-subjects design.
The first factor varies the first-stage assignment mechanism between immediate acceptance (\(\ia\)) and deferred acceptance (\(\da\)).
The second factor varies the appeal rule between no appeals (\(r_0\)), strict appeals (\(r_1\)), and probabilistic appeals (\(r_2\)).
Under \(r_1\), an appeal is upheld only when the appellant has a strict priority claim.
Under \(r_2\), strict priority claims are upheld with certainty, while all other appeals are upheld independently with probability \(p=0.1\).
The six treatments are therefore
\begin{eqnarray*}
\da\circ r_0,\qquad
\da\circ r_1,&\qquad
\da\circ r_2,\qquad\\
\ia\circ r_0,\qquad
\ia\circ r_1,&
\ia\circ r_2.
\end{eqnarray*}

Participants are randomly assigned to one of the six treatments for the duration of the experiment.
The experiment proceeds for ten periods.
In each period, sessions of either 14 or 28 participants were partitioned into seven-person markets. Within each market, the seven student types were assigned one-to-one at random. Participants were rematched across periods within their session. Because rematching occurred across the entire session, each session constitutes one independent observation (matching group), and all standard errors are clustered at the matching group level.
%We preregistered twelve independent 14-participant sessions per treatment. In implementation, some sessions contained 28 rather than 14 participants. Because participants were rematched across the entire session, each 28-participant session constitutes one independent observation rather than two. 
We collected between 9 and 12 independent observations per treatment and 59 independent sessions in total.

Participants have complete information about the preferences of all student types and the priorities of all schools.
They do not, however, observe the rank-order lists submitted by the other participants.
Each participant submits a complete rank-order list over the three schools; truncation and ranking the outside option are not permitted.
Under truthful reporting, the relevant first- and second-stage assignments are
\[
\begin{array}{rcl}
\da(\succ)=(\da\circ r_1)(\succ)
&=&
(s_1:\{i_6,i_7\};\
s_2:\{i_1,i_2\};\
s_3:\{i_4,i_5\};\
s_\emptyset:\{i_3\}), \\[0.5ex]
\ia(\succ)
&=&
(s_1:\{i_4,i_5\};\
s_2:\{i_3,i_6\};\
s_3:\{i_2,i_7\};\
s_\emptyset:\{i_1\}), \\[0.5ex]
(\ia\circ r_1)(\succ)
&=&
(s_1:\{i_4,i_5\};\
s_2:\{i_1,i_3,i_6\};\
s_3:\{i_2,i_7\}).
\end{array}
\]
Thus, strict appeals leave the truthful DA assignment unchanged in this market, whereas under IA they allow \(i_1\) to obtain a seat at \(s_2\).
For this preference profile, truthful reporting is a Nash equilibrium of the
preference revelation game induced by \(\ia\circ r_1\).
Students \(i_2,\ldots,i_7\) obtain their first-choice schools, while \(i_1\)
obtains her second choice, \(s_2\), through a successful strict appeal.
The only school that \(i_1\) strictly prefers to \(s_2\) is \(s_3\), whose
two seats are filled in the first round by \(i_2\) and \(i_7\), both of whom
have strictly higher priority than \(i_1\).
No report can therefore give \(i_1\) a strictly better final assignment.
This equilibrium is supported by complete rank-order lists and hence does not
rely on truncation.

In the appeal treatments, participants observe the first-stage assignment and may appeal to any school that they ranked above their assignment and from which they were rejected.
Filing an appeal carries no monetary fee but requires completing a short typing task.
This small non-monetary friction is identical under the strict and probabilistic appeal rules.
It makes each appeal an active and observable choice without changing either the substantive grounds on which it is evaluated or its probability of success.
This feature deliberately differs from the theoretical benchmark.
In the experiment, participants instead choose whether to appeal and, when several appeal opportunities are available, which claims to pursue.
Appeal take-up is therefore an additional behavioral outcome of the experiment.

Because participants are rematched across the entire 14- or 28-person session, each session (matching group) constitutes one independent observation.
We cluster all standard errors at the session level.

Experimental earnings are denominated in points and converted at the rate \(1\) point \(=\pounds 0.05\).
At the end of the experiment, two of the ten periods are randomly selected for payment, and participants receive the sum of their earnings in those periods together with a \pounds 5 show-up fee.
Before entering the simulated school-choice market, participants must correctly answer a set of comprehension questions.
After the ten periods, they complete a short incentivized risk-elicitation task, a questionnaire concerning their appeal decisions in the appeal treatments, and a sociodemographic survey.

%We preregistered a target of twelve independent observations of fourteen participants per treatment, corresponding to 168 participants per treatment and 1,008 participants in total. This target was chosen ex ante to achieve statistical power of \(0.90\). In implementation, some matching groups contained 28 rather than 14 participants. Because rematching occurs across the entire group, these larger groups count as one independent observation.
We collected between 9 and 12 independent observations per treatment, i.e., 168 participants in each treatment, except for \(\ia\circ r_2\), which included 182 participants.
The experiment therefore involved 1,022 participants across 59 independent matching groups.
Following the session records, we exclude 25 observations recorded after six participants had exited their sessions.
The resulting analysis sample contains \(10{,}195\) participant-period observations.

The experiment received ethics approval from the institutional review board at City St George's, University of London, and was preregistered on the Open Science Framework (\url{https://osf.io/q92de/overview}).
The experiment was conducted at EssexLab at the University of Essex.
%The data and replication materials are available at \textbf{\textcolor{red}{XXX}}.

\section{Results}
\label{sec:results}

We begin by presenting results on individual-level outcomes (truth-telling and appeals) and then discuss market-level outcomes (efficiency and stability). We conclude with a discussion of the trade-offs observed in all treatments.

\subsection{Truth-telling}\label{sec:res_truth}

We define truth-telling as exactly submitting the induced preference ranking. Figure~\ref{fig:truth_telling_treatment} reveals a clear asymmetry in how appeals affect truth-telling under the two mechanisms. Without appeals, truth-telling is \(0.29\) under DA and \(0.28\) under IA, a difference we cannot distinguish from zero (\(p=0.626\)); we note that this is a failure to reject rather than evidence of equivalence. This baseline pattern is striking from a theory point of view, given DA's strategy-proofness, but it is in line with other laboratory experiments where substantial manipulation has been observed under DA, among both student subjects and parents \citep{freer2026experimental,cerrone2024consent}. Appeals increase truth-telling under IA. The percentage of truthful reports rises from \(28\) without appeals to \(39\) under both strict and probabilistic appeals. The corresponding increases under DA are considerably smaller, from \(29\) without appeals to \(32\) under strict appeals and \(33\) under probabilistic appeals.

%YH: I created a new plot and placed averages above the CIs for readibility.
\begin{figure}[htbp]
\centering
\includegraphics[width=\textwidth]{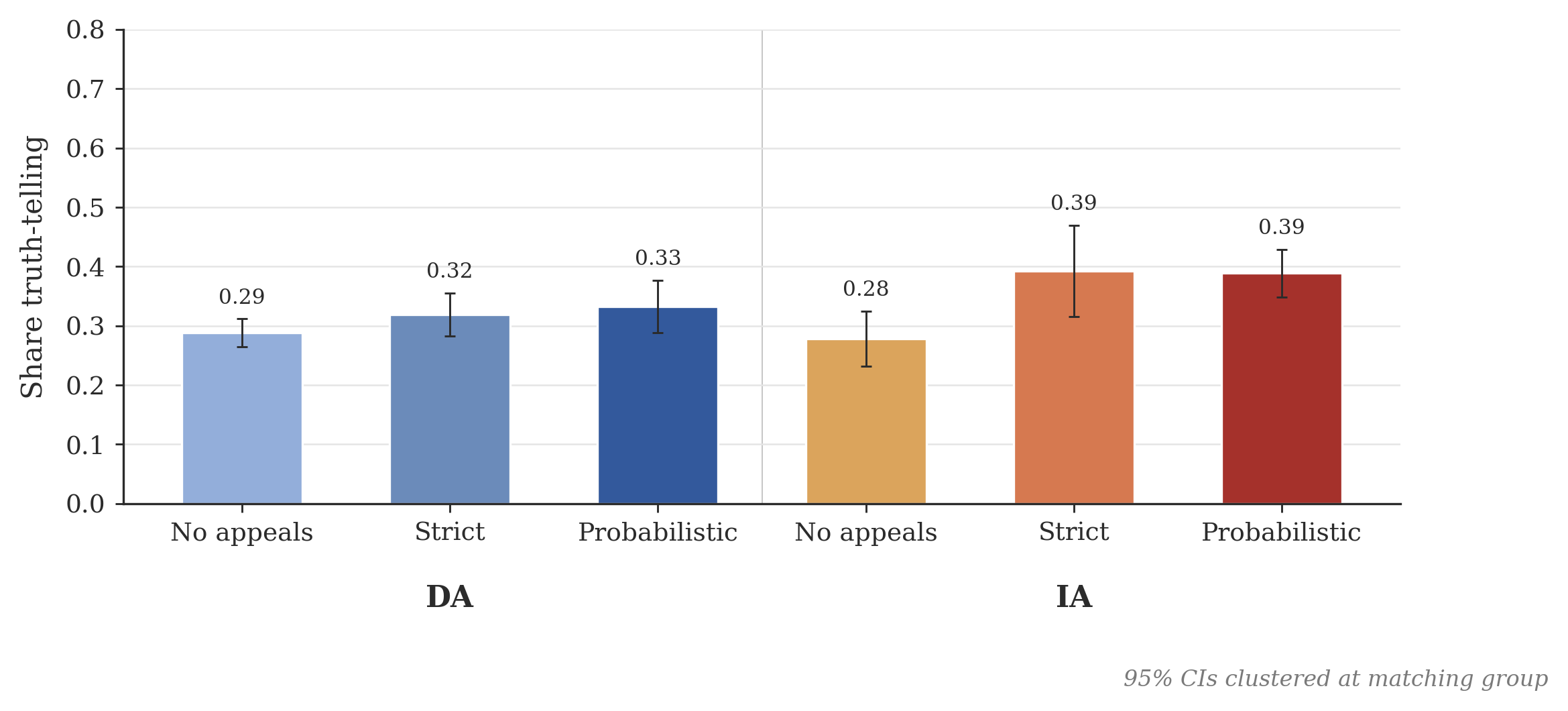}
\caption{Truth-telling by treatment.\\ \footnotesize Bars show the share of participants submitting the complete induced preference ranking. Whiskers are 95\% confidence intervals clustered at the matching-group level.}
\label{fig:truth_telling_treatment}
\end{figure}

To test these treatment differences more formally and account for changes over repeated play, including learning, Table~\ref{tab:truth_appeals} reports average marginal effects from logit regressions with progressively richer controls. Specification~(1) includes treatment indicators only. Specification~(2) adds round effects. Specification~(3) adds participant-type effects, and Specification~(4) adds age, gender, and risk aversion. Standard errors are clustered at the matching-group level.

\begin{table}[htbp]
\centering
\caption{Impact of appeals on truth-telling.}
\label{tab:truth_appeals}
\resizebox{\textwidth}{!}{\begin{threeparttable}
\begin{tabular*}{\textwidth}{@{\extracolsep{\fill}} lcccc @{}}
\toprule
& \((1)\) & \((2)\) & \((3)\) & \((4)\)\\
\midrule
\multicolumn{5}{l}{\textbf{Panel A. IA treatments}}\\
IA strict
& \(0.11^{***}\) & \(0.11^{***}\) & \(0.11^{***}\) & \(0.12^{***}\)\\
& {\footnotesize \((0.04)\)} & {\footnotesize \((0.04)\)} & {\footnotesize \((0.04)\)} & {\footnotesize \((0.04)\)}\\
IA probabilistic
& \(0.11^{***}\) & \(0.11^{***}\) & \(0.11^{***}\) & \(0.11^{***}\)\\
& {\footnotesize \((0.03)\)} & {\footnotesize \((0.03)\)} & {\footnotesize \((0.03)\)} & {\footnotesize \((0.03)\)}\\
\addlinespace
\multicolumn{5}{l}{\textbf{Panel B. DA treatments}}\\
DA strict
& \(0.03\) & \(0.03\) & \(0.03\) & \(0.03^{*}\)\\
& {\footnotesize \((0.02)\)} & {\footnotesize \((0.02)\)} & {\footnotesize \((0.02)\)} & {\footnotesize \((0.02)\)}\\
DA probabilistic
& \(0.04^{*}\) & \(0.04^{*}\) & \(0.04^{*}\) & \(0.04^{*}\)\\
& {\footnotesize \((0.02)\)} & {\footnotesize \((0.02)\)} & {\footnotesize \((0.02)\)} & {\footnotesize \((0.02)\)}\\
\addlinespace
\multicolumn{5}{l}{\textbf{Panel C. Difference in appeal effects between IA and DA}}\\
Strict appeals
& \(0.08^{*}\) & \(0.08^{*}\) & \(0.08^{*}\) & \(0.09^{*}\)\\
& {\footnotesize \((0.05)\)} & {\footnotesize \((0.05)\)} & {\footnotesize \((0.05)\)} & {\footnotesize \((0.05)\)}\\
Probabilistic appeals
& \(0.07^{*}\) & \(0.07^{*}\) & \(0.07^{*}\) & \(0.07^{*}\)\\
& {\footnotesize \((0.04)\)} & {\footnotesize \((0.04)\)} & {\footnotesize \((0.04)\)} & {\footnotesize \((0.04)\)}\\
\addlinespace
Round
& No & Yes & Yes & Yes\\
Participant type
& No & No & Yes & Yes\\
Demographics + risk att.
& No & No & No & Yes\\
\midrule
Observations
& \(10{,}195\) & \(10{,}195\) & \(10{,}195\) & \(10{,}185\)\\
Matching groups
& \(59\) & \(59\) & \(59\) & \(59\)\\
\bottomrule
\end{tabular*}
\begin{tablenotes}
\footnotesize
\item \(^{***}\ p<0.01\); \(^{**}\ p<0.05\); \(^{*}\ p<0.10\). Coefficients are reported as average marginal effects from logit regressions. Standard errors are clustered at the matching-group level. Truth-telling equals one when the participant submits the complete induced preference ranking. The reference categories in Panels~A and B are the corresponding mechanisms without appeals. Panel~C reports \((\text{IA appeal}-\text{IA no appeals})-(\text{DA appeal}-\text{DA no appeals})\), so positive values indicate that appeals increase truth-telling more under IA. Each observation is a participant-round. Specification~(4) excludes ten participant-rounds with missing demographic or risk-attitude data. Matching groups are experimental sessions. Demographic controls are gender and age. Risk att. refers to risk attitudes.
\end{tablenotes}
\end{threeparttable}}
\end{table}

Under IA, strict and probabilistic appeals increase truth-telling by approximately \(11\) percentage points -- an effect that is highly significant and remains nearly unchanged across specifications. Under DA, the estimated
increases are only \(3\) to \(4\) percentage points and are at most marginally significant.

These estimates rely on cluster-robust standard errors computed from nine to twelve matching groups per arm, a range in which such standard errors are known to over-reject. Table~\ref{tab:nonparametric} (Appendix~\ref{app:additional-results}) therefore repeats the comparisons
without distributional assumptions, treating each matching group as a single observation. The two approaches agree for probabilistic appeals under IA, where truth-telling rises by roughly \(10\) percentage points at the session level
(\(p=0.006\)). They do not agree for strict appeals under IA: the estimated shift is of similar magnitude (\(0.09\)), but with nine sessions per arm the rank test does not reject (\(p=0.122\)). The same holds for both DA appeal
rules (\(p=0.122\) and \(p=0.119\)). We therefore regard the effect of probabilistic appeals on truth-telling under IA as firmly established, and the effect of strict appeals as indicative but resting on the parametric specification.

Panel~C directly compares the appeal effects across mechanisms. The effect of strict appeals is \(8\) to \(9\) percentage points larger under IA than under DA, while the effect of probabilistic appeals is approximately \(7\) percentage
points larger. Both differences are marginally significant across specifications.
 
This comparison is a difference-in-differences and therefore has no direct rank-test counterpart: Mann--Whitney compares two samples and cannot test a difference between two differences. Panel~B of Table~\ref{tab:nonparametric}
reports the level comparison at each regime, which is informative but not the same quantity. There, IA and DA are indistinguishable without appeals (\(p=0.626\)) and under strict appeals (\(p=0.453\)), and differ under probabilistic appeals (\(p=0.027\)).

This asymmetry is consistent with the mechanism identified by the theory. Under IA, appeals reduce the risk associated with ranking an ambitious school first because an unjustified rejection may be corrected afterwards. Under DA, truthful reporting is already a dominant strategy, leaving less scope for appeal rights to change initial reports.

As a robustness check, we perform the same analysis when only the first reported school was the same as the induced first preference. The same qualitative conclusions emerge (see Appendix~\ref{app:additional-results}).

\begin{result}
The data support Hypothesis~\ref{hyp:truth}.A and provide partial support for Hypothesis~\ref{hyp:truth}.B. Under IA, both appeal rules raise truth-telling by approximately \(11\) percentage points in the regression estimates. At the
matching-group level this increase is confirmed for probabilistic appeals (\(p=0.006\)) but not for strict appeals (\(p=0.122\)), so the evidence for the strict rule is weaker than the regression alone suggests. Under DA, the effects are only \(3\) to \(4\) percentage points and are at most marginally significant.
\end{result}

\paragraph{Consequential non-truthful reports.}

Exact truth-telling treats every departure from the induced ranking alike, even when it has no effect on the participant's assignment. We therefore replace each non-truthful report with the participant's induced ranking, hold all other reports fixed, and rerun the first-stage mechanism. We classify a report as inconsequential if the assignment is unchanged, beneficial if the observed report yields more points, and harmful if it yields fewer points.\footnote{We evaluate the counterfactual at the first stage because a different report generally changes the set of available appeals. A final-outcome comparison would require assumptions about counterfactual appeal choices and, under the probabilistic rule, unobserved appeal draws.}

The contrast between the mechanisms is pronounced. Under DA, only about \(9\) percent of non-truthful reports are consequential without appeals, \(8\) percent under strict appeals, and \(7\) percent under probabilistic appeals. Every consequential DA report is harmful. Low truth-telling under DA therefore largely reflects reports that do not affect assignments.
Under IA, by contrast, almost two thirds of non-truthful reports are consequential in every treatment. Yet participants exploit IA's manipulability poorly: among consequential reports, about \(40\) percent are beneficial and \(60\) percent are harmful.

Appeals mainly affect whether participants misreport at all. They increase truth-telling under IA, but do not make the remaining non-truthful reporters more successful at manipulating the mechanism.

\subsection{Appeal take-up}\label{sec:res_appeals}

We next examine when participants make use of the appeal stage and whether they direct appeals toward claims that can improve their payoff.

%YH: I created a new plot and placed averages above the CIs for readibility.
\begin{figure}[htbp]
\centering
\includegraphics[width=\textwidth]{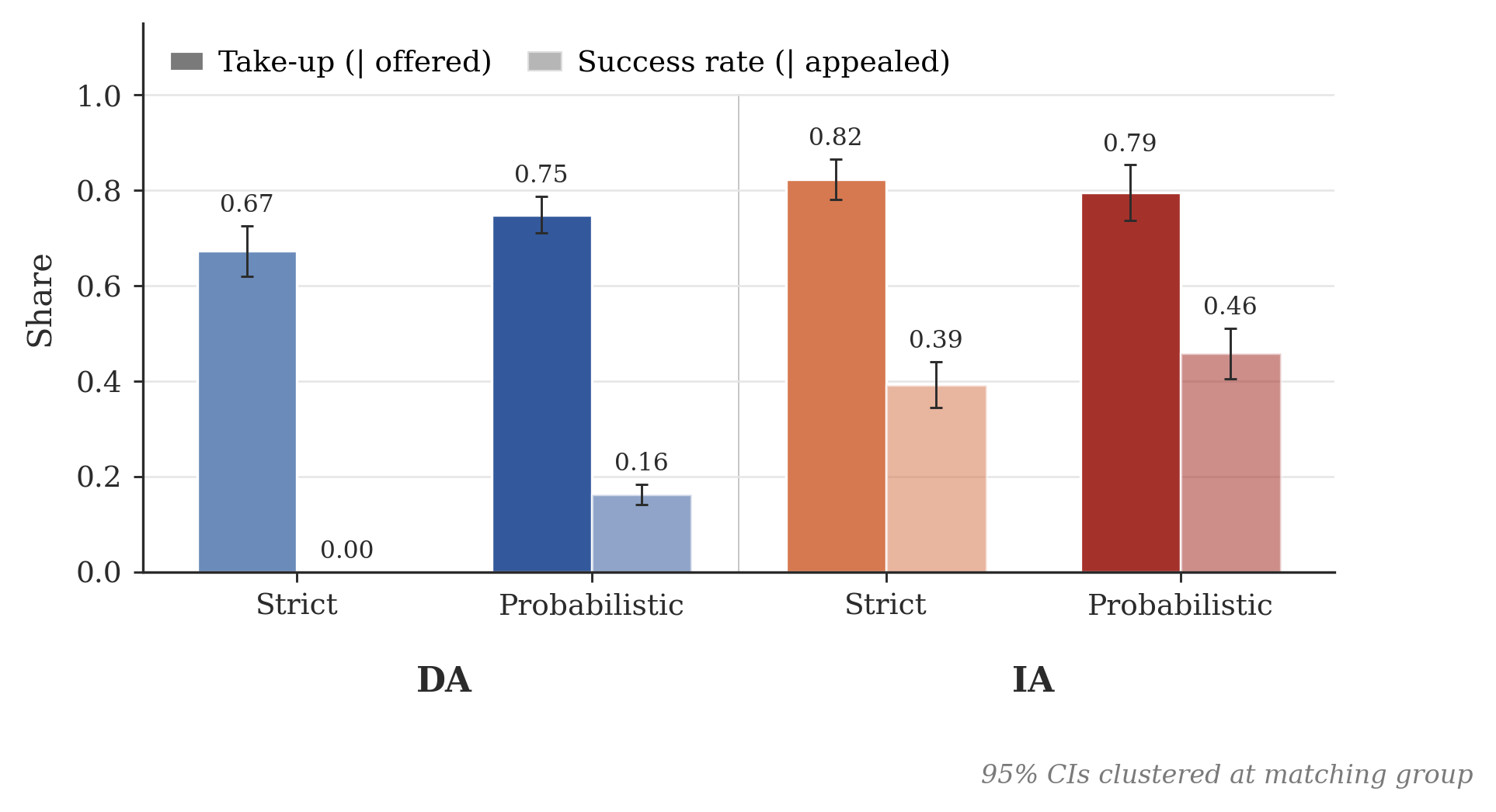}
\caption{Appeal take-up and success by treatment.\\ \footnotesize Dark bars show the share of eligible participant-rounds in which at least one appeal is filed. Light bars show the share of appealing participant-rounds in which at least one appeal succeeds. Whiskers are 95\% confidence intervals clustered at the matching-group level.}
\label{fig:appeal_combined}
\end{figure}

Appeal take-up is high in all four treatments. Participants file more appeals under DA than under IA, averaging between \(6\) and \(7\) appeals over ten rounds under DA compared with about \(4\) under IA, but this reflects the number of opportunities rather than a greater willingness to appeal. DA
produces worse first-stage assignments and therefore roughly twice as many school-specific appeal opportunities, approximately \(11\) per participant compared with \(6\) under IA. Conditional on being offered an opportunity,
take-up is in fact higher under IA. Under strict appeals, at least one appeal is filed in \(67\) percent of eligible participant-rounds under DA and \(82\) percent under IA (\(p=0.005\)). Under probabilistic appeals the corresponding rates are \(75\) and \(80\) percent, a difference the session-level test does not resolve
(\(p=0.103\)). The two observations are consistent: DA yields more appeals overall because it creates more occasions to appeal, while IA yields a higher rate of appeal per occasion.
 
The central difference between DA and IA is whether those appeals eventually work. Under strict appeals no appeal can succeed under DA, because DA creates no strict priority violation; this is a property of the design rather than an
estimated quantity, and we therefore do not test it. Under IA, at least one appeal succeeds in \(39\) percent of appealing participant-rounds. Under
probabilistic appeals the success rates are \(16\) percent under DA and \(46\) percent under IA (\(p<0.001\)). Appeals under DA therefore rely entirely on the random channel, while IA frequently creates priority claims that the appeal stage can correct. Overall, IA produces fewer appeals, but substantially more
effective ones.\footnote{Success is measured at the participant-round level and equals one when at least one filed appeal succeeds. Because a participant may file several probabilistic appeals, each with success probability \(0.1\), the probability that at least one succeeds can exceed \(10\) percent.}

These aggregate rates do not reveal whether participants select carefully among the appeal opportunities available to them. Table~\ref{tab:appeal_use} therefore examines which features of an appeal predict whether it is lodged. Specification~(1) is estimated at the participant-round level and considers whether at least one appeal is filed. Specification~(2) is estimated at the school-specific opportunity level and considers whether a particular rejection is appealed.

%YH: I think these specifications are too complex. I'd rather have a simple one where we regress only appeals on treatment, and a second one with controls. We would do this for both appeals and take-up, yielding four specifications.

\begin{table}[htbp]
\centering
\caption{Determinants of appeal use.}
\label{tab:appeal_use}
\resizebox{\textwidth}{!}{\begin{threeparttable}
\begin{tabular*}{\textwidth}{@{\extracolsep{\fill}} lcc @{}}
\toprule
& Any appeal & Appeal opportunity taken \\
& \((1)\) & \((2)\) \\
\midrule
DA, probabilistic appeals
& \(0.08^{**}\) & \(0.05\) \\
& {\footnotesize \((0.03)\)} & {\footnotesize \((0.03)\)} \\
IA, strict appeals
& \(0.11^{***}\) & \(0.03\) \\
& {\footnotesize \((0.04)\)} & {\footnotesize \((0.04)\)} \\
IA, probabilistic appeals
& \(0.08^{**}\) & \(0.05\) \\
& {\footnotesize \((0.04)\)} & {\footnotesize \((0.04)\)} \\
\addlinespace
First-stage assignment: second choice
& \(0.31^{***}\) & \(0.23^{**}\) \\
& {\footnotesize \((0.09)\)} & {\footnotesize \((0.09)\)} \\
First-stage assignment: third choice
& \(0.49^{***}\) & \(0.40^{***}\) \\
& {\footnotesize \((0.09)\)} & {\footnotesize \((0.08)\)} \\
First-stage assignment: unassigned
& \(0.51^{***}\) & \(0.44^{***}\) \\
& {\footnotesize \((0.10)\)} & {\footnotesize \((0.08)\)} \\
Strict priority claim
& \(-0.01\) & \(0.07^{***}\) \\
& {\footnotesize \((0.02)\)} & {\footnotesize \((0.03)\)} \\
Number of appeal opportunities
& \(0.11^{***}\) & \(\cdot\) \\
& {\footnotesize \((0.03)\)} & \\
Appeal target: second choice
& \(\cdot\) & \(-0.05^{***}\) \\
& & {\footnotesize \((0.01)\)} \\
Appeal target: third choice
& \(\cdot\) & \(-0.10^{***}\) \\
& & {\footnotesize \((0.02)\)} \\
\midrule
Round fixed effects
& Yes & Yes \\
Observations
& \(2{,}912\) & \(5{,}725\) \\
Matching groups
& \(41\) & \(41\) \\
\(R^2\)
& \(0.16\) & \(0.03\) \\
\bottomrule
\end{tabular*}
\begin{tablenotes}
\footnotesize
\item \(^{***}\ p<0.01\); \(^{**}\ p<0.05\); \(^{*}\ p<0.10\). Coefficients are reported from linear probability models. Standard errors are clustered at the matching-group level. The reference category is DA with strict appeals. Column~\((1)\) contains one observation for each participant-round with at least one appeal opportunity. Column~\((2)\) contains one observation for each school-specific appeal opportunity. Only the four appeal treatments are included, so the sample contains \(41\) matching groups. Matching groups are experimental sessions. First-stage assignments and appeal targets are ranked using the induced true preferences. In Column~\((1)\), strict priority claim indicates that at least one available appeal is supported by a strict priority violation.
\end{tablenotes}
\end{threeparttable}}
\end{table}

The first-stage outcome is the strongest predictor of appeal use. Relative to receiving the true first-choice school, receiving the second choice raises the probability of filing an appeal by \(0.31\), receiving the third choice raises it by \(0.49\), and remaining unassigned raises it by \(0.51\). Each additional appeal opportunity raises the probability of filing at least one appeal by \(0.11\). Participants therefore respond strongly to how poorly they fare in the first stage.

They also respond to the quality of individual claims. Conditional on being offered a particular appeal opportunity, a strict priority claim is \(7\) percentage points more likely to be pursued. Appeals to the true second- and third-choice schools are respectively \(5\) and \(10\) percentage points less likely to be lodged than appeals to the true first choice. Once these features are held fixed, the treatment indicators have little effect at the school-specific opportunity level. The aggregate differences across treatments therefore arise principally from the kinds of appeal opportunities the mechanisms generate.

Appeal use is nevertheless imperfect. Under DA with strict appeals, none of the \(1{,}845\) appeal opportunities can succeed, yet participants pursue \(1{,}083\) of them. Under IA with strict appeals, participants pursue \(166\) of \(233\) sure, payoff-improving claims, but also lodge \(465\) appeals with no chance of success. Under IA with probabilistic appeals, they pursue \(178\) of \(256\) sure gains and \(519\) of \(785\) positive-value lottery tickets. Participants are more likely to pursue probabilistic appeals as their expected payoff increases, but leave approximately three in ten sure gains unclaimed.\footnote{Appendix~\ref{app:full-takeup} imposes the full-take-up appeal strategy assumed by the model while holding submitted rank-order lists and first-stage assignments fixed. The results in the main text are
	preserved and sharpened. Under strict appeals, IA's average assigned rank
	improves from \(2.06\) to \(2.01\), while DA remains unchanged at \(2.74\);
	the DA--IA rank gap therefore widens from \(0.68\) to \(0.73\). Under
	probabilistic appeals, average rank improves from \(2.04\) to \(1.94\) under
	IA and from \(2.61\) to \(2.55\) under DA, so IA's rank advantage again
	widens. Full take-up also raises the share of IA market-rounds that strictly
	Pareto-dominate truthful DA from \(0.48\) to \(0.60\) under strict appeals
	and from \(0.46\) to \(0.63\) under probabilistic appeals. These additional
	efficiency gains are accompanied by substantially greater procedural
	stability, which rises to \(0.94\) under strict appeals and \(0.84\) under
	probabilistic appeals. Under the probabilistic rule, full take-up therefore
	makes IA both more rank-efficient and more procedurally stable than DA.}

\subsection{Efficiency}
\label{sec:res_efficiency}

We now test the paper's central welfare prediction. As Figure~\ref{fig:rank_welfare_combined} shows, strict appeals improve IA while leaving DA essentially unchanged. Our primary outcome is the rank of the final assignment in the participant's induced true preferences; rank~\(1\) denotes the most preferred school, and unassigned participants are coded as rank~\(4\). We also report welfare in experimental points in Appendix \ref{app:additional-results}.

%YH: I created a new plot and placed averages above the CIs for readibility.
\begin{figure}[htbp]
\centering
\includegraphics[width=\textwidth]{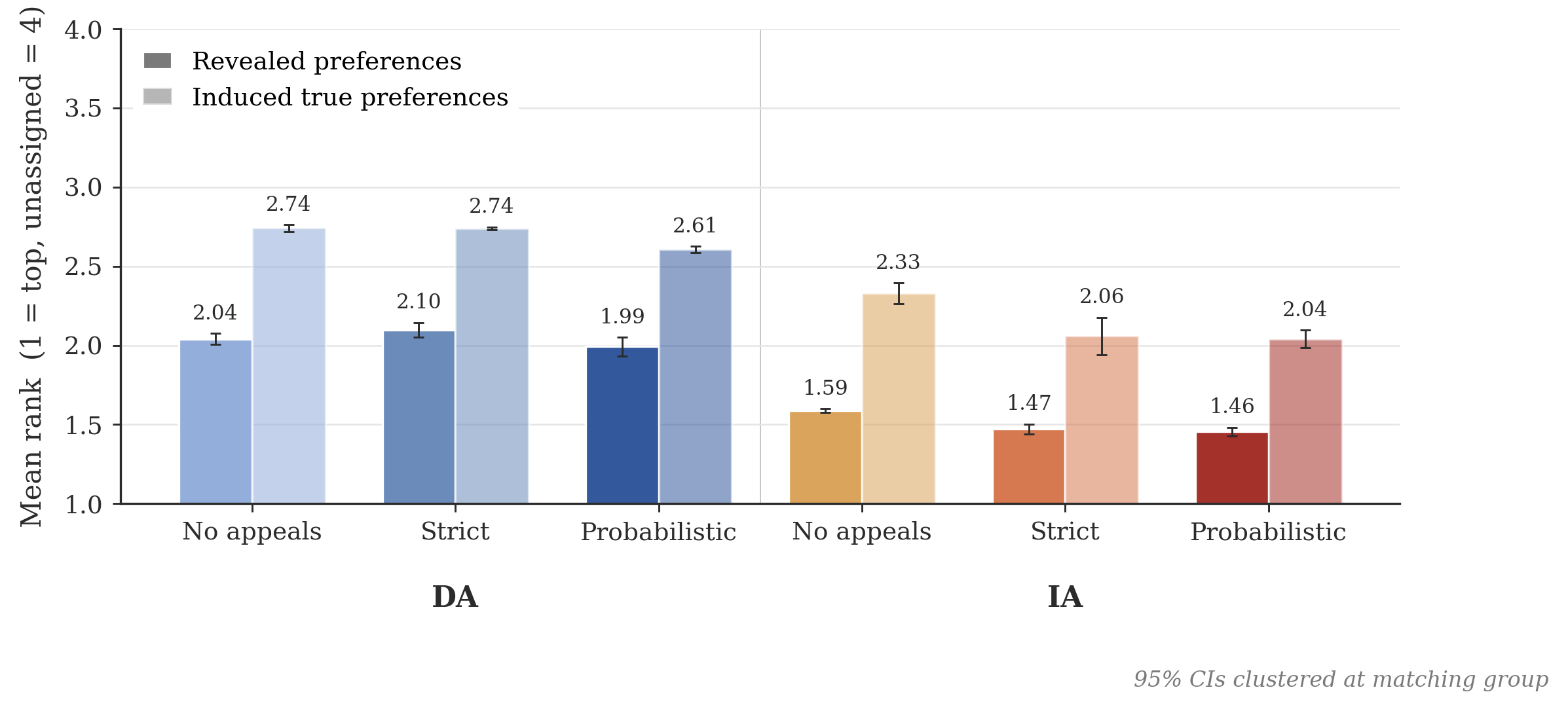}
\caption{Average assigned rank by treatment.\\
\footnotesize Assignments are evaluated against revealed and induced true preferences.
Unassigned participants are coded as rank~\(4\), and lower values indicate better
assignments. Whiskers are 95\% confidence intervals clustered at the
matching-group level.}
\label{fig:rank_welfare_combined}
\end{figure}

IA delivers better average ranks than DA in every treatment. This is not a reversal in the sign of the observed ranking, because IA already outperforms DA without appeals. The relevant comparative static is that strict appeals substantially widen IA's advantage. The average-rank gap between DA and IA rises from \(0.41\) without appeals to
\(0.68\) with strict appeals. We test this increase with a difference-in-differences regression of assigned rank on a mechanism indicator, a strict-appeals indicator and their interaction, estimated on the four corresponding arms with standard errors clustered at the matching group. The interaction is \(-0.27\) (standard error \(0.07\), \(p<0.001\)), so strict
appeals widen IA's rank advantage by \(0.27\) rank positions. This is the linear counterpart of Panel~C of Table~\ref{tab:efficiency_rank_appeals}, which reports the same quantity from the ordered logit.

Appeals also reduce the probability of remaining unassigned under IA (see Appendix~\ref{app:additional-results}). The share of unassigned participants falls from \(0.14\) without appeals to \(0.07\) with strict appeals and \(0.06\) with probabilistic appeals; under DA it remains at
\(0.14\) with strict appeals and falls to \(0.12\) with probabilistic appeals. All three reductions are significant at the matching-group level (\(p<0.001\)).
 
This margin requires a qualification. Six seats are available in a seven-person market, so exactly one participant is unassigned in every market-round in which the first-stage assignment is respected: the observed share is precisely
\(1/7=0.143\) under DA without appeals, DA with strict appeals, and IA without appeals, with no variation across sessions. The arms that fall below \(1/7\)
are exactly those in which a successful appeal seats a participant at the appealed school without displacing an occupant, so that school ends the round
over capacity. This occurs in \(400\) school-market-rounds, with up to four participants at a two-seat school, and only under DA with probabilistic
appeals, IA with strict appeals, and IA with probabilistic appeals. The reduction in the unassigned share therefore reflects seats added by the appeal stage rather than a better allocation of the original six, and we treat it as a
property of the appeal procedure rather than as an efficiency gain.

The level comparison differs from the no-appeals benchmark embedded in Hypothesis~\ref{hyp:reversal}. IA already produces better assignments than DA without appeals, with mean ranks of \(2.33\) and \(2.74\), respectively. Appeals therefore do not reverse the observed ranking between DA and IA. Instead, they widen an IA advantage that is already present in the experiment. This efficiency-improving effect is substantially stronger under strict and probabilistic appeals, because appeals recover part of the inefficiency obtained under IA without appeals while leaving DA broadly unchanged.

Table~\ref{tab:efficiency_rank_appeals} examines the robustness of these comparisons and reports adjusted differences in expected assigned rank from ordered logit models. The specifications follow those in Table~\ref{tab:truth_appeals}. Specification~(1) includes treatment indicators only. Specification~(2) adds round effects. Specification~(3) adds participant types. Specification~(4) adds demographics and risk attitudes. Standard errors are clustered at the matching-group level. The regression estimates closely match the descriptive differences. Under IA, strict and probabilistic appeals improve expected rank by approximately \(0.27\) and \(0.29\) rank points, respectively. Under DA, strict appeals have no detectable effect, while probabilistic appeals improve expected rank by approximately \(0.11\) to \(0.12\) rank points. The estimates are stable across specifications.

\begin{table}[htbp]
\centering
\caption{Impact of appeals on assigned rank.}
\label{tab:efficiency_rank_appeals}
\resizebox{\textwidth}{!}{\begin{threeparttable}
\begin{tabular*}{\textwidth}{@{\extracolsep{\fill}} lcccc @{}}
\toprule
& \((1)\) & \((2)\) & \((3)\) & \((4)\) \\
\midrule
\multicolumn{5}{l}{\textbf{Panel A. IA treatments}} \\
IA strict
& \(-0.27^{***}\) & \(-0.27^{***}\) & \(-0.27^{***}\) & \(-0.27^{***}\) \\
& {\footnotesize \((0.07)\)} & {\footnotesize \((0.07)\)} & {\footnotesize \((0.07)\)} & {\footnotesize \((0.07)\)} \\
IA probabilistic
& \(-0.29^{***}\) & \(-0.29^{***}\) & \(-0.29^{***}\) & \(-0.29^{***}\) \\
& {\footnotesize \((0.05)\)} & {\footnotesize \((0.05)\)} & {\footnotesize \((0.04)\)} & {\footnotesize \((0.04)\)} \\
\addlinespace
\multicolumn{5}{l}{\textbf{Panel B. DA treatments}} \\
DA strict
& \(0.00\) & \(0.00\) & \(0.00\) & \(0.00\) \\
& {\footnotesize \((0.01)\)} & {\footnotesize \((0.01)\)} & {\footnotesize \((0.01)\)} & {\footnotesize \((0.01)\)} \\
DA probabilistic
& \(-0.12^{***}\) & \(-0.12^{***}\) & \(-0.11^{***}\) & \(-0.11^{***}\) \\
& {\footnotesize \((0.01)\)} & {\footnotesize \((0.01)\)} & {\footnotesize \((0.01)\)} & {\footnotesize \((0.01)\)} \\
\addlinespace
\multicolumn{5}{l}{\textbf{Panel C. Change in rank gap between DA and IA relative to no appeals}} \\
Strict appeals
& \(0.26^{***}\) & \(0.26^{***}\) & \(0.27^{***}\) & \(0.27^{***}\) \\
& {\footnotesize \((0.07)\)} & {\footnotesize \((0.07)\)} & {\footnotesize \((0.07)\)} & {\footnotesize \((0.07)\)} \\
Probabilistic appeals
& \(0.16^{***}\) & \(0.16^{***}\) & \(0.18^{***}\) & \(0.18^{***}\) \\
& {\footnotesize \((0.05)\)} & {\footnotesize \((0.05)\)} & {\footnotesize \((0.05)\)} & {\footnotesize \((0.05)\)} \\
\addlinespace
Round effects
& No & Yes & Yes & Yes \\
Participant-type effects
& No & No & Yes & Yes \\
Demographics + risk att.
& No & No & No & Yes \\
\midrule
Observations
& \(10{,}195\) & \(10{,}195\) & \(10{,}195\) & \(10{,}185\) \\
Matching groups
& \(59\) & \(59\) & \(59\) & \(59\) \\
\bottomrule
\end{tabular*}
\begin{tablenotes}
\footnotesize
\item \(^{***}\ p<0.01\); \(^{**}\ p<0.05\); \(^{*}\ p<0.10\). Coefficients are reported as adjusted differences in expected assigned rank from ordered logit models. Lower values indicate better assignments, and unassigned participants are coded as rank~\(4\). The reference categories in Panels~A and B are the corresponding mechanisms without appeals. Panel~C reports \((\text{DA appeal}-\text{IA appeal})-(\text{DA no appeals}-\text{IA no appeals})\), so positive values indicate that appeals widen IA's rank advantage. Standard errors are clustered at the matching-group level. Each observation is a participant-round. Specification~(4) excludes ten participant-rounds with missing demographic or risk-attitude data. Matching groups are experimental sessions. Demographic controls are gender and age. Risk att. refers to risk attitudes.
\end{tablenotes}
\end{threeparttable}}
\end{table}

Panel~C shows that appeals significantly widen the rank gap between DA and IA. In Specification (4), which includes most controls, strict appeals increase the gap by \(0.27\) rank points, while probabilistic appeals increase it by \(0.18\) rank points. The strict rule produces the sharper contrast because it improves IA without generating the random gains that probabilistic appeals also create under DA.

\paragraph{Pareto dominance relative to truthful DA.} Theorem~\ref{thm:ia_r1_dominates_da} is an equilibrium-existence result: strict appeals allow IA to support an equilibrium outcome that weakly Pareto-dominates truthful DA, and strictly Pareto-dominates it whenever the two outcomes differ. Observed laboratory play need not coincide with the equilibrium constructed in the proof. We therefore use a direct allocation-level counterpart of the theorem. For every complete seven-person market-round, we compare the realized final assignment with the truthful DA benchmark. We say that the realized assignment weakly Pareto-dominates this benchmark if every participant is weakly better off according to her induced true preferences, and strictly Pareto-dominates it if, in addition, at least one participant is strictly better off.

Without appeals, observed IA strictly Pareto-dominates truthful DA in 27 market-rounds (\(11\) percent). With strict appeals, this frequency rises to 112 of 235 (\(48\) percent). Strict appeals therefore increase the frequency of strict Pareto dominance by \(36\) percentage points. We estimate this difference from a linear probability
model of the strict-dominance indicator on a strict-appeals indicator, estimated across the two IA arms with one observation per complete seven-person market-round and standard errors clustered at the matching group. The estimated difference is \(0.36\) with a standard error of \(0.046\) (\(p<0.001\)); the one-point discrepancy relative to the raw percentages is due to rounding.

This is a close empirical counterpart of Theorem~\ref{thm:ia_r1_dominates_da}. The theorem identifies a change in the outcomes that IA can sustain in equilibrium; the experiment shows that strict appeals move realized IA assignments sharply toward the Pareto region above truthful DA.

\paragraph{When do the rank improvements emerge?}

Appeal rights can affect final assignments at two points. Anticipating the appeal stage may change the rank-order lists participants submit and therefore
the assignment produced by the first-stage mechanism. Filed and upheld appeals may then generate a further improvement after the first-stage assignment has been announced. We therefore ask how much of the final rank difference is
already present before appeals are resolved and how much accrues subsequently.

For a first-stage mechanism \(M\in\{\da,\ia\}\) and an appeal rule
\(r\in\{r_1,r_2\}\), let
\(\bar R^{\mathrm{first}}_{M,r}\) and
\(\bar R^{\mathrm{final}}_{M,r}\) denote average assigned rank after the
first-stage and final assignments, respectively. Both are evaluated against
induced true preferences, with unassigned participants coded as rank~\(4\).
Lower values therefore indicate better assignments. Under no appeals,
\(\bar R^{\mathrm{final}}_{M,r_0}
=\bar R^{\mathrm{first}}_{M,r_0}\), so

\[
\underbrace{
\bar R^{\mathrm{final}}_{M,r_0}
-
\bar R^{\mathrm{final}}_{M,r}
}_{\text{total rank improvement}}
=
\underbrace{
\bar R^{\mathrm{first}}_{M,r_0}
-
\bar R^{\mathrm{first}}_{M,r}
}_{\text{first-stage difference}}
+
\underbrace{
\bar R^{\mathrm{first}}_{M,r}
-
\bar R^{\mathrm{final}}_{M,r}
}_{\text{subsequent appeal gain}}.
\]

This identity is a chronological accounting decomposition rather than a
causal mediation analysis. The first-stage difference measures how much of
the treatment difference is already present before any appeal is resolved.
In our design, this difference can operate only through changes in submitted
rank-order lists and the assignments they generate. The subsequent appeal
gain measures the improvement from the first-stage to the final assignment
among participants assigned to the appeal treatment. It is evaluated at the
reports those participants submitted when appeals were available and is
non-negative by construction. The decomposition therefore identifies when
the observed rank improvement emerges, but it does not provide a unique
causal allocation between reporting behavior and the appeal stage.

\begin{table}[htbp]
\centering
\caption{Decomposition of improvements in average assigned rank.}
\label{tab:rank_decomposition}
\begin{threeparttable}
\begin{tabular*}{\textwidth}{@{\extracolsep{\fill}}lccc@{}}
	\toprule
	Treatment
	& Total rank
	& First-stage
	& Subsequent appeal \\
	& change
	& change
	& change \\
	& \((1)\) & \((2)\) & \((3)\) \\
	\midrule
  DA, strict appeals
  & \(0.00\) & \(0.00\) & --- \\
  & \((0.01)\) & \((0.01)\) & \\
	
	DA, probabilistic appeals
	& \(0.14^{***}\) & \(0.02\) & \(0.12^{***}\) \\
	& \((0.01)\) & \((0.01)\) & \((0.01)\) \\
	
	\addlinespace
	IA, strict appeals
	& \(0.27^{***}\) & \(0.16^{**}\) & \(0.11^{***}\) \\
	& \((0.07)\) & \((0.06)\) & \((0.01)\) \\
	
	IA, probabilistic appeals
	& \(0.29^{***}\) & \(0.17^{***}\) & \(0.12^{***}\) \\
	& \((0.04)\) & \((0.04)\) & \((0.01)\) \\
	\bottomrule
\end{tabular*}
\begin{tablenotes}[flushleft]
	\footnotesize
	\item Positive entries indicate reductions in average assigned rank. Each observation is a participant-round, and matching groups are experimental sessions. Ranks are evaluated against induced true preferences, with
	unassigned participants coded as rank~\(4\). The entry marked --- is not estimated. No appeal can succeed under DA with strict appeals, so the first-stage and final assignments coincide in every market-round and the subsequent appeal gain is identically zero. The total improvement compares final rank under the indicated appeal rule with final rank under the same first-stage mechanism without appeals. The first-stage difference makes the corresponding comparison before appeals are resolved. The subsequent appeal gain is the within-treatment reduction in average rank between the first-stage and final assignments. The three point estimates in each row satisfy the accounting identity exactly. Standard errors, clustered at the matching-group level, are in parentheses.
	\(^{***}\ p<0.01\), \(^{**}\ p<0.05\), \(^{*}\ p<0.10\).
\end{tablenotes}
\end{threeparttable}
\end{table}

Table~\ref{tab:rank_decomposition} shows that the rank improvement under IA
emerges both before and after appeals are resolved. Strict appeals improve
final assigned rank by \(0.270\) positions. Of this improvement, \(0.163\)
positions are already present in the first-stage assignment, while the appeal
stage generates a further gain of \(0.107\) positions. Thus, approximately
60 percent of the observed improvement is visible before any appeal is
resolved.

Probabilistic appeals produce a similar pattern under IA. They improve final
assigned rank by \(0.288\) positions, of which \(0.167\) positions are present
at the first stage and \(0.121\) accrue subsequently. Under both appeal rules,
more than half of IA's final rank improvement is therefore already visible in
the first-stage assignment.

The contrast with DA is pronounced. Strict appeals leave both first-stage and final
rank essentially unchanged because DA produces no strict priority claim for
the appeal stage to correct. Probabilistic appeals improve final rank by
\(0.136\) positions, but almost all of this improvement, \(0.117\) positions,
arises after DA has run. The first-stage difference of \(0.019\) positions is
small and statistically insignificant.

\begin{result}
The data strongly support Hypothesis~\ref{hyp:reversal}. Strict appeals
improve IA substantially more than DA: they improve mean assigned rank
under IA by \(0.27\) positions while leaving DA unchanged, thereby widening
IA's rank advantage by \(0.27\) positions. They also raise the share of
realized IA outcomes that strictly Pareto-dominate truthful DA from \(11\)
to \(48\) percent. Approximately \(60\) percent of the observed IA rank
improvement is already present before appeals are resolved. Probabilistic
appeals produce the same directional pattern.
\end{result}

\subsection{Stability}
\label{sec:res_stability}

We finally examine whether IA's efficiency gains come at the cost of stability. The strict appeal rule is defined using the underlying priority classes, not the tie-break used by the first-stage mechanism. Accordingly, a strict blocking pair requires the student to have strictly higher underlying priority than an assigned student. Equal-priority students who lose the tie-break do not form strict blocking pairs. We treat wastefulness separately. A market-round is stable if it contains neither a strict blocking pair nor a wasted-seat claim.

Our main measure evaluates stability against \textit{submitted} preferences -- the natural procedural benchmark because both appeal eligibility and the panel's decision are based on the rank-order list submitted in the first stage. We also report stability against \textit{induced} true preferences -- a more demanding benchmark because a school that a participant places below her assignment cannot subsequently be claimed on appeal.

%YH: I created a new plot and placed averages above the CIs for readibility.
\begin{figure}[htbp]
\centering
\includegraphics[width=1\textwidth]{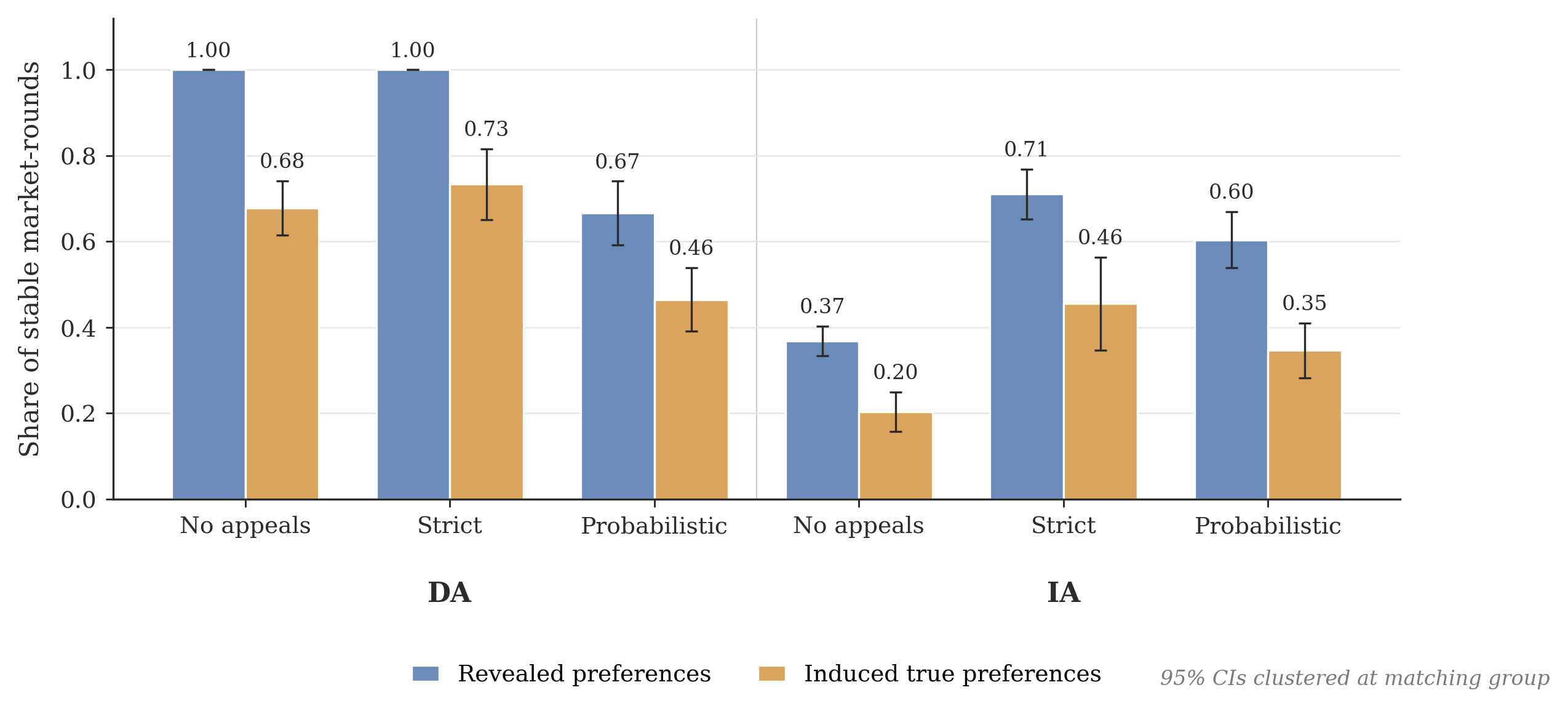}
\caption{Stability by treatment.\\ \footnotesize Bars report the share of stable market-rounds, evaluated against submitted and induced true preferences. A market-round is stable if it contains neither a strict blocking pair nor a wasted-seat claim. Equal-priority tie-break claims are not counted as strict blocking pairs. Whiskers are 95\% confidence intervals clustered at the matching-group level.}
\label{fig:stability}
\end{figure}

Figure~\ref{fig:stability} shows that strict appeals substantially improve IA's procedural stability. Without appeals, IA produces \(0.84\) strict blocking
pairs per market-round and a stable share of \(0.37\). Strict appeals reduce strict blocking pairs to \(0.33\), generate \(0.02\) wasted-seat claims per
market-round, and raise the stable share to \(0.71\). Probabilistic appeals also improve IA, but less cleanly: they reduce strict blocking pairs to
\(0.40\), generate \(0.08\) wasted-seat claims, and raise the stable share to \(0.60\). Every one of these changes is significant at the matching-group level (\(p<0.001\); Table~\ref{tab:nonparametric}, Panel~A).
 
The advantage of the strict rule over the probabilistic rule within IA is smaller but visible. Strict appeals yield a higher procedurally stable share (\(0.71\) against \(0.60\), \(p=0.029\)) and fewer wasted-seat claims
(\(p=0.015\)); the difference in blocking pairs against induced preferences is weaker (\(p=0.095\)) and the difference in mean assigned rank is not resolved (\(p=0.165\)). Strict appeals do create a small amount of waste relative to no appeals (\(p=0.032\)), so the accurate statement is that they create very little, not none.

These findings align closely with the logic of the appeal rules. Strict appeals target precisely the priority violations generated by IA and create
almost no waste. The experiment deliberately departs from the theoretical benchmark by requiring participants to choose which claims to pursue, whereas
the model assumes that every profitable appeal is submitted. Appendix~\ref{app:full-takeup} shows that imposing the model's appeal-stage strategy sharpens the main
findings. Holding submitted rank-order lists and first-stage assignments fixed, IA's remaining strict blocking pairs fall from \(0.33\) to \(0.04\) per
market-round, while its procedurally stable share rises from \(0.71\) to \(0.94\). Thus, the residual instability observed in the experiment does not
reflect a failure of the appeal rule. When appeal-stage behavior follows the theoretical benchmark, strict appeals eliminate almost all of IA's procedural
instability while preserving its substantial efficiency advantage. The small remaining instability reflects that the empirical measure is evaluated against
submitted preferences and also requires non-wastefulness, whereas Lemma~\ref{prop:pareto} concerns appealable strict claims involving schools that the student genuinely prefers.

The DA comparison is the reverse. DA without appeals contains no strict blocking pair and no waste when evaluated against submitted preferences, in every market-round of every session. Strict appeals leave this outcome unchanged, so every market-round remains stable. Because both arms are degenerate, no test is possible or needed, and the corresponding entries in
Table~\ref{tab:stability_effects} and Table~\ref{tab:nonparametric} are marked accordingly rather than reported as precisely estimated zeros. Probabilistic
appeals instead disturb an initially stable assignment: they generate \(0.31\) strict blocking pairs and \(0.78\) wasted-seat claims per market-round,
reducing the stable share to \(0.65\) (\(p<0.001\) for all three). The random channel can admit a participant without a priority claim and leave her
first-stage seat vacant. It therefore creates precisely the instability and waste that strict appeals avoid.

Table~\ref{tab:stability_effects} reports regression estimates at the market-round level. The specifications include treatment indicators and round fixed effects. Standard errors are clustered at the matching-group level. 
The regression estimates confirm the descriptive pattern. Under IA, strict appeals reduce strict blocking pairs by \(0.5\) per market-round and increase the probability of stability by \(0.3\). Probabilistic appeals also improve IA, but their effects are smaller and they create more waste. Under DA, strict appeals have no effect on any outcome. Probabilistic appeals instead increase both strict blocking pairs and wasted-seat claims and reduce stability by \(0.3\).

\begin{table}[htbp]
\centering
\caption{Impact of appeals on stability.}
\label{tab:stability_effects}
\resizebox{\textwidth}{!}{\begin{threeparttable}
\begin{tabular*}{\textwidth}{@{\extracolsep{\fill}} lccc @{}}
\toprule
\cmidrule(lr){2-4}
& Blocking pairs & Wasted-seat claims & Stable market-round \\
& \((1)\) & \((2)\) & \((3)\) \\
\midrule
\multicolumn{4}{l}{\textbf{Panel A. IA treatments}} \\
IA strict
& \(-0.52^{***}\) & \(0.03^{***}\) & \(0.34^{***}\) \\
& {\footnotesize \((0.06)\)} & {\footnotesize \((0.01)\)} & {\footnotesize \((0.03)\)} \\
IA probabilistic
& \(-0.45^{***}\) & \(0.08^{***}\) & \(0.23^{***}\) \\
& {\footnotesize \((0.06)\)} & {\footnotesize \((0.01)\)} & {\footnotesize \((0.04)\)} \\
\addlinespace
\multicolumn{4}{l}{\textbf{Panel B. DA treatments}} \\
DA strict
& --- & --- & --- \\
&  &  & \\
DA probabilistic
& \(0.31^{***}\) & \(0.78^{***}\) & \(-0.33^{***}\) \\
& {\footnotesize \((0.05)\)} & {\footnotesize \((0.09)\)} & {\footnotesize \((0.04)\)} \\
\addlinespace
Round effects
& Yes & Yes & Yes \\
\midrule
Observations
& \(1{,}435\) & \(1{,}435\) & \(1{,}435\) \\
Matching groups
& \(59\) & \(59\) & \(59\) \\
\bottomrule
\end{tabular*}
\begin{tablenotes}
\footnotesize
\item \(^{***}\ p<0.01\); \(^{**}\ p<0.05\); \(^{*}\ p<0.10\). Coefficients for Specifications~(1) and (2) are reported from OLS regression models. Coefficients for Specification~(3) are reported from a linear probability model. The reference categories are the corresponding mechanisms without appeals. Entries marked --- are not estimated. Submitted stability equals one and
both violation counts equal zero in every DA market-round without appeals and with strict appeals, so the outcome has no variation in either arm and the contrast is undefined rather than precisely zero. All outcomes are evaluated against submitted preferences. A strict blocking pair requires a strict difference in the underlying priority classes; equal-priority tie-break claims are excluded. A market-round is stable if it contains neither a strict blocking pair nor a wasted-seat claim. Each observation is a complete seven-person market-round; market-rounds affected by participant exits are excluded. Matching groups are experimental sessions. Standard errors are clustered at the matching-group level.
\end{tablenotes}
\end{threeparttable}}
\end{table}

Stability is lower when assignments are evaluated against induced true preferences. Under IA, the stable share rises from \(0.20\) without appeals to
\(0.46\) with strict appeals (\(p=0.001\)) and \(0.35\) with probabilistic appeals (\(p=0.009\)). Under DA it is \(0.68\) without appeals, \(0.73\) with strict appeals -- a difference the session-level test does not resolve (\(p=0.280\)) -- and \(0.47\) with probabilistic appeals (\(p=0.002\)).

\begin{result}
Strict appeals reduce IA's strict blocking pairs by roughly \(60\) percent, from \(0.84\) to \(0.33\) per market-round, and raise its procedurally stable share from \(0.37\) to \(0.71\), while generating very little waste. Probabilistic
appeals also improve IA, but create more waste and leave more blocking pairs. Under DA, strict appeals leave the stable first-stage outcome unchanged, whereas probabilistic appeals create both blocking pairs and waste. All comparisons are significant at the matching-group level.
\end{result}

\subsection{Trade-offs}
\label{sec:res_synthesis}

The results do not produce a complete ranking of the mechanisms. Instead, the consequences of appeals depend on the first-stage assignment rule. IA creates priority violations that appeals can correct, whereas DA begins from a procedurally stable assignment and therefore leaves strict appeals little to do.

Figure~\ref{fig:radar} summarizes truth-telling, rank-efficiency and stability, evaluated against induced true preferences. Each dimension is normalized by its across-treatment maximum, so that higher values indicate better performance. The figure is intended only as a visual comparison and does not constitute an aggregate welfare index.

\begin{figure}[htbp]
\centering
\includegraphics[width=0.8\textwidth]{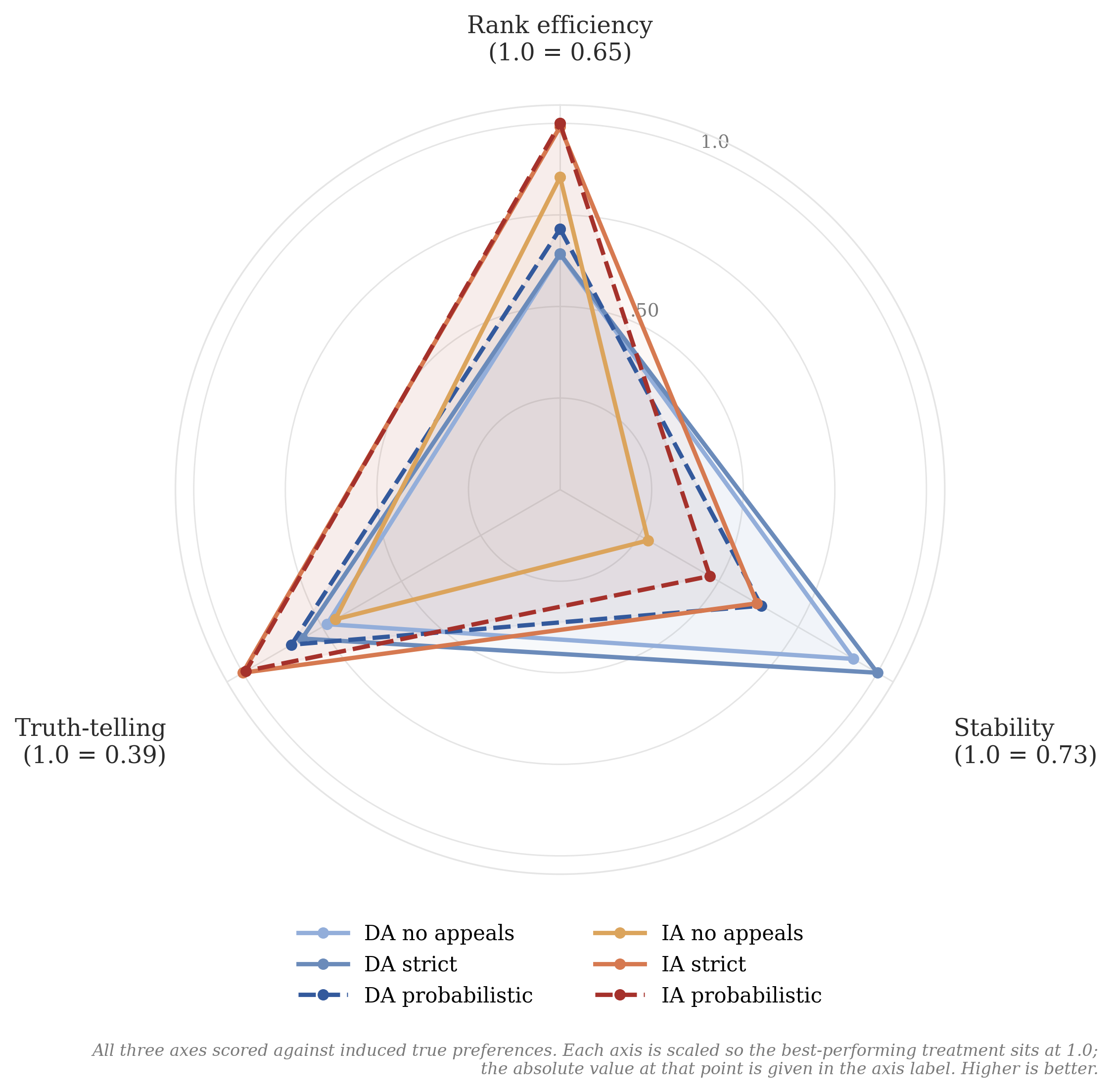}
\caption{Truth-telling, rank-efficiency and stability across treatments.\\ \footnotesize Rank-efficiency is based on assigned rank evaluated against induced true preferences. Stability is the share of market-rounds containing neither a strict blocking pair nor a wasted-seat claim, evaluated against induced true preferences. Each axis is normalized by its across-treatment maximum. The normalization is descriptive and does not assign welfare weights to the three outcomes.}
\label{fig:radar}
\end{figure}

\paragraph{Within IA.} Both appeal rules improve truth-telling, efficiency, and stability. Relative to no appeals, truth-telling rises from \(0.28\) to
approximately \(0.39\), while mean true-preference rank improves from \(2.33\) to \(2.06\) under strict appeals and \(2.04\) under probabilistic appeals. Procedural stability rises from \(0.37\) to \(0.71\) with strict appeals and
\(0.60\) with probabilistic appeals. The two appeal rules deliver efficiency and truth-telling gains that the session-level tests cannot distinguish (\(p=0.165\) and \(p=0.776\)), but strict appeals leave a higher stable share (\(p=0.029\)) and substantially less waste (\(p=0.015\)).
 
\paragraph{Between IA and DA.} IA with strict appeals performs substantially better than DA with strict appeals on efficiency (\(p<0.001\)), but the truth-telling difference is not resolved at the session level (\(p=0.453\)) and DA remains fully procedurally stable. Strict appeals nevertheless close more
than half of IA's initial stability deficit. Probabilistic appeals produce a different comparison: they improve IA along all three dimensions but undermine DA's principal advantage, reducing its procedural stability from \(1.00\) to \(0.65\). Under probabilistic appeals the two mechanisms are no longer distinguishable on procedural stability (\(p=0.420\)) or on blocking pairs (\(p=0.404\)), while IA retains its rank advantage (\(p<0.001\)).

The mechanism and the appeal rule must therefore be evaluated jointly. Strict appeals are particularly valuable under IA because they correct genuine priority violations while creating almost no waste. Under DA they leave the original assignment essentially unchanged. Probabilistic appeals are less targeted and can destabilize an otherwise stable DA outcome.

\section{Conclusion}
\label{sec:conclusion}

Appeals are part of the allocation mechanism, not an afterthought. Strict priority-based appeals leave Deferred Acceptance largely unchanged but make Immediate Acceptance less manipulable, improve student assignments, and reduce priority violations. Much of the experimental improvement appears before any appeal is decided because appeal rights change submitted rankings. Probabilistic appeals also improve Immediate Acceptance, but can create waste and instability after Deferred Acceptance. Assignment mechanisms should therefore be designed and evaluated jointly with the appeal rules that follow them.

\singlespacing
\setlength{\parskip}{-0.2em} 
\bibliographystyle{ecta}
\bibliography{bibliogr}

\newpage
\appendix
\bigskip
%Proofs omitted in main text
\section{Proofs Omitted in Main Text.\hfill}
\label{app:missingproofs}

\subsection{Proof of Lemma \ref{prop:pareto}.}

\begin{proof}
	Let
	\[
	\mu^0=M(\hat\succ)
	\qquad\text{and}\qquad
	\mu=(M\circ r_1)(\hat\succ).
	\]
	Weak Pareto dominance with respect to true preferences is immediate. Each
	student either keeps her first-stage assignment under $\mu^0$ or is reassigned
	through an upheld appeal to a school she truly prefers to it.
	
	We now show that $\mu$ admits no strict blocking pair $(i,s)$ such that
	$s\in A_i(\hat\succ,M)$ and $s\succ_i\mu_i$. Suppose, toward a contradiction,
	that such a pair exists. Since $\mu_i\succsim_i\mu_i^0$, we also have
	$s\succ_i\mu_i^0$. Therefore $s\in a_i$, so student $i$ appealed to $s$.
	
	Because $(i,s)$ is a strict blocking pair of $\mu$, there exists
	$k\in\mu_s^{-1}$ such that $i\rhd_s k$. If
	$k\in(\mu_s^0)^{-1}$, then $(i,s)$ was already a strict blocking pair of the
	initial matching $\mu^0$, and $r_1$ upheld student $i$'s appeal to $s$.
	
	Now suppose $k\notin(\mu_s^0)^{-1}$. Then $k$ was assigned to $s$ through a
	successful appeal. Hence, by definition of $r_1$, there exists an initial
	occupant $j\in(\mu_s^0)^{-1}$ such that $k\rhd_s j$. Since priorities are
	transitive and $i\rhd_s k$, we have $i\rhd_s j$. Thus $(i,s)$ was already a
	strict blocking pair of $\mu^0$, and again $r_1$ upheld student $i$'s appeal
	to $s$.
	
	In either case, $i$'s appeal to $s$ was upheld. Since final reassignment selects
	the student's most preferred school according to $\succ_i$ among her upheld
	appeals, $\mu_i\succsim_i s$, contradicting $s\succ_i\mu_i$.
\end{proof}

\subsection{Proof of Proposition \ref{thm:ia_manipulable}.}

We provide the full proof that we only sketched in the main text. 

\begin{proof}
	We first prove containment. Fix a preference profile
	$(\succ_i,\succ_{-i})$ and suppose student $i$ can profitably manipulate
	$\ia\circ r_1$ at this profile. Let $\hat\succ_i$ be a profitable deviation and write
	\[
	x=(\ia\circ r_1)_i(\succ_i,\succ_{-i})
	\qquad\text{and}\qquad
	y=(\ia\circ r_1)_i(\hat\succ_i,\succ_{-i}).
	\]
	Thus
	\[
	y\succ_i x.
	\]
	Since $r_1$ only adds successful appeals to the first-stage $\ia$ outcome and never worsens a student's assignment,
	\[
	x \succsim_i (\ia\circ r_0)_i(\succ_i,\succ_{-i}).
	\]
	
	We show that student $i$ can also obtain $y$ under $\ia\circ r_0$ by a suitable misreport.
	
	If student $i$ already obtains $y$ in the first-stage $\ia$ outcome under $(\hat\succ_i,\succ_{-i})$, then the same report $\hat\succ_i$ gives student $i$ school $y$ under $\ia\circ r_0$. Therefore
	\[
	(\ia\circ r_0)_i(\hat\succ_i,\succ_{-i})
	=
	y
	\succ_i
	x
	\succsim_i
	(\ia\circ r_0)_i(\succ_i,\succ_{-i}),
	\]
	so $\hat\succ_i$ is a profitable manipulation of $\ia\circ r_0$.
	
	Suppose instead that student $i$ obtains $y$ only through a successful appeal under $r_1$. Let $\mu$ be the first-stage $\ia$ outcome under $(\hat\succ_i,\succ_{-i})$. Since $i$ successfully appeals to $y$, there is a first-stage occupant $j$ of $y$ in $\mu$ such that
	\[
	i \rhd_y j.
	\]
	
	Let $\hat\succ_i'$ be any report that ranks $y$ first. We claim that student $i$ obtains $y$ in the first stage of $\ia$ under $(\hat\succ_i',\succ_{-i})$.
	
	Let $A_y^1$ be the set of students other than $i$ who rank $y$ first. This set
	depends only on $\succ_{-i}$. If fewer than $q_y$ students in $A_y^1$ have
	higher priority at $y$ than $i$, then $i$ is among the $q_y$ highest-priority
	round-1 applicants to $y$ when she ranks $y$ first. Hence $i$ is admitted to
	$y$ in round 1.
	
	So suppose, by way of contradiction, that at least $q_y$ students in $A_y^1$
	have higher priority at $y$ than $i$. Then, independently of $i$'s report, the
	$q_y$ seats of $y$ are filled in round 1 by students who all have
	higher priority than $i$. Since IA assignments are permanent, every
	first-stage occupant of $y$ in $\mu$ has higher priority than $i$.
	This contradicts the existence of a first-stage occupant $j$ of $y$ with
	\[
	i \rhd_y j,
	\]
	since the tie-break only resolves priority ties and therefore cannot reverse a
	strict priority comparison.
	
	Therefore fewer than $q_y$ students in $A_y^1$ have higher priority at $y$
	than $i$, and $i$ is admitted to $y$ in round 1 when she ranks $y$ first.
	Thus
	\[
	(\ia\circ r_0)_i(\hat\succ_i',\succ_{-i})=y.
	\]
	Consequently,
	\[
	(\ia\circ r_0)_i(\hat\succ_i',\succ_{-i})
	=
	y
	\succ_i
	x
	\succsim_i
	(\ia\circ r_0)_i(\succ_i,\succ_{-i}).
	\]
	Hence $\hat\succ_i'$ is a profitable manipulation of $\ia\circ r_0$.
	
	Thus, whenever student $i$ can profitably manipulate $\ia\circ r_1$ at a
	profile, the same student can profitably manipulate $\ia\circ r_0$ at that
	profile. Therefore $\ia\circ r_0$ is as manipulable as $\ia\circ r_1$.

	It remains to show that the containment is strict.
	Example~\ref{ex:moremanipulable} in the main text gives a preference profile
	that is manipulable under $\ia\circ r_0$ but not under $\ia\circ r_1$.
	Therefore $\ia\circ r_0$ is more manipulable than $\ia\circ r_1$.
\end{proof}

\newpage

%Observational study
%\section{Observational Study}
%\label{sec:observational_study}
%\input{obsdata}
%\newpage

%r2
\section{Analyzing $r_2$.\hfill}
\label{app:randomized}

We now consider the probabilistic appeal rule \(r_2\). Since \(r_2\) is
random, the outcome of a mechanism followed by \(r_2\) is a lottery over
final matchings rather than a deterministic matching. We therefore use
stochastic dominance to define incentive and welfare comparisons.

\subsection{Definitions}

\paragraph{Random matchings}
Let \(\mathcal{M}\) denote the set of possibly infeasible matchings.
A \emph{randomized matching} is a probability distribution
\[
\lambda \in \Delta(\mathcal{M}).
\]
For a student \(i\), the \emph{marginal lottery} induced by \(\lambda\)
is denoted by \(\lambda_i\), where
\[
\lambda_i(s)
=
\sum_{\mu \in \mathcal{M}:\, \mu_i = s} \lambda(\mu)
\qquad
\text{for each } s \in S \cup \{s_\emptyset\}.
\]

For two marginal lotteries \(\lambda_i\) and \(\lambda_i'\), we write
\[
\lambda_i \succeq_i^{sd} \lambda_i'
\]
if \(\lambda_i\) \emph{first-order stochastically dominates}
\(\lambda_i'\) with respect to \(\succ_i\): for every school
\(s \in S \cup \{s_\emptyset\}\),
\[
\sum_{t:\, t \succeq_i s} \lambda_i(t)
\ge
\sum_{t:\, t \succeq_i s} \lambda_i'(t).
\]
Equivalently, \(\lambda_i \succeq_i^{sd} \lambda_i'\) if every
von Neumann -- Morgenstern utility representation \(u_i\) of \(\succ_i\)
gives weakly higher expected utility under \(\lambda_i\) than under
\(\lambda_i'\).

For two randomized matchings \(\lambda\) and \(\lambda'\), we write
\[
\lambda \succeq^{sd} \lambda'
\]
if
\[
\lambda_i \succeq_i^{sd} \lambda_i'
\]
for every student \(i\). We write
\(\lambda_i \succ_i^{sd} \lambda_i'\) if
\(\lambda_i \succeq_i^{sd} \lambda_i'\) and the inequality defining
stochastic dominance is strict for at least one upper contour set.

\paragraph{Strategy-proofness and equilibrium for randomized mechanisms.}
A randomized mechanism \(\Phi\) maps each reported preference profile
\(\sigma\) to a lottery
\[
\Phi(\sigma) \in \Delta(\mathcal{M}).
\]
It is \emph{sd-strategy-proof} if, for every student \(i\), every true
preference profile \(\succ\), and every report \(\hat\sigma_i\),
\[
\Phi_i(\succ_i, \succ_{-i})
\succeq_i^{sd}
\Phi_i(\hat\sigma_i, \succ_{-i}).
\]
A report profile \(\sigma^*\) is an \emph{sd-Nash equilibrium} of
\(\Phi\) if there is no student \(i\) and no unilateral report
\(\hat\sigma_i\) such that
\[
\Phi_i(\hat\sigma_i,\sigma^*_{-i})
\succ_i^{sd}
\Phi_i(\sigma^*).
\]
This is the ordinal equilibrium notion appropriate for the present model:
it rules out deviations that are unambiguously profitable for every
von Neumann--Morgenstern utility representation of \(\succ_i\).

\paragraph{The probabilistic appeal rule \(r_2\).}
Fix \(p \in (0,1)\). Given a first-stage matching \(\mu\) and a reported
preference profile \(\sigma\), the rule \(r_2\) generates a lottery over
final matchings as follows. A student \(i\) may appeal only to schools
that she ranked above her first-stage match \(\mu_i\) under \(\sigma_i\)
and from which she was rejected. For every admissible appeal by student
\(i\) to school \(s\):
\begin{itemize}
	\item if \((i,s)\) is a strict blocking pair in \(\mu\), that is, if
	\(i \rhd_s j\) for some \(j \in \mu_s^{-1}\), the appeal is upheld
	with probability one;
	
	\item otherwise, the appeal is upheld independently with probability
	\(p\).
\end{itemize}
Thus \(r_2\) nests \(r_1\): every appeal upheld under \(r_1\) is upheld
under \(r_2\) with certainty, while appeals not upheld under \(r_1\) are
upheld independently with probability \(p\). 
For each realization of appeal outcomes, student \(i\) is assigned to the
highest-ranked school under her true preference \(\succ_i\) among the upheld
appeals, if any, and otherwise remains at her first-stage match \(\mu_i\).
The resulting lottery over final matchings is denoted
\[
r_2(\mu,\sigma).
\]

For any first-stage mechanism \(M\), we write
\[
(M \circ r_2)(\sigma)=r_2(M(\sigma),\sigma).
\]

\subsection{Welfare Ordering of Appeal Rules}

The probabilistic rule nests the strict rule: it preserves every appeal
that succeeds under \(r_1\) and may uphold additional appeals. Since the
student chooses her most-preferred upheld claim according to her true
preference, the comparison is pointwise and therefore also holds in
stochastic-dominance terms.

\begin{proposition}
	\label{prop:r2_dominates_r1}
	For any mechanism \(M\), any true preference profile \(\succ\), and any
	reported preference profile \(\sigma\),
	\[
	(M\circ r_2)(\sigma)
	\succeq^{sd}
	(M\circ r_1)(\sigma).
	\]
\end{proposition}

\begin{proof}
	Fix a student \(i\), and let \(\mu=M(\sigma)\) be the first-stage
	matching. Under \(r_1\), student \(i\) keeps \(\mu_i\) or moves to her
	most-preferred school under \(\succ_i\) among the appeals supported by a
	strict blocking pair. Call this deterministic outcome \(a_i\).

	Under \(r_2\), every appeal upheld under \(r_1\) is upheld with
	probability one, while additional admissible appeals may also be upheld.
	Thus, in every realization of the \(r_2\) lottery, the set from which
	student \(i\)'s final assignment is selected contains \(a_i\), possibly
	together with schools she prefers to \(a_i\). Her realized assignment
	under \(r_2\) is therefore weakly preferred to \(a_i\) according to
	\(\succ_i\). Hence
	\[
	(M\circ r_2)_i(\sigma)
	\succeq_i^{sd}
	(M\circ r_1)_i(\sigma).
	\]
	Since \(i\) was arbitrary, the result follows.
\end{proof}

Because random appeals may be upheld without a priority claim, the
realized matchings under \(r_2\) need not contain fewer blocking pairs.

\subsection{A Stochastic Equilibrium Reversal}

The reversal result also has an equilibrium counterpart under probabilistic
appeals. The relevant equilibrium notion is sd-Nash equilibrium, because the
model specifies ordinal preferences but not a particular cardinal utility
representation.

Let
\[
\mu=\da(\succ).
\]
For each student \(i\), define the \emph{DA-cutoff report}
\(\sigma_i^\mu\) as follows. If \(\mu_i\in S\), student \(i\) lists, in
her true order, every school that she weakly prefers to \(\mu_i\), and
truncates her list immediately after \(\mu_i\). If
\(\mu_i=s_\emptyset\), she lists every school she prefers to being
unassigned, again in her true order. Let
\[
\sigma^\mu=(\sigma_i^\mu)_{i\in I}.
\]

\begin{theorem}
	\label{thm:ia_r2_dominates_da}
	For every school choice problem with strict priorities and every
	\(p\in(0,1)\), the DA-cutoff profile \(\sigma^\mu\) is an sd-Nash
	equilibrium of \(\ia\circ r_2\). Moreover,
	\[
	(\ia\circ r_2)(\sigma^\mu)
	\succeq^{sd}
	\da(\succ),
	\]
	where the deterministic DA assignment is viewed as a degenerate
	randomized matching. If the two randomized matchings differ, at least one
	student is strictly better off in stochastic-dominance terms.
\end{theorem}

\begin{proof}
	Write
	\[
	\nu=\ia(\sigma^\mu)
	\qquad\text{and}\qquad
	x=(\ia\circ r_1)(\sigma^\mu).
	\]
	We first show that
	\[
	x_i\succeq_i\mu_i
	\qquad\text{for every student }i.
	\]

	If \(\nu_i\neq s_\emptyset\), this is immediate: the DA-cutoff report
	contains no school that \(i\) ranks below \(\mu_i\). Now suppose that
	\(\nu_i=s_\emptyset\) and \(\mu_i=s\in S\). Student \(i\) applied to
	\(s\) and was rejected. School \(s\) must be full under \(\mu\). Indeed,
	if it had an empty DA seat, non-wastefulness of DA would imply that no
	student assigned elsewhere under \(\mu\) ranks \(s\) above her DA
	assignment. Hence only students in \(\mu_s^{-1}\) could apply to \(s\)
	under the DA-cutoff profile, and \(i\) could not be rejected.

	Since \(s\) is full under \(\mu\), it has \(q_s\) DA assignees, one of
	whom is \(i\). The \(q_s\) first-stage occupants of \(s\) under \(\nu\)
	exclude \(i\), so at least one of them, say \(j\), is not assigned to
	\(s\) under \(\mu\). Because \(j\) applies to \(s\) under her DA-cutoff
	report,
	\[
	s\succ_j\mu_j.
	\]
	Stability of \(\mu\) and strict priorities imply that every DA assignee
	at \(s\), including \(i\), has higher priority at \(s\) than \(j\).
	Thus \((i,s)\) is a strict blocking pair of \(\nu\), and \(r_1\)
	upholds \(i\)'s appeal to \(s\). Hence \(x_i\succeq_i\mu_i\).

	Next observe that every school \(t\succ_i x_i\) is an admissible appeal
	at \(\sigma^\mu\). The report \(\sigma_i^\mu\) follows the true order,
	and \(t\succ_i x_i\succeq_i\nu_i\), so \(i\) applied to \(t\) before
	reaching her first-stage assignment and was rejected. Moreover, the
	appeal to \(t\) is not supported by a strict blocking pair; otherwise
	\(r_1\) would assign \(i\) to \(t\) or to a still better school,
	contradicting the definition of \(x_i\). Under \(r_2\), student \(i\)
	therefore has an independent probability-\(p\) appeal ticket to every
	school she strictly prefers to \(x_i\), while \(x_i\) is guaranteed.

	We now show that \(\sigma^\mu\) is an sd-Nash equilibrium. Fix student
	\(i\) and a unilateral deviation \(\hat\sigma_i\). Let
	\[
	\hat\nu_i
	=
	\ia_i(\hat\sigma_i,\sigma^\mu_{-i})
	\qquad\text{and}\qquad
	\hat x_i
	=
	(\ia\circ r_1)_i(\hat\sigma_i,\sigma^\mu_{-i})
	\]
	denote the first-stage and strict-rule outcomes under the deviation.

	First suppose that \(\hat x_i\preceq_i x_i\). For any upper contour set
	whose cutoff is weakly below \(x_i\), the equilibrium lottery assigns
	probability one to that set. For an upper contour set strictly above
	\(x_i\), the deviation can enter the set only through non-strict
	probability-\(p\) appeals: a first-stage assignment or a successful
	strict appeal in that set would imply
	\(\hat x_i\succ_i x_i\). At \(\sigma^\mu\), by contrast, \(i\) holds a
	probability-\(p\) ticket to every school in the upper contour set.
	Independence of the appeal draws therefore implies that the equilibrium
	probability of reaching each upper contour set is weakly greater than
	under the deviation. Thus the equilibrium lottery sd-dominates the
	deviation lottery.

	Now suppose that \(\hat x_i\succ_i x_i\). We claim that at least one
	school \(t\succ_i\hat x_i\) is not an admissible appeal under the
	deviation. Suppose, to the contrary, that every such school is
	appealable.

	If \(\hat x_i=\hat\nu_i\), then every school preferred to \(\hat x_i\)
	must be ranked before \(\hat x_i\) and must reject \(i\). Reordering
	applications at which \(i\) is rejected does not change any other
	student's IA history: in each such round the same applicants are
	accepted as if \(i\) had not applied. Student \(i\) therefore reaches
	\(\hat x_i\) no earlier than under the true-order DA-cutoff report and
	faces the same or a fuller school. But she is rejected from
	\(\hat x_i\) under \(\sigma_i^\mu\), because
	\(\hat x_i\succ_i x_i\). Hence she cannot obtain \(\hat x_i\) directly.

	If instead \(\hat x_i\) is obtained through a strict appeal, student
	\(i\) is rejected from \(\hat x_i\) in the first stage. Applying to
	\(\hat x_i\) later cannot create a new strict claim because IA
	acceptances are permanent. If she applies earlier, then either the
	school is already full, in which case any lower-priority occupant would
	also be present when she applies under \(\sigma_i^\mu\), or its
	remaining seats are filled in that round by applicants who outrank her.
	In neither case can moving the application earlier turn the non-strict
	rejection at \(\sigma_i^\mu\) into a strict blocking pair. Thus
	\(\hat x_i\) cannot be obtained through a new strict appeal either.

	This contradiction proves the claim. Let
	\[
	T_i(\hat x_i)=\{t\in S:t\succ_i\hat x_i\}.
	\]
	At the equilibrium profile, every school in \(T_i(\hat x_i)\) is an
	independent probability-\(p\) appeal ticket. Under the deviation, at
	least one of these tickets is missing, and no school in
	\(T_i(\hat x_i)\) can be obtained directly or through a strict appeal
	by the definition of \(\hat x_i\). Hence
	\[
	\Pr_{\;(\ia\circ r_2)_i(\sigma^\mu)}
	\bigl(T_i(\hat x_i)\bigr)
	=
	1-(1-p)^{|T_i(\hat x_i)|}
	\]
	is strictly larger than the corresponding probability under the
	deviation, which is at most
	\[
	1-(1-p)^{|T_i(\hat x_i)|-1}.
	\]
	The deviation therefore cannot stochastically dominate the equilibrium
	lottery. Since the deviation was arbitrary, \(\sigma^\mu\) is an
	sd-Nash equilibrium.

	Finally, \(x_i\succeq_i\mu_i\) for every \(i\), and \(r_2\) preserves
	every assignment obtained under \(r_1\) while possibly adding better
	random outcomes. It follows that
	\[
	(\ia\circ r_2)(\sigma^\mu)
	\succeq^{sd}
	\mu
	=
	\da(\succ).
	\]
	If the two randomized matchings differ, at least one marginal
	stochastic-dominance comparison is strict.
\end{proof}

\subsection{DA Incentive Properties Remain}

\begin{proposition}
	\label{prop:da_r2_sp}
	For every \(p \in (0,1)\), \(\da \circ r_2\) is sd-strategy-proof.
\end{proposition}

\begin{proof}
	Fix student \(i\), true preferences \(\succ_i\), and reports
	\(\succ_{-i}\) of the other students. Let
	\[
	\mu = \da(\succ_i,\succ_{-i})
	\quad \text{and} \quad
	\nu = \da(\hat\sigma_i,\succ_{-i})
	\]
	be the truthful and misreported DA outcomes. Since DA is
	strategy-proof,
	\[
	\mu_i \succeq_i \nu_i.
	\]
	
	Let
	\[
	\Lambda = (\da \circ r_2)(\succ_i,\succ_{-i})
	\quad \text{and} \quad
	\widehat\Lambda = (\da \circ r_2)(\hat\sigma_i,\succ_{-i})
	\]
	be the induced lotteries. We show
	\[
	\Lambda_i \succeq_i^{sd} \widehat\Lambda_i.
	\]
	
	Fix a school \(s\), and let
	\[
	U_i(s)=\{t \in S \cup \{s_\emptyset\}: t \succeq_i s\}.
	\]
	
	\emph{Case 1: \(\mu_i \in U_i(s)\).}
	Truthful reporting gives student \(i\) a first-stage match in
	\(U_i(s)\). Since appeals can only move a student to schools ranked
	above her first-stage match, every realization under truthful reporting
	places \(i\) in \(U_i(s)\). Hence
	\[
	\Pr_{\Lambda_i}(U_i(s))=1
	\ge
	\Pr_{\widehat\Lambda_i}(U_i(s)).
	\]
	
	\emph{Case 2: \(\mu_i \notin U_i(s)\).}
	Every school in \(U_i(s)\) is strictly preferred by \(i\) to \(\mu_i\).
	Under truthful DA, student \(i\) applies to every such school before
	reaching \(\mu_i\), and is rejected from each of them. Hence every
	school in \(U_i(s)\) is an admissible appeal school under truthful
	reporting.
	
	Under the misreport, student \(i\) can receive a school in \(U_i(s)\)
	after the appeal stage only if some admissible appeal to a school in
	\(U_i(s)\) is upheld. Every such school is also admissible under
	truthful reporting, by the previous paragraph. Thus the set of
	truthful admissible appeals to schools in \(U_i(s)\) contains the set
	of misreport-induced admissible appeals to schools in \(U_i(s)\).
	
	Moreover, no such appeal is upheld with probability one under either
	report. Indeed, DA is stable with respect to the reported preference
	profile and the priority order used by the mechanism. Thus, whenever
	student \(i\) is rejected from a school \(t\) that she ranked above her
	DA match, \((i,t)\) is not a strict blocking pair of the DA outcome.
	Consequently, under \(r_2\), any admissible appeal by \(i\) to such a
	school is upheld with probability \(p\), and not with probability one.
	
	Since truthful reporting gives \(i\) a superset of the probability-\(p\)
	appeal tickets to schools in \(U_i(s)\), the probability that at least
	one appeal to a school in \(U_i(s)\) is upheld is weakly higher under
	truthful reporting than under the misreport. Therefore
	\[
	\Pr_{\Lambda_i}(U_i(s))
	\ge
	\Pr_{\widehat\Lambda_i}(U_i(s)).
	\]
	
	Since this holds for every upper contour set,
	\[
	\Lambda_i \succeq_i^{sd} \widehat\Lambda_i.
	\]
	Because \(i\) and \(\hat\sigma_i\) were arbitrary,
	\(\da \circ r_2\) is sd-strategy-proof.
\end{proof}

The proof reveals two distinct channels through which truthful reporting
dominates under \(\da \circ r_2\). First, DA's strategy-proofness
ensures a weakly better first-stage outcome. Second, truthful reporting
generates weakly more appeal opportunities, since listing a school above
one's match is necessary for filing an appeal, and truthful reporting
lists every truly preferred school. These channels reinforce each other:
the truthful report maximizes both the fallback, namely the DA match,
and the set of lottery tickets.

\subsection{IA Incentives under Probabilistic Appeals}

Unlike DA, IA remains manipulable under \(r_2\). However, probabilistic
appeals reduce the gain from standard IA manipulations. Under \(r_2\),
a rejected application to a desired school is not wasted: even absent a
priority claim, it creates a lottery ticket. We formalize this effect in
two propositions.

\begin{proposition}
	\label{prop:ia_r2_lottery_truth}
	Fix \(p \in (0,1)\), a preference profile \(\succ\), a student \(i\),
	and a von Neumann -- Morgenstern utility representation \(u_i\) of
	\(\succ_i\). Consider two reports for \(i\). Under report \(\sigma_i\),
	student \(i\) applies to school \(b\), is rejected, and receives
	first-stage match \(a\), where
	\[
	b \succ_i a,
	\]
	and the appeal to \(b\) is admissible under \(r_2\). Under report
	\(\hat\sigma_i\), student \(i\) also receives \(a\) in the first stage
	but does not generate an admissible appeal to \(b\). Holding fixed all
	other appeal opportunities, let \(Z_i\) denote the random final assignment
	generated by the first-stage match \(a\) and these other appeal
	opportunities. Then:
	\begin{enumerate}
		\item if the appeal to \(b\) is upheld under \(r_1\), that is, if
		\((i,b)\) is a strict blocking pair, the expected utility gain from
		\(\sigma_i\) over \(\hat\sigma_i\) is
		\[
		\mathbb{E}\left[\max\{u_i(b)-u_i(Z_i),0\}\right];
		\]
		
		\item if the appeal to \(b\) is not upheld under \(r_1\), the
		expected utility gain is
		\[
		p\,\mathbb{E}\left[\max\{u_i(b)-u_i(Z_i),0\}\right].
		\]
	\end{enumerate}
\end{proposition}

\begin{proof}
	Under \(\hat\sigma_i\), student \(i\)'s final assignment is \(Z_i\).
	Under \(\sigma_i\), student \(i\) has the same first-stage match and
	the same other appeal opportunities, but additionally has an admissible
	appeal to \(b\).
	
	If \((i,b)\) is a strict blocking pair, \(r_2\) upholds the appeal with
	probability one. The additional appeal changes \(i\)'s final assignment
	only when \(b\succ_i Z_i\), in which case the utility gain is
	\(u_i(b)-u_i(Z_i)\). Taking expectations gives
	\[
	\mathbb{E}\left[\max\{u_i(b)-u_i(Z_i),0\}\right].
	\]
	
	If \((i,b)\) is not a strict blocking pair, \(r_2\) upholds the appeal
	with probability \(p\), independently of the other appeal outcomes.
	The additional appeal again matters only when \(b\succ_i Z_i\).
	Hence the expected utility gain is
	\[
	p\,\mathbb{E}\left[\max\{u_i(b)-u_i(Z_i),0\}\right].
	\]
\end{proof}

\begin{proposition}
	\label{prop:ia_r2_reduces_manipulation}
	Fix \(p \in (0,1)\), a preference profile \(\succ\), a student \(i\),
	and a von Neumann -- Morgenstern utility representation \(u_i\) of
	\(\succ_i\). Suppose that under truthful IA, student \(i\) applies to
	school \(b\), is rejected, and obtains \(a\), with
	\[
	b \succ_i a.
	\]
	Suppose also that a deviation \(\hat\sigma_i\) obtains \(b\) directly
	in the first stage. Holding fixed appeal opportunities to schools
	strictly preferred to \(b\):
	\begin{enumerate}
		\item if the truthful appeal to \(b\) is upheld under \(r_1\), the
		deviation yields no gain from obtaining \(b\);
		
		\item if the truthful appeal to \(b\) is not upheld under \(r_1\),
		the expected utility gain from the deviation is at most
		\[
		(1-p)\bigl(u_i(b)-u_i(a)\bigr).
		\]
	\end{enumerate}
\end{proposition}

\begin{proof}
	Under truthful reporting, student \(i\) receives \(a\) in the first
	stage and has an admissible appeal to \(b\). Couple the realizations of
	appeals to schools strictly preferred to \(b\) under truthful reporting
	and under the deviation.
	
	If the appeal to \(b\) is upheld under \(r_1\), then \(r_2\) upholds it
	with probability one. If an appeal to a school strictly preferred to
	\(b\) succeeds, both reports yield the same better outcome; otherwise,
	both yield \(b\). Hence the deviation yields no gain from obtaining
	\(b\) directly.
	
	If the appeal to \(b\) is not upheld under \(r_1\), then \(r_2\) upholds
	it with probability \(p\). Whenever an appeal to a school strictly
	preferred to \(b\) succeeds, obtaining \(b\) directly yields no gain.
	Otherwise, the deviation gives \(b\), while truthful reporting gives
	\(b\) with probability \(p\) and an outcome weakly preferred to \(a\)
	with probability \(1-p\). The expected utility gain from the deviation
	is therefore at most
	\[
	(1-p)\bigl(u_i(b)-u_i(a)\bigr).
	\]
\end{proof}

Propositions~\ref{prop:ia_r2_lottery_truth}
and~\ref{prop:ia_r2_reduces_manipulation} capture a local incentive
effect: probabilistic appeals create an option value for truthful
applications to risky schools. For this particular safe-school
manipulation channel, the direct gain from obtaining \(b\) in the first
stage is reduced to at most a factor of \(1-p\) of its no-appeals value,
and this bound vanishes as \(p \to 1\). This does not make IA
sd-strategy-proof, because other manipulations, including preemption and
manipulations that create appeal opportunities to still better schools,
may remain profitable.

\subsection{Discussion}

The results in this appendix complete the incentive analysis across all
three appeal rules. Table~\ref{tab:incentive_summary} summarizes.

\begin{table}[H]
	\centering
	\caption{Incentive properties across appeal rules}
	\label{tab:incentive_summary}
	\resizebox{\textwidth}{!}{\begin{tabular}{lcccc}
		\toprule
		& \(r_0\)  & \(r_1\)  & \(r_2\) 
		& \(r_3\)  \\
		\midrule
		DA strategy-proof
		& yes & yes & yes (sd)
		& no \\
	IA manipulable
	& yes & yes, but less & yes, lower safe-school gains
	& yes, but less \\
		\bottomrule
	\end{tabular}}
\end{table}

Among the rules considered here, \(r_3\) is the one under which DA
loses strategy-proofness in the presence of weak priorities, because
appeals can be upheld within a priority class after a tie-break loss.
Under strict priorities, \(r_3\) coincides with \(r_1\), and DA remains
strategy-proof. 
Both \(r_1\) and \(r_2\) preserve DA's strategy-proofness
deterministically and in stochastic dominance, respectively. Under IA,
\(r_1\) reduces manipulability, while \(r_2\) reduces the gains from an
important class of manipulations.

The probabilistic rule \(r_2\) therefore has a distinctive role. Like
\(r_1\), it preserves DA's incentive property, now in
stochastic-dominance terms. Under IA, it also supports an sd-Nash
equilibrium whose lottery stochastically Pareto-dominates truthful DA.
At the same time, \(r_2\) gives students a positive reason to rank
desired schools even when they have no priority claim: a rejected
application becomes a lottery ticket. The rule does not make IA
sd-strategy-proof, but it reduces the expected gain from an important
class of safe-school manipulations.

These findings have a direct institutional interpretation. England's
School Admission Appeals Code mandates that panels uphold appeals when
admission arrangements were misapplied, corresponding to \(r_1\), but
panels also exercise discretion in borderline cases, leading to
substantial variation in success rates across local authorities. This
discretionary component is captured by \(r_2\). Our results show that
this heterogeneity, while potentially concerning from an equity
perspective, does not undermine DA's incentive properties and further
discourages strategic misreporting under IA.

\newpage

%r3
\section{Analyzing $r_3$.\hfill}
\label{app:weakpriorities}
\subsection{\(r_3\), Pareto improvements, and weak blocking pairs}

Lemma~\ref{prop:pareto} shows that strict appeals produce a Pareto improvement
and eliminate every appealable strict blocking pair involving a school that
the student truly prefers to her final assignment. The corresponding result
under the weak appeal rule is as follows.

\begin{lemma}
	For every mechanism \(M\) and every report profile \(\hat\succ\),
	\((M\circ r_3)(\hat\succ)\) weakly Pareto-dominates
	\((M\circ r_0)(\hat\succ)\) with respect to true preferences and admits no
	weak blocking pair \((i,s)\) such that
	\[
	s\in A_i(\hat\succ,M)
	\quad\text{and}\quad
	s\succ_i(M\circ r_3)_i(\hat\succ).
	\]
\end{lemma}

The proof is identical to that of Lemma~\ref{prop:pareto}, replacing strict
priority comparisons with weak priority comparisons, and is therefore omitted.

\subsection{$\ia \circ r_0$ is more manipulable than $\ia \circ r_3$}\label{prop:ia_less_manipulable_r1_r3}

In the main text, we proved that $\ia \circ r_0$ is more manipulable than $\ia \circ r_1$. Now we show the same for $\ia \circ r_3$.
\begin{proposition}
	$\ia \circ r_0$ is more manipulable than $\ia \circ r_3$.
\end{proposition}
\begin{proof}[Proof for \(r_3\)]
	Suppose student \(i\) can profitably manipulate \(\ia\circ r_3\). Let
	\(\succ_i'\) be a profitable report, holding all other reports fixed at
	\(\succ_{-i}\). Let
	\[
	\mu^0=\ia(\succ_i',\succ_{-i})
	\]
	be the first-stage IA outcome under this deviation, and let
	\[
	\mu=(\ia\circ r_3)(\succ_i',\succ_{-i})
	\]
	be the final outcome after weak appeals. Write \(y=\mu_i\). Since the
	deviation is profitable,
	\[
	y\succ_i(\ia\circ r_3)_i(\succ).
	\]
	Appeals never make a student worse off relative to the first-stage IA outcome,
	so
	\[
	(\ia\circ r_3)_i(\succ)\succeq_i \ia_i(\succ).
	\]
	Hence
	\[
	y\succ_i \ia_i(\succ).
	\]
	We show that \(i\) can also obtain \(y\) by manipulating plain IA.
	
	If \(y=\mu_i^0\), then \(i\) already obtains \(y\) in the first-stage IA
	outcome under the report \(\succ_i'\). Therefore the same report is a
	profitable manipulation of \(\ia=\ia\circ r_0\).
	
	Suppose instead that \(y\neq \mu_i^0\). Then \(i\) obtains \(y\) through a
	successful weak appeal. Hence
	\[
	y\succ_i'\mu_i^0
	\]
	and there exists a first-stage occupant \(j\in(\mu_y^0)^{-1}\) such that
	\[
	i\unrhd_y j .
	\]
	
	We claim that \(i\) obtains \(y\) directly under plain IA by reporting \(y\)
	as her first choice. Suppose, toward a contradiction, that she does not. Since
	\(i\) applies to \(y\) in the first round, this means that all \(q_y\) seats
	of \(y\) are filled in the first round by a set \(H\) of students, all
	different from \(i\), who rank \(y\) first and who beat \(i\) under the
	tie-broken priority order used by IA.
	
	Because only \(i\)'s report has changed, every student in \(H\) also ranks
	\(y\) first in the truthful profile. Moreover, the same students in \(H\) are
	accepted by \(y\) in the first round of truthful IA: removing \(i\) from the
	first-round applicant pool at \(y\) cannot prevent the students in \(H\) from
	filling all \(q_y\) seats, and IA acceptances are permanent. Hence every
	student in \(H\) is assigned to \(y\) in the truthful first-stage IA outcome.
	
	If some \(h\in H\) is tied with \(i\) under the original weak priority at
	\(y\), then
	\[
	i\unrhd_y h .
	\]
	Since \(h\) is assigned to \(y\) in the truthful first-stage IA outcome and
	since \(y\succ_i\ia_i(\succ)\), student \(i\) has a truthful weak appeal to
	\(y\). Under \(r_3\), this appeal is upheld. Therefore the truthful final
	outcome \((\ia\circ r_3)_i(\succ)\) assigns \(i\) to \(y\) or to a school she
	prefers to \(y\), contradicting
	\[
	y\succ_i(\ia\circ r_3)_i(\succ).
	\]
	
	Therefore every student in \(H\) must strictly outrank \(i\) at \(y\). Since
	the students in \(H\) rank \(y\) first, they occupy all seats at \(y\) from
	the first round onward under any unilateral report by \(i\). Thus, under the
	deviation \(\succ_i'\), every first-stage occupant of \(y\) strictly outranks
	\(i\). This contradicts the fact that \(i\)'s appeal to \(y\) is successful
	under \(r_3\), which requires some first-stage occupant \(j\in(\mu_y^0)^{-1}\)
	with
	\[
	i\unrhd_y j .
	\]
	
	Hence \(i\) must be admitted to \(y\) directly when she reports \(y\) first.
	Since \(y\succ_i\ia_i(\succ)\), reporting \(y\) first is a profitable
	manipulation of plain IA. Therefore every profitable manipulation of
	\(\ia\circ r_3\) implies a profitable manipulation of \(\ia\circ r_0=\ia\).
	
	It remains to show that the containment is strict.
	Example~\ref{ex:moremanipulable} in the main text has strict priorities, so
	\(r_3\) coincides with \(r_1\). The profile is manipulable under
	\(\ia\circ r_0\) but not under \(\ia\circ r_1=\ia\circ r_3\).
	Therefore \(\ia\circ r_0\) is more manipulable than \(\ia\circ r_3\).
\end{proof}

\subsection{$\da \circ r_3$ is manipulable}
To see that $\da \circ r_3$ is manipulable, consider the school choice problem below with four students and four schools, each with one seat, and a common tie-break that favors students with larger indices. 
\begin{example}Preferences and priorities are given as follows (manipulations in parentheses):
	\label{ex:da_manipulable}
	\textcolor{white}{text}
	
	\begin{center}
		\begin{tabular}{CCCCC|CCCC}
			
			\succ_{i_1}&\succ_{i_2}&\succ_{i_3}&(\succ'_{i_3})&\succ_{i_4}&\unrhd_{s_1}&\unrhd_{s_2}&\unrhd_{s_3}&\unrhd_{s_4}\\
			\hline
			s_4 & s_4 & s_2 & (s_2)&s_1  & i_1 & i_2     &  \cdot     & i_4 \\
			s_1 & s_1 & s_4 & (s_3)&s_2 &  i_2 & i_3,i_4 &       \cdot& i_3 \\
			\cdot  & s_2 & s_3& \cdot& s_4  & i_4 &    \cdot     &       \cdot& i_2 \\
			\cdot & \cdot    & \cdot& \cdot    & \cdot    & \cdot        & \cdot    & \cdot     & i_1 \\
			
		\end{tabular}
	\end{center}
	
	With truthful preference reporting, $\da$ yields the matching
	$(i_1- s_1, i_2- s_2, i_3- s_3, i_4- s_4)$.
	There are no successful appeals under $r_3$, and thus the DA matching remains unaltered.
	Now let student $i_3$ misreport her preferences by submitting $s_2 \succ'_{i_3} s_3$. Now DA yields the first-stage matching
	$(i_1-s_1, i_2-s_4, i_3-s_3, i_4-s_2)$. Student $i_3$ then appeals her rejection at $s_2$. Since student $i_3$ and student $i_4$ have the same priority at $s_2$, the pair $(i_3,s_2)$ is a weak blocking pair. Hence the appeal succeeds under $r_3$, and student $i_3$ is reassigned from $s_3$ to $s_2$.
	Since student $i_3$ truly prefers $s_2$ to $s_3$, this misreport is profitable. Therefore $\da\circ r_3$ is manipulable.
\end{example}

\subsection{$\da \circ r_3$ can generate more successful appeals than $\ia \circ r_3$}

The preceding example shows that weak appeals can undermine the incentive
properties of DA. We now show that weak appeals can also overturn the natural
intuition that stable mechanisms should generate fewer successful appeals. Under
weak appeals, DA may create appeal opportunities for students who lose a
tie-break against otherwise equal-priority applicants. As a result, truthful DA
may generate more successful appeals than truthful IA.

\begin{proposition}
	\label{prop:r3_ia_fewer_appeals_than_da}
	There exists a school choice problem with weak priorities such that, under
	truthful reporting, \(\ia\circ r_3\) generates fewer successful appeals than
	\(\da\circ r_3\).
\end{proposition}

\begin{proof}
	Consider a school choice problem with seven students,
	\(i_1,\ldots,i_7\), and three schools, \(s_1,s_2,s_3\), each with two seats.
	Priorities are weak, and ties within a priority class are broken in favour of
	students with smaller indices. Preferences and weak priorities are as follows.
	
	\begin{table}[H]
		\centering
		\begin{tabular}{ccccccc|ccc}
			\toprule
			\(\succ_{i_1}\) & \(\succ_{i_2}\) & \(\succ_{i_3}\) &
			\(\succ_{i_4}\) & \(\succ_{i_5}\) & \(\succ_{i_6}\) &
			\(\succ_{i_7}\)
			&
			\(\unrhd_{s_1}\) & \(\unrhd_{s_2}\) & \(\unrhd_{s_3}\) \\
			\midrule
			\(s_3\) & \(s_3\) & \(s_2\) & \(s_1\) & \(s_1\) & \(s_2\) & \(s_3\)
			&
			\(i_6,i_7\) & \(i_4\) & \(i_4,i_5\) \\
			\(s_2\) & \(s_2\) & \(s_1\) & \(s_3\) & \(s_2\) & \(s_3\) & \(s_2\)
			&
			\(i_5\) & \(i_1,i_2,i_3,i_5\) & \(i_3,i_6\) \\
			\(s_1\) & \(s_1\) & \(s_3\) & \(s_2\) & \(s_3\) & \(s_1\) & \(s_1\)
			&
			\(i_2,i_3,i_4\) & \(i_6,i_7\) & \(i_2,i_7\) \\
			\(s_\emptyset\) & \(s_\emptyset\) & \(s_\emptyset\) &
			\(s_\emptyset\) & \(s_\emptyset\) & \(s_\emptyset\) & \(s_\emptyset\)
			&
			\(i_1\) & \(\cdot\) & \(i_1\) \\
			\bottomrule
		\end{tabular}
	\end{table}
	
	Under truthful IA, the first-stage matching is
	\[
	(i_1-s_\emptyset,\ i_2-s_3,\ i_3-s_2,\ i_4-s_1,\ i_5-s_1,\ i_6-s_2,\ i_7-s_3).
	\]
	The only successful weak appeal is by student \(i_1\) to school \(s_2\).
	Indeed, \(i_1\) was rejected from \(s_2\), prefers \(s_2\) to being
	unassigned, and belongs to the same weak priority class at \(s_2\) as the
	incumbent \(i_3\). Hence \(i_1\)'s appeal to \(s_2\) is upheld under \(r_3\).
	No other student who fails to receive a more preferred school has a successful
	weak appeal. Thus \(\ia\circ r_3\) generates exactly one successful appeal.
	
	Now consider truthful DA, computed using the same tie-breaking rule. The
	first-stage DA matching is
	\[
	(i_1-s_2,\ i_2-s_2,\ i_3-s_\emptyset,\ i_4-s_3,\ i_5-s_3,\ i_6-s_1,\ i_7-s_1).
	\]
	Under this outcome, there are two successful weak appeals. Student \(i_3\)
	appeals to \(s_2\), where she belongs to the same weak priority class as the
	incumbents \(i_1\) and \(i_2\). Hence her appeal is upheld. Student \(i_5\)
	also appeals to \(s_2\), where she too belongs to the same weak priority class
	as \(i_1\) and \(i_2\). Hence her appeal is also upheld. No further appeal is
	needed for the comparison.
	
	Therefore, under truthful reporting, \(\ia\circ r_3\) generates one successful
	appeal, whereas \(\da\circ r_3\) generates two successful appeals. This proves
	the claim.
\end{proof}

In the previous example, truthful reporting is a Nash equilibrium of
\(\ia\circ r_3\). Indeed, under truthful reporting and weak appeals, every
student except \(i_1\) receives her first-choice school, while \(i_1\) obtains
\(s_2\) through a weak appeal. The only school \(i_1\) prefers to \(s_2\) is
\(s_3\), but both seats at \(s_3\) are occupied by students with strictly higher
weak priority than \(i_1\). Hence \(i_1\) cannot obtain \(s_3\), either
directly or by appeal, and no student has a profitable deviation. Thus, the
lower number of successful appeals under IA arises at an equilibrium outcome.
The reversal theorem below shows that this is not an isolated phenomenon:
under weak appeals, IA admits equilibrium outcomes with strong welfare
properties relative to DA.

\subsection{The Reversal Theorem}
Theorem~\ref{thm:ia_r1_dominates_da} shows that, under strict priorities,
\(\ia\circ r_1\) admits a Nash equilibrium whose outcome weakly
Pareto-dominates truthful DA. We now extend this result to weak priorities
under the weak appeal rule \(r_3\).
\begin{theorem}
	\label{thm:ia_r3_dominates_da_weak}
	Fix weak priorities \((\unrhd_s)_{s\in S}\), and fix a tie-breaking rule
	\(\tau\) that induces strict priorities \((\rhd_s^\tau)_{s\in S}\) used by
	\(\ia\) and \(\da\). Under the weak appeal rule \(r_3\), there exists a Nash
	equilibrium \(\sigma^*\) of \(\ia\circ r_3\) such that
	\[
	(\ia\circ r_3)(\sigma^*) \succsim \da(\succ).
	\]
	Furthermore, if $(\ia\circ r_3)(\sigma^*)\neq \da(\succ)$,
	then at least one student is strictly better off.
\end{theorem}

\begin{proof}
	Throughout the proof, \(\da(\succ)\) denotes the student-proposing DA outcome
	computed using the tie-broken strict priorities \((\rhd_s^\tau)_{s\in S}\).
	Under \(r_3\), a student may appeal only to schools that she ranked above her
	first-stage IA assignment and from which she was rejected in the first stage.
	An appeal to such a school \(s\) is successful if and only if the student has
	weakly higher priority at \(s\) than some first-stage occupant of \(s\), that
	is, if \(i\unrhd_s j\) for some first-stage occupant \(j\) of \(s\).
	
	We prove the result using an auxiliary mechanism, denoted \(\mia^3\).

	\paragraph{Modified IA for weak appeals.}
	The mechanism \(\mia^3\) proceeds in synchronous rounds. In round \(k\), every
	currently unassigned student applies to the \(k\)-th school on her true
	preference list. Admissions are permanent. At each school \(s\), regular seats
	are filled only while fewer than \(q_s\) students have been regularly admitted
	to \(s\). If \(s\) has remaining regular capacity, it admits the
	highest-priority current applicants according to the tie-broken order
	\(\rhd_s^\tau\), up to the remaining regular capacity. Once \(q_s\) students
	have been regularly admitted to \(s\), no further regular admissions occur.
	Any further current applicant \(i\) is admitted by overbooking if and only if
	\[
	i\unrhd_s j
	\]
	for some regularly admitted student \(j\) at \(s\). In that case, \(s\)'s
	effective capacity increases by one. Students admitted within the regular
	capacity are called \emph{regularly admitted}; students admitted through the
	capacity-expansion rule are called \emph{overbooked}. All other applicants are
	rejected and proceed to the next school on their list in the following round.
	
	Let \(\mia^3(\succ)\) denote the outcome.

	\paragraph{Rejection lemma.}
	
	\begin{lemma}
		\label{lem:mia3_rejection}
		If student \(i\) is rejected from school \(s\) under \(\mia^3(\succ)\), then
		every regularly admitted student at \(s\) under \(\mia^3(\succ)\) has strictly
		higher weak priority at \(s\) than \(i\). That is, for every regularly admitted
		student \(j\) at \(s\),
		\[
		j \unrhd_s i
		\quad\text{and not}\quad
		i\unrhd_s j.
		\]
	\end{lemma}
	
	\begin{proof}
		Student \(i\) is rejected from \(s\) only if all regular seats at \(s\) have
		already been filled and \(i\) is not weakly higher priority than any regularly
		admitted student at \(s\). If there existed a regularly admitted student \(j\)
		at \(s\) such that \(i\unrhd_s j\), then the overbooking rule would admit
		\(i\). Therefore, for every regularly admitted student \(j\) at \(s\), it is
		not the case that \(i\unrhd_s j\). Since \(\unrhd_s\) is complete, this means
		that \(j\) has strictly higher weak priority than \(i\). Regular admittees are
		never removed, and no further regular admittees can be added after the
		\(q_s\) regular seats are filled. Hence the claim holds for every regular
		admittee of \(s\) under \(\mia^3(\succ)\).
	\end{proof}
	
	By the definition of overbooking, if student \(i\) is overbooked at school
	\(s\) under \(\mia^3(\succ)\), then there exists a regularly admitted student
	\(j\) at \(s\) such that
	\[
	i\unrhd_s j.
	\]

	\paragraph{\(\mia^3\) weakly Pareto-dominates DA.}
	
	\begin{lemma}
		\label{lem:mia3_dominates_da}
		For every school choice problem with weak priorities and fixed tie-breaking,
		\[
		\mia^3(\succ)\succsim \da(\succ).
		\]
	\end{lemma}
	
	\begin{proof}
		Suppose, toward a contradiction, that some student is strictly better off under
		\(\da(\succ)\) than under \(\mia^3(\succ)\). Let
		\[
		B=\{i\in I:\da_i(\succ)\succ_i \mia^3_i(\succ)\}.
		\]
		By assumption, \(B\neq\varnothing\).
		
		For each \(i\in B\), let
		\[
		s_i=\da_i(\succ).
		\]
		Since \(s_i\succ_i \mia^3_i(\succ)\), student \(i\) applied to \(s_i\) under
		\(\mia^3\) before receiving her \(\mia^3\)-assignment, and was rejected. Let
		\(r_i\) be the round in which \(i\) applied to \(s_i\) under \(\mia^3\). Let
		\(\rho_i\) be the round in which \(i\) receives her \(\mia^3\)-assignment, and
		set \(\rho_i=\infty\) if \(i\) is unassigned under \(\mia^3\). Then
		\[
		r_i<\rho_i.
		\]
		
		By Lemma~\ref{lem:mia3_rejection}, every regularly admitted student at
		\(s_i\) under \(\mia^3\) has strictly higher weak priority at \(s_i\) than
		\(i\). Since \(i\) was rejected from \(s_i\), the \(q_{s_i}\) regular seats at
		\(s_i\) had already been filled. Let \(D(s_i)\) denote the set of these
		\(q_{s_i}\) regular admittees.
		
		Since \(i\notin D(s_i)\), \(\da_i(\succ)=s_i\), and school \(s_i\) has only
		\(q_{s_i}\) seats under DA, at most \(q_{s_i}-1\) students from \(D(s_i)\) can
		be assigned to \(s_i\) under \(\da(\succ)\). Hence there exists some student
		\[
		h\in D(s_i)
		\]
		who is not assigned to \(s_i\) under \(\da(\succ)\).
		
		By Lemma~\ref{lem:mia3_rejection}, \(h\) has strictly higher weak priority
		than \(i\) at \(s_i\). Therefore \(h\) also has higher tie-broken priority than
		\(i\) at \(s_i\):
		\[
		h\rhd^\tau_{s_i} i.
		\]
		Since \(i\) is assigned to \(s_i\) under DA and DA is stable with respect to
		the tie-broken priorities, stability implies
		\[
		\da_h(\succ)\succ_h s_i.
		\]
		Otherwise, \((h,s_i)\) would be a blocking pair of the DA outcome under the
		tie-broken priorities. But \(h\)'s \(\mia^3\)-assignment is \(s_i\), because
		\(h\) is regularly admitted to \(s_i\) under \(\mia^3\). Therefore
		\[
		\da_h(\succ)\succ_h \mia^3_h(\succ),
		\]
		so \(h\in B\).
		
		Moreover, \(h\) was regularly admitted to \(s_i\) before or at the round in
		which \(i\) was rejected from \(s_i\). Hence
		\[
		\rho_h\le r_i<\rho_i.
		\]
		
		Thus, from any \(i\in B\), we can construct another student \(h\in B\) with
		\[
		\rho_h<\rho_i.
		\]
		If all students in \(B\) had \(\rho_i=\infty\), this argument would produce a
		student in \(B\) with finite \(\rho_h\), a contradiction. Hence \(B\) contains
		at least one student with finite \(\rho_i\). Choose \(i^*\in B\) with minimal
		finite \(\rho_{i^*}\). Applying the argument above to \(i^*\) produces
		\(h\in B\) with finite \(\rho_h<\rho_{i^*}\), contradicting the minimality of
		\(\rho_{i^*}\). Therefore \(B=\varnothing\), and
		\[
		\mia^3(\succ)\succsim \da(\succ).
		\]
	\end{proof}

	\paragraph{Constructing the report profile.}
	
	Partition students according to their status under \(\mia^3(\succ)\):
	\begin{itemize}
		\item \(D\): regularly admitted students;
		\item \(O\): overbooked students;
		\item \(U\): unassigned students.
	\end{itemize}
	
	Define \(\sigma^*\) as follows:
	\begin{itemize}
		\item If \(i\in D\) and \(\mia^3_i(\succ)=s\), then \(i\) ranks \(s\)
		first, followed by the remaining schools in her true order.
		\item If \(i\in O\) and \(\mia^3_i(\succ)=s\), then \(i\) reports
		truthfully down to \(s\), truncating her list after \(s\).
		\item If \(i\in U\), then \(i\) reports truthfully.
	\end{itemize}

	\paragraph{Implementation.}
	
	We show that
	\[
	(\ia\circ r_3)(\sigma^*)=\mia^3(\succ).
	\]
	
	Fix a school \(s\), and let \(D_s\) denote the set of students regularly
	admitted to \(s\) under \(\mia^3\). Every student in \(D_s\) ranks \(s\) first
	under \(\sigma^*\), so every student in \(D_s\) applies to \(s\) in round 1 of
	IA.
	
	We first show that the round-1 admits at \(s\) under \(\ia(\sigma^*)\) are
	exactly the students in \(D_s\).
	
	First suppose that \(|D_s|<q_s\). Then no student outside \(D_s\) ranks \(s\)
	first under \(\sigma^*\). To see this, suppose that some student
	\(i\notin D_s\) does. If \(i\in D\), then \(i\) ranks her own
	\(\mia^3\)-school first, so she cannot rank \(s\) first unless she belongs to
	\(D_s\), a contradiction. If \(i\in O\) or \(i\in U\), then \(s\) is also
	her first choice under \(\succ_i\), so she applied to \(s\) in round 1 under
	\(\mia^3\). Since \(|D_s|<q_s\), school \(s\) never filled its regular
	capacity under \(\mia^3\). Hence \(i\) would have been regularly admitted to
	\(s\), again a contradiction. Therefore the only round-1 applicants to \(s\)
	under \(\sigma^*\) are the students in \(D_s\), and they are admitted.
	
	Now suppose that \(|D_s|=q_s\). Consider any additional student
	\(i\notin D_s\) who also applies to \(s\) in round 1 under \(\sigma^*\).
	There are three cases.
	
	First, suppose \(i\in O\) and \(\mia^3_i(\succ)=s\). Then \(i\) is assigned to
	\(s\) by overbooking under \(\mia^3\). If \(i\) ranks \(s\) first under
	\(\sigma_i^*\), then \(s\) is also \(i\)'s first choice under \(\succ_i\). Thus
	\(i\) applied to \(s\) in round 1 under \(\mia^3\). Since \(i\) was not
	regularly admitted at \(s\), the students in \(D_s\) selected for the regular
	seats at \(s\) in round 1 have higher tie-broken priority than \(i\). Hence
	\(i\) is not selected for a regular seat at \(s\) in round 1 of IA under
	\(\sigma^*\).
	
	Second, suppose \(i\in O\), \(\mia^3_i(\succ)=t\neq s\), and \(s\succ_i t\).
	Then \(i\) was rejected from \(s\) under \(\mia^3\). By
	Lemma~\ref{lem:mia3_rejection}, every student in \(D_s\) has strictly higher
	weak priority at \(s\) than \(i\), and therefore higher tie-broken priority
	than \(i\).
	
	Third, suppose \(i\in U\) and \(s\) is her first choice. Then \(i\) was
	rejected from \(s\) under \(\mia^3\). Again, by
	Lemma~\ref{lem:mia3_rejection}, every student in \(D_s\) has strictly higher
	weak priority at \(s\) than \(i\), and therefore higher tie-broken priority
	than \(i\).
	
	Thus every student in \(D_s\) has higher tie-broken priority at \(s\) than
	every additional round-1 applicant to \(s\) under \(\sigma^*\). Since
	\(|D_s|=q_s\), the students in \(D_s\) are exactly the round-1 admits at
	\(s\) under \(\ia(\sigma^*)\).
	
	Now consider \(i\in O\), and let \(\mia^3_i(\succ)=s\). Under
	\(\sigma_i^*\), student \(i\) applies to all schools she truly prefers to
	\(s\), and then to \(s\). At every school \(t\succ_i s\), student \(i\) was
	rejected under \(\mia^3\). By Lemma~\ref{lem:mia3_rejection}, every regular
	admittee at \(t\) has strictly higher weak priority than \(i\), and therefore
	higher tie-broken priority than \(i\). Since the round-1 IA admits at \(t\)
	are exactly the students in \(D_t\), student \(i\) is rejected from every such
	\(t\) in the IA first stage.
	
	At \(s\), the regular seats are held by \(D_s\), so \(i\) is rejected in the
	first stage. Her list is truncated at \(s\), so she applies to no further
	school. Since \(i\) is overbooked at \(s\) under \(\mia^3\), there exists
	\(j\in D_s\) such that
	\[
	i\unrhd_s j.
	\]
	This student \(j\) occupies \(s\) in the IA first stage. Therefore \(i\)'s
	weak appeal to \(s\) is upheld under \(r_3\), and \(i\) obtains \(s\), exactly
	as under \(\mia^3\).
	
	Finally, consider \(i\in U\). Student \(i\) was rejected from every school she
	prefers to being unassigned under \(\mia^3\). By
	Lemma~\ref{lem:mia3_rejection}, every regular admittee at every such school
	has strictly higher weak priority than \(i\), and therefore higher tie-broken
	priority than \(i\). Since the first-stage IA seats at each such school are
	filled by the students in the corresponding \(D_t\), student \(i\) is rejected
	everywhere in the first stage and has no successful weak appeal. Hence \(i\)
	remains unassigned, as under \(\mia^3\).
	
	Therefore
	\[
	(\ia\circ r_3)(\sigma^*)=\mia^3(\succ).
	\]

	\paragraph{\(\sigma^*\) is a Nash equilibrium.}
	
	Fix a student \(i\), and consider any school \(t\) such that
	\[
	t\succ_i \mia^3_i(\succ).
	\]
	Student \(i\) was rejected from \(t\) under \(\mia^3\). By
	Lemma~\ref{lem:mia3_rejection}, every regular admittee of \(t\) has strictly
	higher weak priority at \(t\) than \(i\), and therefore higher tie-broken
	priority than \(i\). Under \(\sigma^*_{-i}\), these regular admittees rank
	\(t\) first and fill the first-stage IA seats at \(t\), independently of
	\(i\)'s report. Since IA admissions are permanent, \(i\) cannot obtain \(t\)
	directly, regardless of when she applies to \(t\).
	
	Nor can \(i\) obtain \(t\) through a weak appeal. The first-stage occupants of
	\(t\) are precisely the students in \(D_t\), all of whom have strictly higher
	weak priority at \(t\) than \(i\). Thus \(i\) is not weakly higher priority
	than any first-stage occupant of \(t\), so no weak appeal to \(t\) succeeds.
	
	Therefore no deviation allows \(i\) to obtain a school she strictly prefers to
	\(\mia^3_i(\succ)\). Since \(i\) was arbitrary, \(\sigma^*\) is a Nash
	equilibrium of \(\ia\circ r_3\).
	
	Combining the implementation result with Lemma~\ref{lem:mia3_dominates_da},
	we obtain
	\[
	(\ia\circ r_3)(\sigma^*)=\mia^3(\succ)\succsim\da(\succ).
	\]
	If \((\ia\circ r_3)(\sigma^*)\neq \da(\succ)\), then Pareto dominance together
	with strict preferences implies that at least one student is strictly better
	off. This proves the theorem.
\end{proof}

\newpage

%further theory examples
%\section{Further Examples}
%\label{app:iamanipulable}
%\input{furtherex}
%\newpage

%additional lab results
\section{Additional Laboratory Results.\hfill}
\label{app:additional-results}
\noindent This appendix reports a small set of robustness and supplementary
outcomes in the same order as the main text. Participant-level analyses use
the paper sample of 10,195 participant-round observations and 59 independent
matching groups. The blocking-pair analysis uses the 1,435 complete
seven-person market-rounds constructed from the same sample.

\subsection{Non-parametric tests}
\label{app:non-parametric-tests}

\begin{table}[htbp]
	\centering
	\caption{Nonparametric tests on matching-group means.}
	\label{tab:nonparametric}
	\begin{threeparttable}
		\begin{tabular*}{\textwidth}{@{\extracolsep{\fill}} lcccc @{}}
			\toprule
			& Truth- & Assigned & Stable & Blocking \\
			& telling & rank & share & pairs \\
			& (full ranking) & (induced) & (submitted) & (submitted) \\
			& \((1)\) & \((2)\) & \((3)\) & \((4)\) \\
			\midrule
			\multicolumn{5}{@{}l}{\textbf{Panel A. Appeal effects within mechanism}} \\
			DA, strict
			& \(0.03\) & \(0.00\) & --- & --- \\
			& {\footnotesize \([0.122]\)} & {\footnotesize \([0.789]\)} &  &  \\
			DA, probabilistic
			& \(0.03\) & \(-0.14^{***}\) & \(-0.40^{***}\) & \(0.35^{***}\) \\
			& {\footnotesize \([0.119]\)} & {\footnotesize \([<0.001]\)} & {\footnotesize \([<0.001]\)} & {\footnotesize \([<0.001]\)} \\
			IA, strict
			& \(0.09\) & \(-0.22^{***}\) & \(0.35^{***}\) & \(-0.55^{***}\) \\
			& {\footnotesize \([0.122]\)} & {\footnotesize \([0.004]\)} & {\footnotesize \([<0.001]\)} & {\footnotesize \([<0.001]\)} \\
			IA, probabilistic
			& \(0.10^{***}\) & \(-0.29^{***}\) & \(0.23^{***}\) & \(-0.48^{***}\) \\
			& {\footnotesize \([0.006]\)} & {\footnotesize \([<0.001]\)} & {\footnotesize \([<0.001]\)} & {\footnotesize \([<0.001]\)} \\
			\addlinespace
			\multicolumn{5}{@{}l}{\textbf{Panel B. IA \(-\) DA at each appeal regime}} \\
			No appeals
			& \(-0.01\) & \(-0.42^{***}\) & \(-0.65^{***}\) & \(0.85^{***}\) \\
			& {\footnotesize \([0.626]\)} & {\footnotesize \([<0.001]\)} & {\footnotesize \([<0.001]\)} & {\footnotesize \([<0.001]\)} \\
			Strict appeals
			& \(0.04\) & \(-0.64^{***}\) & \(-0.28^{***}\) & \(0.30^{***}\) \\
			& {\footnotesize \([0.453]\)} & {\footnotesize \([<0.001]\)} & {\footnotesize \([<0.001]\)} & {\footnotesize \([<0.001]\)} \\
			Probabilistic appeals
			& \(0.05^{**}\) & \(-0.58^{***}\) & \(-0.05\) & \(0.09\) \\
			& {\footnotesize \([0.027]\)} & {\footnotesize \([<0.001]\)} & {\footnotesize \([0.420]\)} & {\footnotesize \([0.404]\)} \\
			\bottomrule
		\end{tabular*}
		\begin{tablenotes}[flushleft]
			\footnotesize
			\item \(^{***}\ p<0.01\); \(^{**}\ p<0.05\); \(^{*}\ p<0.10\). Panel~A reports each appeal rule minus the no-appeals arm of the same mechanism; Panel~B reports IA minus DA within each appeal regime. Entries are Hodges--Lehmann location shifts, that is, the median of all pairwise differences between matching-group means in the two arms; exact two-sided Mann--Whitney \(p\)-values are in brackets. Each observation is one matching group: participant-level outcomes are averaged over all participant-rounds in a session, and market-level outcomes over all complete market-rounds in a session. There are nine matching groups in each arm except DA probabilistic (eleven) and IA probabilistic (twelve), so every comparison rests on between \(18\) and \(23\) independent observations. Lower values are better in Columns~(2) and~(4). Cells marked --- have no variation in either arm: DA is stable with respect to submitted preferences by construction, so submitted stability equals one and the submitted blocking-pair count equals zero in every DA market-round without appeals and with strict appeals. \(p\)-values are unadjusted for multiple testing.
		\end{tablenotes}
	\end{threeparttable}
\end{table}

\subsection{First-choice truth-telling}
\label{app:first-choice-truth}

The main analysis defines truth-telling as submitting the complete induced preference ranking. A less demanding measure asks only whether the participant ranks her induced first-choice school first. Figure~\ref{fig:first_choice_truth} and Table~\ref{tab:first_choice_truth_appeals} repeat the main truth-telling analysis using this measure.

\begin{figure}[htbp]
    \centering
    \includegraphics[width=1.00\textwidth]{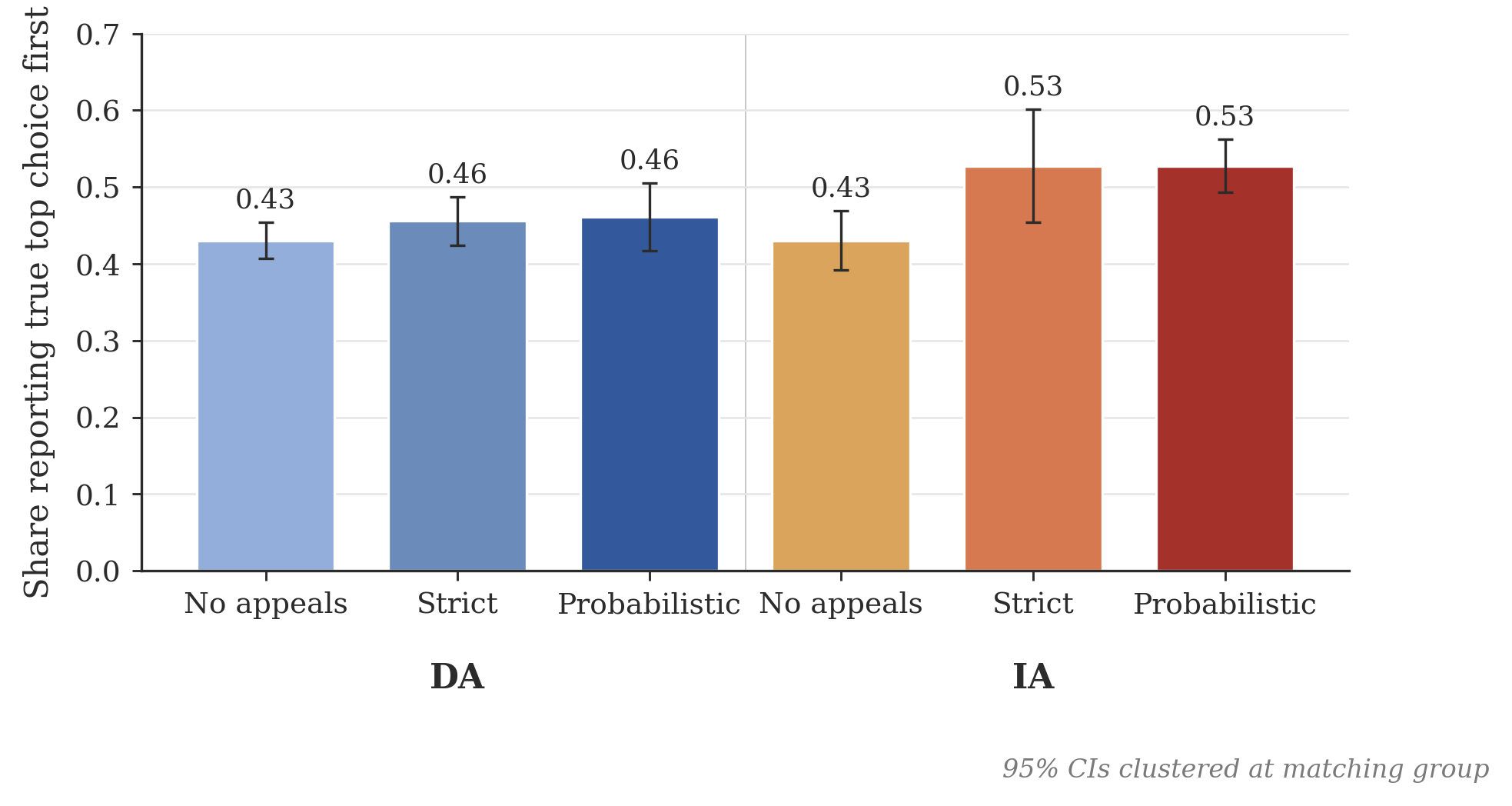}
    \caption{First-choice truth-telling by treatment.\\
    \footnotesize Bars show the share of participants who rank their induced first-choice school first. Whiskers are 95\% confidence intervals clustered at the matching-group level.}
    \label{fig:first_choice_truth}
\end{figure}

The pattern is the same as for complete truth-telling. Under IA, first-choice truth-telling rises from 0.43 without appeals to 0.53 under both strict and probabilistic appeals. Under DA, it rises only to 0.46 under either rule. In the fully controlled specification, strict and probabilistic appeals increase first-choice truth-telling under IA by approximately 10 percentage points. The corresponding DA effects are approximately 3 percentage points and are not statistically significant.

\begin{table}[htbp]
\centering
\caption{Impact of appeals on first-choice truth-telling}
\label{tab:first_choice_truth_appeals}
\begin{threeparttable}
\begin{tabular*}{\textwidth}{@{\extracolsep{\fill}} lcccc @{}}
\toprule
& \((1)\) & \((2)\) & \((3)\) & \((4)\) \\
\midrule
\multicolumn{5}{l}{\textbf{Panel A. IA treatments}} \\
IA strict & \(0.10^{**}\) & \(0.10^{**}\) & \(0.10^{**}\) & \(0.10^{**}\) \\
 & \((0.04)\) & \((0.04)\) & \((0.04)\) & \((0.04)\) \\
IA probabilistic & \(0.10^{***}\) & \(0.10^{***}\) & \(0.10^{***}\) & \(0.10^{***}\) \\
 & \((0.03)\) & \((0.03)\) & \((0.03)\) & \((0.03)\) \\
\addlinespace
\multicolumn{5}{l}{\textbf{Panel B. DA treatments}} \\
DA strict & \(0.03\) & \(0.03\) & \(0.03\) & \(0.03\) \\
 & \((0.02)\) & \((0.02)\) & \((0.02)\) & \((0.02)\) \\
DA probabilistic & \(0.03\) & \(0.03\) & \(0.03\) & \(0.03\) \\
 & \((0.02)\) & \((0.02)\) & \((0.02)\) & \((0.03)\) \\
\addlinespace
Round effects & No & Yes & Yes & Yes \\
Participant-type effects & No & No & Yes & Yes \\
Demographics + risk att. & No & No & No & Yes \\
\midrule
Observations & 10{,}195 & 10{,}195 & 10{,}195 & 10{,}185 \\
Matching groups & \(59\) & \(59\) & \(59\) & \(59\) \\
\bottomrule
\end{tabular*}
\begin{tablenotes}[flushleft]
\footnotesize
\item Notes. \(^{***}p<0.01\), \(^{**}p<0.05\), \(^{*}p<0.100\). Entries are average marginal effects from logit regressions. Standard errors are clustered at the matching-group level. First-choice truth-telling equals one when the participant ranks her induced first-choice school first. The baseline in each panel is the corresponding mechanism without appeals.
\end{tablenotes}
\end{threeparttable}
\end{table}

\paragraph{Evolution over rounds.}
Figure~\ref{fig:first_choice_truth_rounds} reports role-balanced first-choice truth-telling in each round. From round 2 onward, both IA appeal treatments remain above IA without appeals in every round. The DA series overlap and display no persistent ordering. The treatment pattern therefore does not arise from a short-lived response in the opening rounds.

\begin{figure}[htbp]
    \centering
    \includegraphics[width=1.00\textwidth]{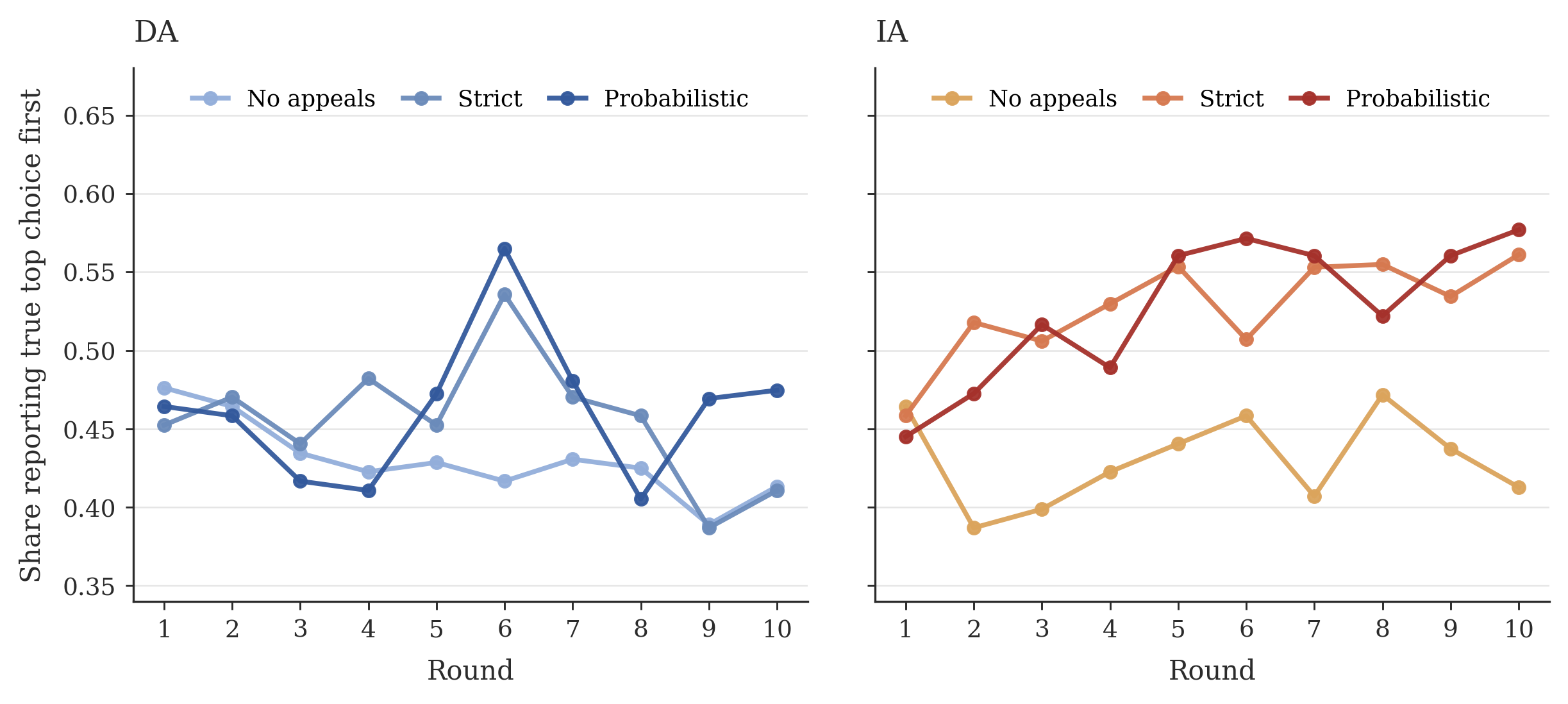}
    \caption{First-choice truth-telling over rounds.\\
    \footnotesize Each observation is the mean first-choice truth-telling rate in a treatment-round, giving equal weight to each of the seven participant types.}
    \label{fig:first_choice_truth_rounds}
\end{figure}

\subsection{Efficiency in experimental points}
\label{app:points}

The main text measures efficiency using assigned rank. Because experimental points are an affine transformation of true-preference rank, the points measure yields the same descriptive ordering. Figure~\ref{fig:points} reports mean points by treatment.

\begin{figure}[htbp]
    \centering
    \includegraphics[width=1.00\textwidth]{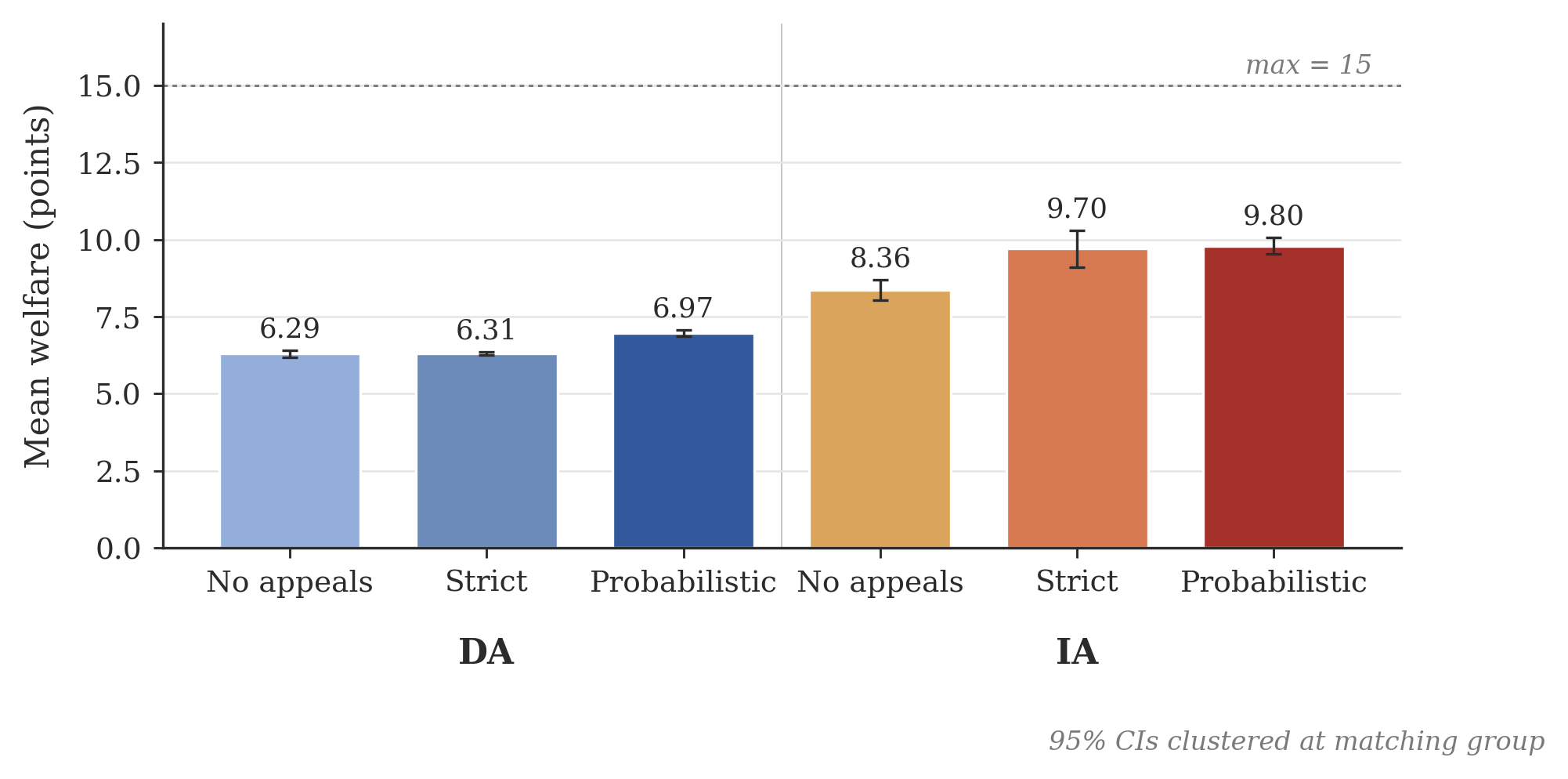}
    \caption{Experimental points by treatment.\\
    \footnotesize Bars report mean points from the final assignment. Higher values indicate better assignments. Whiskers are 95\% confidence intervals clustered at the matching-group level.}
    \label{fig:points}
\end{figure}

Under IA, mean points rise from 8.36 without appeals to 9.70 under strict appeals and 9.80 under probabilistic appeals. Strict appeals leave DA essentially unchanged, at 6.29 without appeals and 6.31 with strict appeals, while probabilistic appeals raise mean points to 6.97.

\paragraph{Unassigned participants.}
Figure~\ref{fig:unassigned} shows that part of IA's efficiency gain occurs on the extensive margin. The unassigned share falls from 0.14 without appeals to 0.07 under strict appeals and 0.06 under probabilistic appeals. Under DA, strict appeals leave the unassigned share at 0.14, while probabilistic appeals reduce it to 0.12.

\begin{figure}[htbp]
    \centering
    \includegraphics[width=0.90\textwidth]{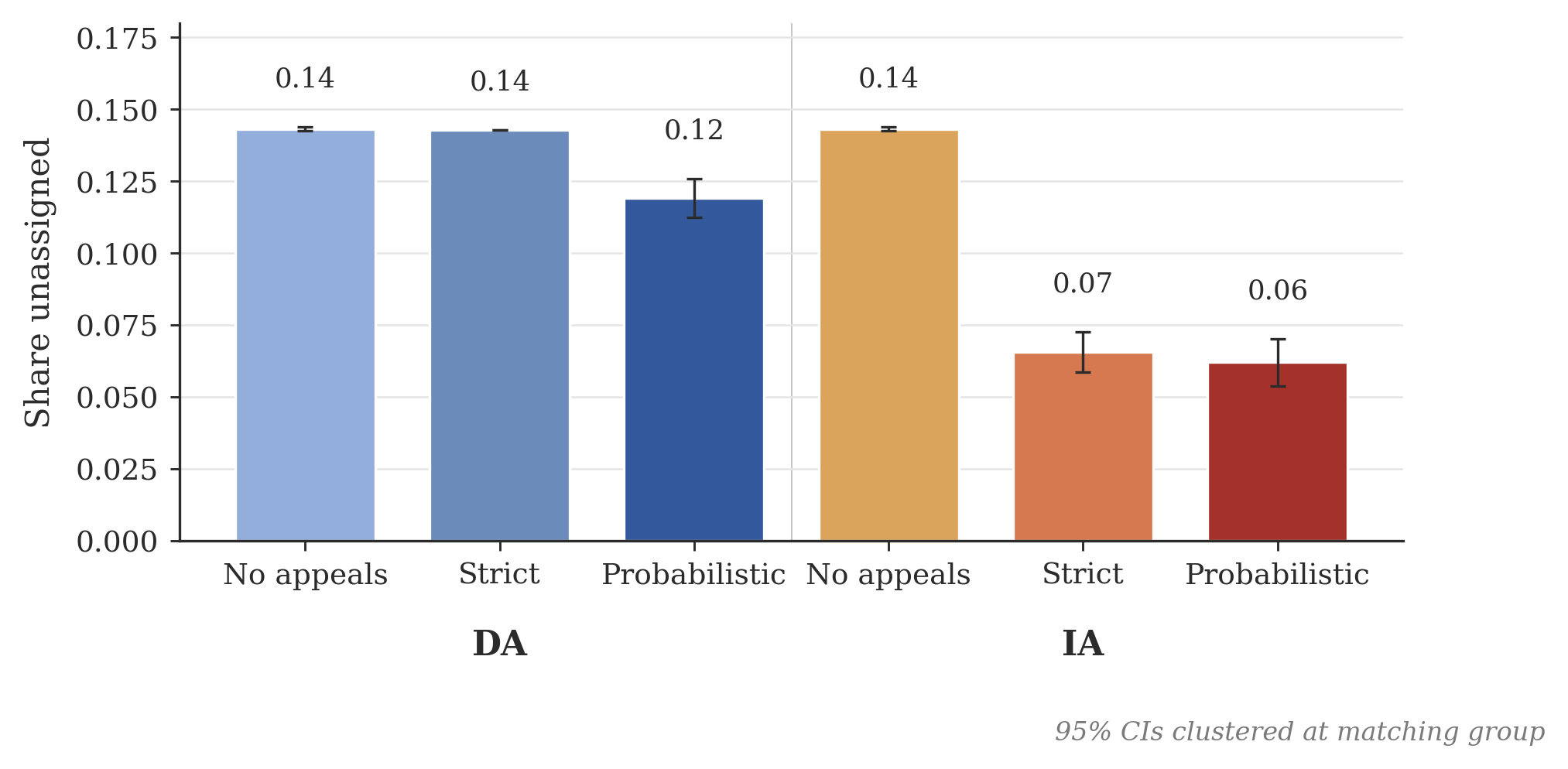}
    \caption{Share of participants unassigned by treatment.\\
    \footnotesize Whiskers are 95\% confidence intervals clustered at the matching-group level.}
    \label{fig:unassigned}
\end{figure}

Table~\ref{tab:points_appeals} repeats the regression specifications from the rank analysis using experimental points. All specifications use the 10,195-observation paper sample; the specification with personal controls uses the 10,185 observations for which those controls are complete. The estimates confirm the descriptive pattern. In the fully controlled specification, strict appeals raise IA welfare by 1.36 points and leave DA unchanged. Probabilistic appeals raise welfare by 1.46 points under IA and 0.67 points under DA. Strict and probabilistic appeals therefore widen IA's points advantage by 1.36 and 0.79 points, respectively.

\begin{table}[htbp]
\centering
\caption{Impact of appeals on experimental points}
\label{tab:points_appeals}
\begin{threeparttable}
\begin{tabular*}{\textwidth}{@{\extracolsep{\fill}} lcccc @{}}
\toprule
& \((1)\) & \((2)\) & \((3)\) & \((4)\) \\
\midrule
\multicolumn{5}{l}{\textbf{Panel A. IA treatments}} \\
IA strict & \(1.35^{***}\) & \(1.35^{***}\) & \(1.34^{***}\) & \(1.36^{***}\) \\
 & \((0.33)\) & \((0.33)\) & \((0.33)\) & \((0.32)\) \\
IA probabilistic & \(1.44^{***}\) & \(1.44^{***}\) & \(1.44^{***}\) & \(1.46^{***}\) \\
 & \((0.21)\) & \((0.21)\) & \((0.21)\) & \((0.21)\) \\
\addlinespace
\multicolumn{5}{l}{\textbf{Panel B. DA treatments}} \\
DA strict & \(0.01\) & \(0.01\) & \(0.01\) & \(0.00\) \\
 & \((0.06)\) & \((0.06)\) & \((0.06)\) & \((0.07)\) \\
DA probabilistic & \(0.68^{***}\) & \(0.68^{***}\) & \(0.68^{***}\) & \(0.67^{***}\) \\
 & \((0.07)\) & \((0.07)\) & \((0.07)\) & \((0.08)\) \\
\addlinespace
\multicolumn{5}{l}{\textbf{Panel C. Increase in IA's points advantage relative to no appeals}} \\
Strict appeals & \(1.34^{***}\) & \(1.34^{***}\) & \(1.33^{***}\) & \(1.36^{***}\) \\
 & \((0.33)\) & \((0.33)\) & \((0.33)\) & \((0.33)\) \\
Probabilistic appeals & \(0.76^{***}\) & \(0.76^{***}\) & \(0.76^{***}\) & \(0.79^{***}\) \\
 & \((0.22)\) & \((0.22)\) & \((0.22)\) & \((0.22)\) \\
\addlinespace
Round effects & No & Yes & Yes & Yes \\
Participant-type effects & No & No & Yes & Yes \\
Demographics + risk att. & No & No & No & Yes \\
\midrule
Observations & 10{,}195 & 10{,}195 & 10{,}195 & 10{,}185 \\
Matching groups & \(59\) & \(59\) & \(59\) & \(59\) \\
\bottomrule
\end{tabular*}
\begin{tablenotes}[flushleft]
\footnotesize
\item Notes. \(^{***}p<0.01\), \(^{**}p<0.05\), \(^{*}p<0.100\). Entries are adjusted differences in experimental points from linear regressions. Standard errors are clustered at the matching-group level. The baselines in Panels A and B are the corresponding mechanisms without appeals. Panel C reports \((\text{IA appeal}-\text{DA appeal})-(\text{IA no appeals}-\text{DA no appeals})\), so positive values indicate that appeals widen IA's points advantage. Demographics + risk att. are age, gender, and risk aversion.
\end{tablenotes}
\end{threeparttable}
\end{table}

\subsection{Strict blocking pairs}
\label{app:blocking-pairs}

The main text reports the share of stable market-rounds and the corresponding regression estimates. Figure~\ref{fig:blocking_pairs} provides the continuous counterpart by reporting the mean number of strict blocking pairs. No additional regression analysis is needed because Table~\ref{tab:stability_effects} already reports treatment effects on this outcome.

\begin{figure}[htbp]
    \centering
    \includegraphics[width=1.00\textwidth]{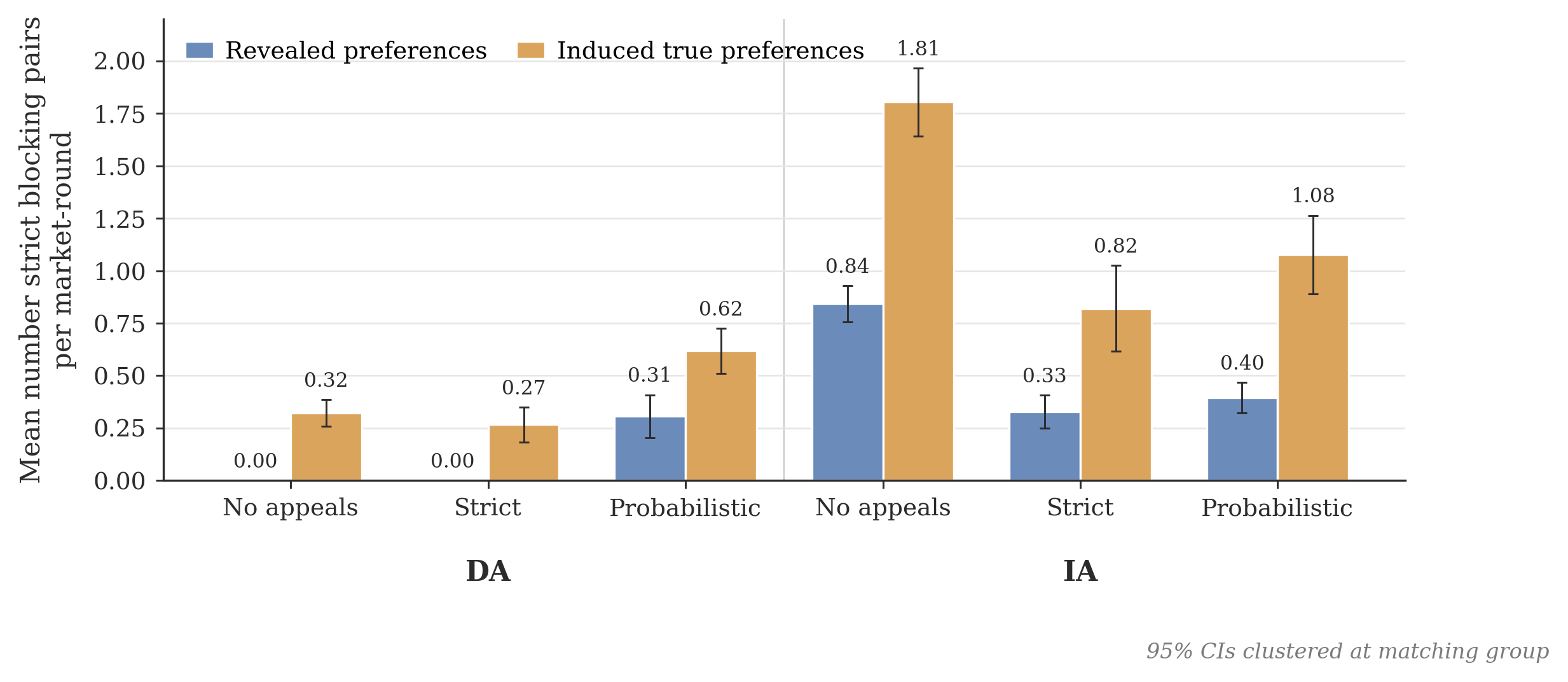}
    \caption{Strict blocking pairs by treatment.\\
    \footnotesize Bars report the mean number of strict blocking pairs per market-round, evaluated against submitted and induced true preferences. Equal-priority tie-break claims are excluded. Treatment means give equal weight to each independent matching group. Whiskers are 95\% confidence intervals across matching groups.}
    \label{fig:blocking_pairs}
\end{figure}

Strict appeals reduce IA's mean number of strict blocking pairs by more than half under both preference benchmarks, while leaving DA essentially unchanged. Probabilistic appeals also reduce blocking pairs under IA, but create priority violations after DA. This is the same qualitative pattern documented by the stable-share measure in the main text.

\newpage

\section{Full Appeal Take-Up.\hfill}
\label{app:full-takeup}
The theoretical analysis assumes that appeals are costless and therefore that
each student follows the weakly dominant appeal strategy
\[
a_i^{\mathrm{FT}}
=
\{s\in A_i(\hat\succ,M):s\succ_i M_i(\hat\succ)\}.
\]
We assess this assumption by replacing observed appeal choices with
\(a_i^{\mathrm{FT}}\), while holding submitted rank-order lists and
first-stage assignments fixed. Under the strict rule, every strict-priority
claim in this set succeeds. Under the probabilistic rule, strict-priority
claims succeed with certainty and every remaining claim succeeds independently
with probability \(0.1\). If several appeals succeed, the student receives her
most-preferred successful school according to her induced true preferences, as
in the model. We integrate exactly over all probabilistic appeal outcomes.
This exercise is a conditional theoretical benchmark rather than an
additional experimental treatment: in particular, it does not allow initial
reports to respond to automatic appeal filing.

\begin{table}[htbp]
	\centering
	\caption{Observed outcomes and the theoretical full-take-up benchmark}
	\label{tab:full-takeup}
	\begin{threeparttable}
		\begin{tabular*}{\textwidth}{@{\extracolsep{\fill}}lcccc@{}}
			\toprule
			& \multicolumn{2}{c}{DA} & \multicolumn{2}{c}{IA} \\
			\cmidrule(lr){2-3}\cmidrule(lr){4-5}
			& Observed & Full take-up & Observed & Full take-up \\
			\midrule
			\multicolumn{5}{l}{\textbf{Panel A. Strict appeals}} \\
			Average assigned rank
			& 2.739 & 2.739 & 2.059 & 2.010 \\
			Strict blocking pairs
			& 0.000 & 0.000 & 0.328 & 0.043 \\
			Procedurally stable share
			& 1.000 & 1.000 & 0.719 & 0.935 \\
			Strictly dominates truthful DA
			& 0.000 & 0.000 & 0.477 & 0.600 \\
			\addlinespace
			\multicolumn{5}{l}{\textbf{Panel B. Probabilistic appeals}} \\
			Average assigned rank
			& 2.606 & 2.551 & 2.041 & 1.940 \\
			Strict blocking pairs
			& 0.307 & 0.386 & 0.396 & 0.142 \\
			Procedurally stable share
			& 0.648 & 0.613 & 0.606 & 0.843 \\
			Strictly dominates truthful DA
			& 0.346 & 0.414 & 0.458 & 0.630 \\
			\bottomrule
		\end{tabular*}
		\begin{tablenotes}[flushleft]
			\footnotesize
			\item Notes. Lower average rank is better. Reports and first-stage assignments
			are held fixed. Full take-up replaces observed appeal choices with every
			appealable claim to a school genuinely preferred to the first-stage assignment
			and assigns the most-preferred successful school when several appeals are
			upheld. Average rank is measured at the participant-round level. Strict
			blocking pairs and Pareto dominance are averaged across complete seven-person
			market-rounds. Procedural stability is evaluated against submitted preferences
			and averaged first within matching group and then across matching groups, as
			in Figure~\ref{fig:stability}. Probabilistic outcomes are integrated exactly.
			Treatments without appeals are omitted because the counterfactual leaves them
			unchanged.
		\end{tablenotes}
	\end{threeparttable}
\end{table}

Full take-up sharpens the asymmetry under the strict rule. DA remains exactly
unchanged. Under IA, average assigned rank improves from \(2.059\) to \(2.010\),
strict blocking pairs fall from \(0.328\) to \(0.043\) per market-round, and
procedural stability rises from \(0.719\) to \(0.935\). The share of
market-rounds that strictly Pareto-dominate truthful DA rises from \(0.477\) to
\(0.600\). By construction, no appealable strict blocking pair involving a
school that the student genuinely prefers remains, as predicted by
Lemma~\ref{prop:pareto}. The outcome need not be fully procedurally stable
because that empirical measure is evaluated against submitted preferences and
also requires non-wastefulness.

The probabilistic benchmark reinforces the mechanism-specific comparison.
Full take-up further improves IA's average rank and procedural stability.
Under DA, it modestly improves average rank but increases strict blocking pairs
and reduces procedural stability because additional random admissions can
disturb an initially stable assignment and leave first-stage seats vacant.

\newpage

\section{Survey of Appeal Panel Members.\hfill}
\label{app:survey}

This appendix reports a survey conducted with members of school admission
appeal panels in England.\footnote{We thank Paul Stemp for his help.} The purpose of the survey was to understand how experienced panel members evaluate
different types of appeal claims, and we conducted it to understand how to best model appeal rules in our theoretical model that were consistent not only with the regulations, but with how appeals are upheld in practice.

The survey was implemented online using Qualtrics. We contacted all local
authorities in England and asked them to circulate the survey to individuals
involved in school admission appeals. Only a limited number of local authorities
agreed to circulate the survey or replied positively to our request. The local
authorities that circulated the survey were Merton, Barnsley, Barnet, Darlington,
Hartlepool, Kirklees, and Wokingham. Participation was voluntary
and anonymous, and the survey was approved through the ethics process at City,
University of London.

The survey asked respondents to evaluate a set of hypothetical appeal cases. In
each case, respondents were presented with a short description of a child who had
been refused admission to a school and whose parents had appealed the decision.
The cases varied the reason for the appeal. One case involved a clear priority or
admissions-rule violation. A second case involved a medical reason, where the
appealed school was close to a hospital and the child had a serious health
condition. A third case involved a comparatively weak reason, namely that the
school offered an after-school Italian course. A fourth case involved equal
priority and tie-breaking. A fifth case asked respondents how many additional
students they would be comfortable admitting when many rejected students appeal
to a school already at its published admissions number.

After each likelihood scenario, respondents were asked whether the appeal was
likely to be successful. The response categories were \emph{Very unlikely},
\emph{Unlikely}, \emph{Hard to say}, \emph{Likely}, and \emph{Very likely}. These questions appeared in random order. 
Respondents were also invited to explain their answer in an open text box. The
survey also included two broader questions. First, respondents were asked how
often appeal panel decisions are unanimous. Second, they were asked an
open-ended question about the possible reasons why an appeal may be successful.
These questions were included to understand whether panel members view appeal
decisions as deterministic or as involving substantial discretion and case-by-case
judgement.

Table~\ref{tab:panel-survey-likelihoods} reports the distribution of responses
across the four main hypothetical cases. Missing responses are reported
separately.

\begin{table}[htbp]
	\centering
	\caption{Perceived likelihood that an appeal succeeds}
	\label{tab:panel-survey-likelihoods}
	\resizebox{\textwidth}{!}{\begin{tabular}{lcccc}
		\toprule
		& Priority violation & Medical reason & Trivial reason & Tie/equal-priority case \\
		\midrule
		Very unlikely & 2 & 1 & 28 & 9 \\
		Unlikely      & 3 & 7 & 32 & 16 \\
		Hard to say   & 37 & 46 & 20 & 42 \\
		Likely        & 16 & 20 & 2 & 11 \\
		Very likely   & 17 & 2 & 1 & 4 \\
		No answer     & 19 & 18 & 11 & 12 \\
		\midrule
		Total         & 94 & 94 & 94 & 94 \\
		\bottomrule
	\end{tabular}}
\end{table}

Several patterns are worth noting. First, respondents distinguished sharply
between different reasons for appeal. The weak or idiosyncratic reason was
generally regarded as unlikely to succeed: 60 out of 94 respondents chose either
\emph{Very unlikely} or \emph{Unlikely}. By contrast, the priority-violation case
attracted substantially more favourable responses. Only 5 respondents chose
\emph{Very unlikely} or \emph{Unlikely}, while 33 chose \emph{Likely} or
\emph{Very likely}. Equivalently, among respondents who gave a directional answer
rather than selecting \emph{Hard to say}, 33 out of 38 classified the
priority-violation appeal as likely or very likely to succeed. This supports our
focus on the strict appeal rule $r_1$ in the main text: experienced panel members
broadly agree that appeals are especially compelling when the initial rejection
reflects a priority or admissions-rule violation.

Second, the medical-reason case generated considerable uncertainty. Although 22
respondents classified the appeal as \emph{Likely} or \emph{Very likely}, 46
selected \emph{Hard to say}. This is consistent with the idea that some appeals are
not resolved by a simple priority comparison. Instead, panel members may need to
weigh the strength of the family's circumstances against the prejudice to the
school from admitting an additional pupil. This motivates the probabilistic appeal
rule in the model: some claims may be institutionally admissible but difficult to
predict ex ante.

Third, the tie or equal-priority case also generated substantial uncertainty. A
large fraction of respondents selected \emph{Hard to say}, while the remaining
responses were spread across both positive and negative categories. This suggests
that cases involving equal priority, tie-breaking, or borderline admissibility may
be treated differently across panels or local contexts. This evidence motivates
our distinction between strict appeals, weak appeals, and probabilistic appeals.

The survey also asked respondents how often panel decisions are unanimous. The
answers are reported in Table~\ref{tab:panel-survey-unanimity}. Most respondents
reported that unanimous decisions occur often, but not always. This is further
evidence that appeal decisions involve judgement and that similar cases need not
be assessed identically by all panel members.

\begin{table}[htbp]
	\centering
	\caption{Frequency of unanimous panel decisions}
	\label{tab:panel-survey-unanimity}
\begin{tabular}{lc}
		\toprule
		Response & Number of respondents \\
		\midrule
		Sometimes & 17 \\
		Often     & 60 \\
		Always    & 4 \\
		No answer & 13 \\
		\midrule
		Total     & 94 \\
		\bottomrule
	\end{tabular}
\end{table}

The capacity question provides similar evidence of discretion. Respondents were
asked to consider a secondary school with a published admissions number of 100
that had already accepted 100 students, while another 100 rejected students were
appealing. They were then asked for the maximum number of additional students
they would be comfortable admitting. Responses varied substantially: many
respondents answered zero, while others were willing to admit a positive number
of additional students. This reinforces the view that appeal panels may differ in
how they weigh the prejudice to the school against the claims of appealing
families.

The open-ended responses provide additional context. Respondents listed several
reasons why appeals may be successful, including mistakes in applying
oversubscription criteria, admissions arrangements that were not lawful or not
correctly applied, unreasonable decisions, medical or social needs, special
educational needs, sibling or twin considerations, recent changes in family
circumstances, the lack of a reasonable alternative school, and cases where the
prejudice to the child or family was judged to outweigh the prejudice to the
school. These answers are consistent with the institutional description in the main
text: some appeals operate as corrections of procedural or priority errors, while
others involve a balancing exercise.

The survey has limitations. Although we contacted all local authorities
in England, only a small subset agreed to circulate the survey or replied
positively. The resulting coverage is therefore limited, and the respondents
should not be interpreted as representative of all admission appeal panel members
in England. Respondents also self-selected into participation, and we do not
observe the full population of panel members who received the survey invitation.
Finally, the hypothetical cases are necessarily stylized and cannot capture all
legal and factual details that would be available in an actual appeal hearing. For
these reasons, we use the survey only as qualitative and descriptive support for
the modeling assumptions, not as a basis for estimating appeal success
probabilities.

Despite these limitations, the survey is useful for the purposes of this paper. It
confirms that experienced panel members distinguish sharply between different
grounds for appeal; that priority or procedural violations are viewed as especially
important; and that medical, welfare, tie-breaking, and balancing cases generate
substantial uncertainty. This supports the modelling choice to study more than one
appeal rule. The strict rule captures appeals that correct priority or procedural
violations. The weak rule captures cases in which equal-priority or tie-breaking
claims may be recognised. The probabilistic rule captures the discretionary and
heterogeneous component of appeal adjudication.

\subsection{Survey Verbatim Text}\label{app:panel-survey-instrument}

\paragraph{Consent screen.}
You are invited to take part in a research survey being conducted by Dr Claudia
Cerrone, Dr Yoan Hermstrüwer and Dr Josuè Ortega.

The purpose of this survey is to understand school admissions appeals in England.
The survey has received ethics approval from City, University of London. Your
participation in the survey is voluntary and the information you provide is
anonymous.

The data collected from this survey may be used for scientific publication. By
continuing to the questions, you consent to participate in this survey.

\emph{I continue to the questions}

\paragraph{Question 1.}
Livia and Matty have applied for the same primary school (non-infant class) and
are both being evaluated on the basis of the admission criterion of distance to the
school. Livia lives 975 meters from the school and is rejected, while Matty lives
950 meters from the school and is admitted.

Livia's parents have decided to appeal the school's decision to reject Livia. They
argue that Livia should be admitted to that school as it offers an after school
Italian course and they are very keen on Livia learning Italian. Note that this
appeal is for a place in Year 5 and therefore not covered by infant class size
legislation.

Do you think that this appeal is likely to be successful?

\emph{Very unlikely; Unlikely; Hard to say; Likely; Very likely}

If you want to explain your decision, please use the box below.

\paragraph{Question 2.}
Charlie and Dora have applied for the same primary school (non-infant class) and
are both being evaluated on the basis of the admission criterion of distance to the
school. Charlie lives 975 meters from the school and is rejected, while Dora lives
950 meters from the school and is admitted.

Charlie's parents have decided to appeal the school's decision to reject Charlie.
They argue that Charlie should be admitted to that school as it is close to a
hospital and Charlie has a life-threatening heart defect. Note that this appeal is
for a place in Year 5 and therefore not covered by infant class size legislation.

Do you think that this appeal is likely to be successful?

\emph{Very unlikely; Unlikely; Hard to say; Likely; Very likely}

If you want to explain your decision, please use the box below.

\paragraph{Question 3.}
Appeals are usually successful if school admission arrangements are unlawful or
have not been applied correctly. Apart from this, what are possible reasons why
an appeal might be successful? Give as many reasons as you want, and please
separate them by a comma.

\paragraph{Question 4.}
Based on your experience, how often are panels' decisions unanimous?

\emph{Never; Rarely; Sometimes; Often; Always}

\paragraph{Question 5.}
A selective secondary school admits students based on their grades. The school
admits all students with an average grade above 80 and reserves two spots for
students with an average grade of exactly 80. Three applicants had an average
grade of 80, so a tie-breaker was used to allocate the two available places.

The parents of the student who was not selected are appealing the school's
decision to reject their child, arguing that it is unfair for their child to be
rejected purely due to random chance.

Is their appeal likely to succeed?

\emph{Very unlikely; Unlikely; Hard to say; Likely; Very likely}

If you want to explain your decision, please use the box below.

\paragraph{Question 6.}
Ana and Bob have applied for the same primary school and are both being
evaluated on the basis of the admission criterion of distance to the school. Ana
lives 925 meters from the school and is rejected, while Bob lives 1,050 meters
from the school and is admitted.

Ana's parents have decided to appeal the school's decision to reject Ana. They
argue that Ana should be admitted to the school as the admissions criteria have
not been applied correctly. Note that this appeal is for a Reception place under
infant class size legislation.

Do you think that this appeal is likely to be successful?

\emph{Very unlikely; Unlikely; Hard to say; Likely; Very likely}

If you want to explain your decision, please use the box below.

\paragraph{Question 7.}
A secondary school has a Published Admissions Number of 100, which was
calculated in compliance with admissions law. While 100 students have been
accepted, another 100 were rejected. All the parents of the rejected children are
appealing, arguing that their children would benefit from the school's unique
mathematics program. The school argues that it has no extra resources to
accommodate more students, and that each additional student admitted would lead
to larger classes, and consequently less resources per pupil.

You must decide how many appeals will be successful.

What is the maximum number of additional students you would be comfortable
admitting to the school?

\begin{center}
	\emph{End of survey.}
\end{center}

\end{document}